\documentclass[12pt]{article}

\usepackage{geometry, fullpage, authblk}
\usepackage{xr-hyper}
\usepackage{times}
\usepackage{hyperref}[]
\hypersetup{
	colorlinks=true,
	linkcolor=darkblue,
	filecolor=magenta,      
	urlcolor=darkblue,
	citecolor=blue,
}

\usepackage[normalem]{ulem}

\usepackage{url}

\usepackage{bigints}
\usepackage{amsfonts,amsmath,amssymb,amsthm,bbm}
\usepackage{wrapfig}
\usepackage{verbatim,float,url}
\usepackage{graphicx}
\usepackage{sidecap}
\usepackage{subfigure}
\usepackage{subcaption}
\usepackage[skip=0pt]{caption}
\usepackage{natbib}
\usepackage[english]{babel}
\usepackage{cancel}

\usepackage[dvipsnames]{xcolor}
\definecolor{darkblue}{RGB}{10, 10, 200}

\usepackage{cleveref}
\usepackage{enumerate}
\usepackage{multirow}
\usepackage{amsmath}

\usepackage{collcell}
\makeatletter
\newcolumntype{G}{>{\collectcell\@gobble}c<{\endcollectcell}@{}}
\makeatother

\usepackage{bm}

\usepackage{mathtools}

\newcommand{\bbe}{\bm{\beta}}
\newcommand{\bx}{\mathbf{x}}

\newtheorem{theorem}{Theorem}
\newtheorem{lemma}{Lemma}

\newtheorem{proposition}{Proposition}

\usepackage{mdwlist}

\theoremstyle{definition}
\newtheorem{definition}{Definition}

\usepackage{ctable}
\usepackage{colortbl}

\usepackage{mathtools}

\newcommand{\iidsim}{\overset{iid}{\sim}} % iid simulated
\renewcommand{\&}{and}
\usepackage{ulem}

\usepackage[ruled,vlined]{algorithm2e}
\SetKwInput{kwInit}{Initialize}

\DeclareMathOperator*{\argmax}{arg\,max}

\theoremstyle{plain}
\newtheorem{assumption}{Assumption}
\allowdisplaybreaks
\usepackage{bm}

\newcommand{\algoname}{{CPVS}}
\newcommand{\algonamesp}{{CPVS }}

\title{Computationally efficient segmentation for non-stationary time series with oscillatory patterns}

\newcommand{\abs}{
We propose a novel approach for change-point detection and parameter learning in multivariate non-stationary time series exhibiting oscillatory behaviour. We approximate the process through a piecewise function defined by a sum of sinusoidal functions with unknown frequencies and amplitudes plus noise. The inference for this model is non-trivial. However, discretising the parameter space allows us to recast this complex estimation problem into a more tractable linear model, where the covariates are Fourier basis functions. Then, any change-point detection algorithms for segmentation can be used. The advantage of our proposal is that it bypasses the need for trans-dimensional Markov chain Monte Carlo algorithms used by state-of-the-art methods. Through simulations, we demonstrate that our method is significantly faster than existing approaches while maintaining comparable numerical accuracy. We also provide high probability bounds on the change-point localization error. We apply our methodology to climate and EEG sleep data.

\vspace{0.1in}

\noindent\textbf{Keywords:} Bayesian variable selection, spectral analysis, change point, EEG sleep data, localization rate. 

}

\author[1]{Nicolas Bianco
}
\author[2,3]{Lorenzo Cappello}
\affil[1]{Scientific Computing Center, Karlsruhe Institute of Technology, Germany}
\affil[2]{Department of Economics and Business, Universitat Pompeu Fabra, Spain}
\affil[3]{Data Science Center, Barcelona School of Economics, Spain}

\date{\today }

\begin{document}
	
\maketitle
\thispagestyle{empty}

\centerline{\bf Abstract}
\medskip
\abs
\normalsize
\setlength{\parskip}{.2cm }
\setlength{\parindent}{0.5cm}

\clearpage
\pagenumbering{arabic}
	
\section{Introduction}
\label{sec:intro}

Many dynamic processes in environmental science and neuroscience exhibit cyclical patterns. In these applications, knowing the frequency of the cycles is relevant to the understanding of the underlying phenomena. Traditional statistical techniques assume that the set of frequencies is constant over time. However, the availability of large data sets that characterize modern applications is debunking this assumption. An example is given by electroencephalogram (EEG) sleep data, with events such as spindles \citep{parekh2017multichannel} and micro-arousals \citep{qian2021review} causing sudden changes in the EEG signals recorded.

The simplest departure from i.i.d. or stationarity assumptions is given by change point models, which define a partition of the data into shorter segments. Such partition is useful statistically -- within each segment one can apply traditional estimators -- but also scientifically -- the segmentation is often meaningful to the understanding of the phenomena that are studied. For example, researchers are interested in segmenting EEG data because sleep events are related to neuropsychiatric conditions \citep{manoach2019abnormal}. For these reasons, research in change point models had a new impetus; see \cite{truong2020selective} for a review.

Most works either assume a simple parametric structure within each segment --  \textit{e.g.}, piecewise constant mean \citep{fric14,fryzlewicz2014wild}, piecewise constant regression parameters \citep{bai1998estimating,baranowski2019narrowest} -- or are fully non-parametric \citep{matteson2014nonparametric,madrid2021optimal}. 
The shared thread is that they all admit a closed-form ``score" within each segment -- \textit{e.g.}, the likelihood evaluated at the MLE -- which is then used downstream for the segmentation. 
The problem considered here assumes a parametric model that does not admit a closed-form estimator making  within-segment inference non-trivial.
Statistical and computational issues are exacerbated working with multivariate series,  a standard setting in modern applications; for example, EEG signals are recorded at multiple locations. In this paper, we introduce a computationally efficient methodology to segment multivariate non-stationary time series exhibiting cyclical behaviour. 

We model the oscillatory pattern within each segment as a sum of an unknown number of sinusoidal functions, with different frequencies and intensities, plus a noise process. The model appeared in the literature to analyse the cyclical and oscillatory signals \citep{andrieu1999joint,shumway2019time,hadj2020bayesian,wu2024frequency}, especially in the context of physiological time series, and relies on a sparsity assumption that just a few of the frequencies of a periodogram are relevant to describe the series \citep{shumway2019time}. Statistical inference for this model is challenging because the model dimension needs to be selected, and then one needs to estimate model parameters conditionally on a given dimension.

State-of-the-art segmentation approaches for this class of model are Bayesian \citep[e.g.,][]{rosen2012adaptspec,hadj2020bayesian}. This is somewhat against the trend in the change point detection literature, where Bayesian approaches are not widely used because their computational cost is magnitudes of order larger than competing algorithms \citep{fearnhead2006exact,liu20,cappello2022bayesian}. 
A possible explanation for the success of Bayesian methods is their ability to incorporate within-segment inference uncertainty.
Indeed, a frequentist approach would estimate within-segment model dimension, for example maximizing an information criterion, and then detect change points conditionally on such estimates. As a result, the accuracy of the segmentation might be affected by model selection mistakes within each segment. On the other hand, Bayesian methods avoid this issue, treating model dimension as random. \cite{hadj2020bayesian} provide evidence of that through numerical studies. 

In addition to larger computational cost, a further drawback of Bayesian approaches stems from the reliance on complex reversible jump Markov chain Monte Carlo (RJ-MCMC) \citep{green1995reversible}. RJ-MCMC must be carefully initialized and designed to enable good mixing \citep{brooks2003efficient}. Moreover, assessing the convergence of these algorithms is difficult. For example, \cite{fearnhead2006exact} points out that \cite{green1995reversible} RJ-MCMC has not converged in the coal mining change point application. Here, the model is arguably even more complex. \cite{hadj2020bayesian}, for example, use two trans-dimensional components: a first one to estimate the model dimensionality within-segment \citep{andrieu1999joint},
and a second one to estimate the size of the partition in the change point model \citep{green1995reversible}.  \cite{hadj2020bayesian} highlight all these difficulties in the discussion.

Our proposal relies on two ideas. The first one is that it is possible to develop a hybrid methodology that preserves the uncertainty quantification within segments provided by Bayesian methods but uses computationally efficient frequentist procedures to detect change points. The reason that makes this hybrid approach attractive is that it preserves the ability to account for the model selection uncertainty while avoiding the computational burden of Bayesian change point detection algorithms. We retain the uncertainty quantification of the parameters within a segment, but lose that on the change point locations. However, we avoid RJ-MCMC altogether resulting in an algorithm order of magnitudes faster, less prone to convergence issues, hence more suitable for large multivariate data sets.

The second idea is that, through a discretisation of the space of frequencies, we can recast the within-segment inference (model selection and parameter inference) as a variable selection problem for a linear regression with a design matrix whose columns are Fourier basis. The discretisation of the frequency parameter is justified by classical methods in oscillatory frequency detection which evaluate the periodogram at the discrete set of canonical frequencies $\{\omega_{jn}=2\pi j/n:j=1,\ldots,\lfloor n/2\rfloor\}$ \citep{shumway2019time}. By discretizing and reformulating the mean signal as a linear model, frequencies are no longer parameters to be estimated in a non-linear model, but covariates to be selected in a simpler, linear one.
Then, any method that allows for joint variable selection and regression coefficient estimation can be used. While a potential drawback is the sensitivity the construction of the grid, which we will discuss in more details, recasting the problem as a linear regression speeds up computation massively since we can leverage decades of research on scalable variable selection, while we show that accuracy is not particularly affected. In our proposal, we use Bayesian variable selection within each segment to account for model dimension uncertainty, as explained in the previous paragraph. 

A further benefit of our proposal is that it is amenable to theoretical analysis. To our knowledge, the localization rate of change point procedures in this class of models - changes in frequencies, intensities, and dimensionality - has not been studied before. We establish high probability bounds on the estimation error of our method. A feature of our result is that we account for both parameter and model selection uncertainty within each segment. Having rewritten the problem as a variable selection problem is crucial in the proof technique. This result is new in the change point detection literature and adds to the literature on theoretical guarantees for Bayesian change point detection methods, currently limited \citep{liu20,cappello2023bayesian,cappello2022bayesian}. 

In terms of related works, the closest are \cite{hadj2020bayesian} (previously discussed) and \cite{wu2024frequency}, who considered the same mean signal but with a non-stationary noise process. Like our approach, they employed a discretized frequency space and introduced the notion of ``dense" periodogram, demonstrating that the parametric estimation rate for frequency detection \citep{genton2007statistical} can be achieved through a finer frequency grid. In addition, they proposed two block bootstrap schemes suitable for frequency estimation and change point detection tailored in the presence of non-stationary noise. Neither of these works treated the multivariate case.

The paper builds on two streams of literature, that on locally stationary time series \citep[see][for a formal definition]{dahlhaus1997fitting} and the one on change point detection. Remaining in the frequency domain, \cite{rosen2009local} model nonparametrically a time-varying log spectral density via a tailored Bayesian mixture of splines. A limitation of this method is that it requires a prior specification of the partition, a drawback solved by the AdaptSPEC method \citep{rosen2012adaptspec}; \cite{bertolacci2022adaptspec} generalize AdaptSPEC to a panel data.  
An alternative way to model nonstationary periodic times is to work in the time domain. Here, the main focus has been on mixtures of autoregressive processes \citep{davis2006structural,lau2008bayesian,wood2011bayesian}. In particular, \cite{davis2006structural} model a non-stationary time series as a composition of stationary autoregressive processes (AR), where the number and location of breakpoints are unknown. The problem of finding the best partition is treated as a statistical model selection problem. The proposed algorithm uses a genetic algorithm to solve the optimization problem, where the target function is the Minimum Description Length (MDL) criteria of \cite{rissanen1978modeling}. Wavelet-based methods operate at intermediate time-frequency resolutions and have been widely used to segment non-stationary time series. Recent contributions include \cite{cho2012multiscale,killick2013wave,preuss2015detection,korkas2017multiple}. However, most of these works focus on detecting changes in the second-order structure of the series, rather than directly modelling the oscillatory behaviour.

The second literature is on change point detection, an area that dates back to \cite{page1954continuous}. Our work belongs to the offline change point detection literature, where we assume that we first observe all the data, and then segment the series. A class of exact segmentation approaches relies on dynamic programming \citep{auger1989algorithms}, which can be slow. However, the development of efficient algorithms is an active area of research \citep{killick2012optimal,fric14,maidstone2017optimal}. An alternative class of methods performs approximate inference through greedy search algorithms, for example, Binary Segmentation (BS) \citep{vostrikova1981detecting}. Such algorithms are generally faster than dynamic programming, but a vanilla version of the algorithm can be statistically suboptimal. There are solutions to this, such as Wild BS \citep{fryzlewicz2014wild}, Seeded BS \citep{kovacs2023seeded}, and narrowest-over-threshold (NOT) \citep{baranowski2019narrowest}. Despite being flexible to other choices, we will argue that optimistic search (OS) \citep{kovacs2020optimistic} is particularly suitable for our model because it allows us to minimize the number of within-segment inferences. Relevant to this work is also the literature on theoretical guarantees for the change point method; see references in \cite{yu2020review}. From the Bayesian perspective, popular methods to detect multiple change points are based on product partition models \citep{bar92}. An active area of research is speeding up computations, mostly relying on alternatives to MCMC \citep{chi98,fearnhead2006exact,rig12,cappello2022bayesian,cappello2023bayesian,berlind2025}.

\section{Methodology}
\label{sec:method}

Many time series show periodic behaviour, and a primary role of spectral analysis is to estimate these underlying periodicities \citep{koopmans1995spectral}. Basic results in Fourier analysis state that it is possibly to approximate arbitrary well a function in a finite interval by a weighted sum of sine and cosine functions with different frequencies. This idea is exploited by several authors who show that in practice only a few sinusoidal functions are sufficient to represent well the cyclical behaviour in many real data problems \citep[see][among the others]{andrieu1999joint,shumway2019time,hadj2020bayesian,wu2024frequency}. We start considering a framework with a panel of $d$ independent series that share the same change points.
Assume that we observe $T$ realizations $\mathbf{y}_t=(y_{1t},\ldots,y_{dt})^\intercal$ of $d$-dimensional random vector $\bm{Y}_t \sim N_d(\bm{\mu}_t,\bm{\Sigma}_t) $ where $\bm{\Sigma}_t=\mathrm{diag}(\sigma^2_{1t},\ldots,\sigma^2_{dt})$, and for $i=1,\ldots,d$:
\begin{align}\label{eq:model}
	\mu_{it}= \sum_{l=1}^{L_{ij}} \left(\beta_{ijl}^{(1)} \sin (2\pi \omega_{ijl}t)+\beta_{ijl}^{(2)} \cos (2\pi \omega_{ijl}t)\right), \quad \sigma_{it}^2=\sigma_{ij}^2,  \qquad t_{j-1} < t \leq t_j,
\end{align}

where $\{t_1,\ldots,t_{m}\}$ is a set of $m$ time instances that partition $\{1,\ldots,T\}$ in $m+1$ segments, and conventionally $t_0=0$ and $t_{m+1}=T$. For a generic series $i$, we refer to $L_{ij}$, the number of sinusoidal functions used in segment $j$, as the within-segment model dimensionality, to $\bm{\omega}_{ij}=(\omega_{ij1},\ldots, \omega_{ijL_{ij}})$, with $\omega_{ijl}\in(0,1/2)$ as the vector of frequencies, $\bm{\beta}_{ij}^{(1)}$ and $\bm{\beta}_{ij}^{(2)}$  as the vectors of intensities associated to a given segment, and denote $\bm{\theta}_{ij}=(L_{ij},\bm{\omega}_{ij},\bm{\beta}_{ij}^{(1)},\bm{\beta}_{ij}^{(2)})$, and $\bm{\theta}_{j}=(\bm{\theta}_{ij})_{i=1:d} \in \Theta_j$.  Following \cite{hadj2020bayesian} and \cite{wu2024frequency}, the mean function in \eqref{eq:model} is not continuous and follows a piecewise structure. We discuss continuous mean signals in Sections~\ref{sec:sim} and \ref{sec:cont}.

In parametric change point detection, one builds on the assumption that $\mathbf{Y}_t \iidsim \mathbb{P}_{\bm{\theta}_j}$ for $t\in (t_{j-1},t_j]$, and parametrize the full series by $( \{t_j\}_{1:m},(\bm{\theta}_j)_{1:m+1} )$. The log-likelihood can be expressed as:
\begin{equation}\
	\log p (\mathbf{y}_{1:T}| \{t_j\}_{1:m},(\bm{\theta}_j)_{1:m+1})= \sum_{j=1}^{m+1} \sum_{t=t_{j-1}+1}^{t_j} \log p(\mathbf{y}_t|\bm{\theta}_j).
\end{equation}	
Standard methodologies typically estimate the number of change points $m$ and their locations $\{t_1,\ldots,t_{m}\}$ maximizing a penalized parametric log-likelihood
\begin{equation}\label{eq:cp}
	\hat{t}_1,\ldots, \hat{t}_{\hat{m}} = \argmax_{\mathcal{P} \in \mathcal{T}} \max_{(\bm{\theta}_j)_{1:|\mathcal{P}|+1}} \log p (\mathbf{y}_{1:T}| \mathcal{P},(\bm{\theta}_j)_{1:|\mathcal{P}|+1}) + \lambda | \mathcal{P}|,
\end{equation}
where $\mathcal{T}$ is the space of partitions of $\{1,\ldots,T\}$, $\mathcal{P}$ is an element of $\mathcal{T}$, and $|\mathcal{P}|$ its cardinality. 

Optimization $\eqref{eq:cp}$ is appealing because the MLE of the parametric models considered in the standard change point detection problem is usually available in closed form; for example, the mean of a Gaussian \citep{yao1988estimating,fryzlewicz2014wild}, linear regression coefficients \citep{bai1998estimating,baranowski2019narrowest}, variance of a Gaussian \citep{chen1997testing,killick2012optimal}, and exponential family natural parameters \citep{fric14}. This means that $ \max_{(\bm{\theta}_j)_{1:|\mathcal{P}|+1}} \log p(\mathbf{y}_{1:T}| \mathcal{P},(\bm{\theta}_j)_{1:|\mathcal{P}|+1})$ can be written explicitly, and the primary challenge in solving \eqref{eq:cp} is the combinatorial optimization over $\mathcal{T}$.  Here, the setting is more challenging: no closed form estimator is available for parameters in \eqref{eq:model} and different MLE estimates are necessary as within-segment model dimension varies. In the next subsection, we propose a simple way to tackle these challenges. 

In addition, solving \eqref{eq:cp} is only one of the possible ways to estimate the change points. Another popular idea is to substitute the log-likelihood evaluated at the MLE with a marginal likelihood:
\begin{equation}\label{eq:marginal}
	\hat{t}_1,\ldots, \hat{t}_{\hat{m}} = \argmax_{\mathcal{P} \in \mathcal{T}}  p(\mathbf{y}_{1:T}| \mathcal{P})+ \lambda | \mathcal{P}|,
\end{equation}
where, given a prior distribution on $(\bm{\theta}_k)_{1:|\mathcal{P}|+1}$,
\begin{equation*}
p (\mathbf{y}_{1:T}| \mathcal{P}) = \int_{\Theta_1\ldots \Theta_{|\mathcal{P}|+1}} p(\mathbf{y}_{1:T}| \mathcal{P},(\bm{\theta}_k)_{1:|\mathcal{P}|+1}) \,d\, p((\bm{\theta}_k)_{1:|\mathcal{P}|+1}),    
\end{equation*}
where with an abuse of notation we use the integral both for continuous (frequencies, intensities) and discrete (model dimension) parameters. The use of the marginal likelihood has motivated the development of information criteria specific to change point detection, for example, the Modified BIC \citep{zhang2007modified}, and it is common in the Bayesian literature in the single change point case \citep{smi75,cappello2022bayesian}. In the problem we are considering, the advantage of the marginal likelihood is that it explicitly accounts for model dimension uncertainty, while $ \max_{(\bm{\theta}_j)_{1:|\mathcal{P}|+1}} \log p(\mathbf{y}_{1:T}| \mathcal{P},(\bm{\theta}_j)_{1:|\mathcal{P}|+1})$ fixes the size of the model for each segment.

In the rest of the section, we discuss how to infer the model's parameters within a segment (Subsection~\ref{sec:within}), how to solve the optimization problem \eqref{eq:marginal} (Subsection~\ref{sec:cp}), and study theoretically our proposal in the single change point case (Subsection~\ref{subsec:theory}). Inference on the within segment parameters is described for a generic $i=1,\ldots,d$ series since, given any interval, the series are modelled independently. On the other hand, the gain function to detect the change-points presented in Subsection~\ref{sec:cp} is defined for the joint vector of time series.

\subsection{ Within segment inference}\label{sec:within}
Assume to be in a generic segment $I_j=(t'_{j-1},t'_j]$ with $n_j=|I_j|$, with $t'_{j-1}, t'_j$ denoting the time instances that define a candidate region with constant parameters. The goal is to estimate $L_{ij}$, $\bm{\omega}_{ij}$, $\bm{\beta}_{ij}$, and $\sigma^2_j$ given vector of observations $\mathbf{y}_{iI_j}=(y_{it'_{j-1}},\ldots, y_{it'_{j}-1})$, for all $i=1,\ldots,d$. The main difficulties arise in the inference for $L_{ij}$ and $\bm{\omega}_{ij}$ since there is no closed-form estimator for $\omega_{ijl}$, and, similarly, no conjugate prior is available for Bayesian analysis. Here, we show that a discretisation of the space of frequencies allows us to rewrite inference for model~\eqref{eq:model} as a linear regression problem. This reformulation enables the use of inference methods for sparse high-dimensional linear models to jointly estimate all model parameters, as we detail below.
The key idea of the subsection is that inference for this regression model can be order of magnitude faster than trying to do inference in \eqref{eq:model} without sacrificing accuracy, as we will illustrate in Section~\ref{sec:sim}.
While the discretisation may introduce approximation error—since the true frequencies might not lie exactly on the chosen grid—\cite{wu2024frequency} demonstrate that selecting a sufficiently dense grid can mitigate this issue.

Assume that $\omega_{ijl} \in \Omega:=\{w_{1},\ldots,w_{p}\}$,  $0<w_1<\dots<w_p<1/2$, let $\mathbf{X}^{(1)}$ and  $\mathbf{X}^{(2)}$ be two $T \times p$ matrices whose columns are Fourier basis, \textit{i.e.}, $[\mathbf{X}^{(1)}]_{tk}=\sin(2\pi w_{k} t)$ and $[\mathbf{X}^{(2)}]_{tk}=\cos(2\pi w_{k}t)$, and $\mathbf{X}_{I_j}^{(1)}$ and  $\mathbf{X}_{I_j}^{(2)}$ the $n_j \times p$ sub-matrices corresponding to rows $I_j$. Then, we can rewrite the $i$-th component in \eqref{eq:model} as a standard linear regression: 
\begin{equation}\label{eq:varsel}
	\mathbf{y}_{iI_j} = \mathbf{X}^{(1)}_{I_j} \boldsymbol{b}_{ij}^{(1)}+\mathbf{X}^{(2)}_{I_j} \boldsymbol{b}_{ij}^{(2)} +\sigma_{ij}\bm{\varepsilon}_{iI_j}, \qquad \bm{\varepsilon}_{iI_j}\sim N(0,\mathbf{I}_{n_j}),
\end{equation}
where $\boldsymbol{b}_{ij}^{(1)}$ and $\boldsymbol{b}_{ij}^{(2)}$ are sparse vectors of coefficients, with $L_{ij}$ non-zero entries satisfying:
\begin{equation}\label{eq:sparsity_condition}
	\left\{k: b^{(1)}_{ijk} \neq 0\right\}= \left\{k: b^{(2)}_{ijk} \neq 0\right\}.
\end{equation}
The non-zero entries of $\boldsymbol{b}_{ij}^{(1)}$ and $\boldsymbol{b}_{ij}^{(2)}$ correspond to the intensities $\bm{\beta}^{(1)}_{ij}$ and $\bm{\beta}^{(2)}_{ij}$ (up to a permutation). Applying variable selection to \eqref{eq:varsel}, the equivalence between the two models is evident:
$\bm{\beta}_{ij}$ are regression coefficients, $\bm{\omega}_{ij}$ are obtained by selecting which features to include in the regression, and the dimensionality $L_{ij}$ is the number of features included. The benefit of using this second model is also clear as we have many computationally efficient algorithms to do variable selection. The only peculiar feature is the constraint \eqref{eq:sparsity_condition}, which resembles a group LASSO-type constraint, and it can be incorporated in most standard variable selection methodologies. 

A drawback is that the assumption that the true frequencies are in $\Omega$ is unlikely to hold. This causes a non-perfect equivalence between the two models, making \eqref{eq:varsel} a surrogate of \eqref{eq:model}. An intuitive solution is to increase $p$, i.e., a finer mesh grid $\Omega$. Variable selection remains feasible as we can enforce a sparsity constraint, which here reads as only a ``small" number of frequencies are relevant. However, covariates become increasingly correlated, making variable selection more challenging. This motivates the use of a variable selection method ``robust" to correlated regressors. A method particularly suited for this task is the Sum of Single Effects (SuSiE) regression proposed in \cite{wang2020simple}, discussed below in greater detail, whereas we defer how to select the grid $\Omega$ in practice to Supplementary D.

SuSiE is a Bayesian stepwise selection method that assumes that at most $N_E$ entries of $\bm{b}^{(1)}$ and $\bm{b}^{(2)}$ (we drop subscripts $i,j$ for lighter notation) are non-zero.  Hence, we can write a prior for $\bm{b}^{(1)}$ and $\bm{b}^{(2)}$ as a sum of $N_E$ single effect models
\begin{align}\label{eq:prior_susie}
	\bm{b}^{(1)} &= \sum_{e=1}^{N_E} \bm{b}_{e}^{(1)}, \qquad \bm{b}^{(2)} = \sum_{e=1}^{N_E} \bm{b}_{e}^{(2)}, \\
	\bm{b}_{e}^{(1)} &= \bm{\gamma}_{e}b_{e}^{(1)}, \qquad \bm{b}_{e}^{(2)} = \bm{\gamma}_{e}b_{e}^{(2)}, \\
	\bm{\gamma}_{e} &\sim Mult_p(1,\bm{\pi}), \qquad \left(\begin{array}{c}
		b_{e}^{(1)}  \\ b_{e}^{(2)}
	\end{array}\right) \sim N_2(0,\sigma^2_0\mathbf{I}_2),
\end{align}
with $e=1,\ldots, N_e$, each $\bm{\gamma}_e\in\{0,1\}^p$ is p-vector indicator variable with one non-zero entry describing which regressor to include in the model,  $\sigma^2_0>0$ and $\bm{\pi}\in (0,1)^p$ are fixed hyperparameters and define the prior variance of the coefficients and the prior inclusion probability for each frequency, respectively. The indicators $(\bm{\gamma}_e)_{1:N_e}$ are shared for $\bm{b}_{e}^{(1)}$ and $\bm{b}_{e}^{(2)}$ to enforce \eqref{eq:sparsity_condition}. 

Rather than approximating the posterior distribution $p((\bm{\gamma}_e)_{1:N_e}, (\bm{b}_e)_{1:N_e} | \mathbf{y}_{iI_j})$ directly, SuSiE employs a backfitting algorithm \citep{fri81} that solves a single effect model at a time. The peculiar feature is that instead of returning a variable to include at each step, it returns at each step a distribution that describes uncertainty on the selected variables.  \cite{wang2020simple} show that this constitutes a mean-field variational Bayes (VB) approximation of the posterior distribution of the form $\prod_{e=1}^{N_E} q_e(\bm{\gamma}_e, \bm{b}_e)$. By model construction, variables that have high posterior mass $q_e(\bm{\gamma}_e)$ are likely to capture similar effects, possibly due to the high correlation. This feature makes SuSiE particularly suitable for our case as $p$ increases: frequencies close to the true ones should be captured by the same effect because the corresponding Fourier basis functions are highly correlated. Figure \ref{fig:example_susie} depicts a numerical illustration of this effect. We simulate data from $\eqref{eq:model}$ with constant parameters ($d=1$, $L_{j}=2$, $\omega_1=1/30$, $\omega_2=1/15$) and fit SuSiE using equally spaced grids with varying $p$ (true frequencies not included in $\Omega$). With $p=250$ (coarser grid), adjacent columns in $\mathbf{X}^{(1)}$ and $\mathbf{X}^{(2)}$ are less correlated and $q(\bm{\gamma}_1)$ and $q(\bm{\gamma}_2)$ are concentrated on the frequency in $\Omega$ closest to the truth. As $p$ increases, correlation in $\mathbf{X}^{(1)}$ and $\mathbf{X}^{(2)}$  increases, and the posterior distributions remain centred around the true frequency, but more diffuse and incorporating more features.

Other Bayesian variable selection methods may yield the same information, but tend to be less computationally efficient. Fast methods and algorithms are available \citep{johnson2012bayesian,rovckova2014emvs}, but often provide only marginal posterior inclusion probability; see \cite{wang2020simple} for a further discussion. The approximate posterior distribution $\prod_{e=1}^{N_E} q_e(\bm{\gamma}_e, \bm{b}_e)$ allows us to quantify uncertainty both in the intensities and the frequencies. What we lose -- in comparison to a ``complete" Bayesian variable selection method or \cite{hadj2020bayesian} -- is the uncertainty quantification on $L_{ij}$: \cite{wang2020simple} provide a heuristic to estimate it but not to quantify uncertainty. A key to estimating model dimension is the parameter $N_E$, which serves as an upper bound on the model dimension, and as such, there is no need to tune it segment by segment. \cite{wang2020simple} argue that SuSiE, when used for variable selection, is robust even when $N_E$ is overstated. In our application of SuSiE within a change point detection framework, we found the method to be similarly robust when $N_E$ is understated when detecting the change points - in some simulations, this even led to improved segmentation performance - and then overstated when estimating the parameters within each segment. We revisit the topic in Supplementary D, where we describe an automatic procedure to select $N_E$.

\begin{figure}[!t]
	\centering
	\subfigure[$T=100$ and $p=250$]{\includegraphics[width=.3\textwidth]{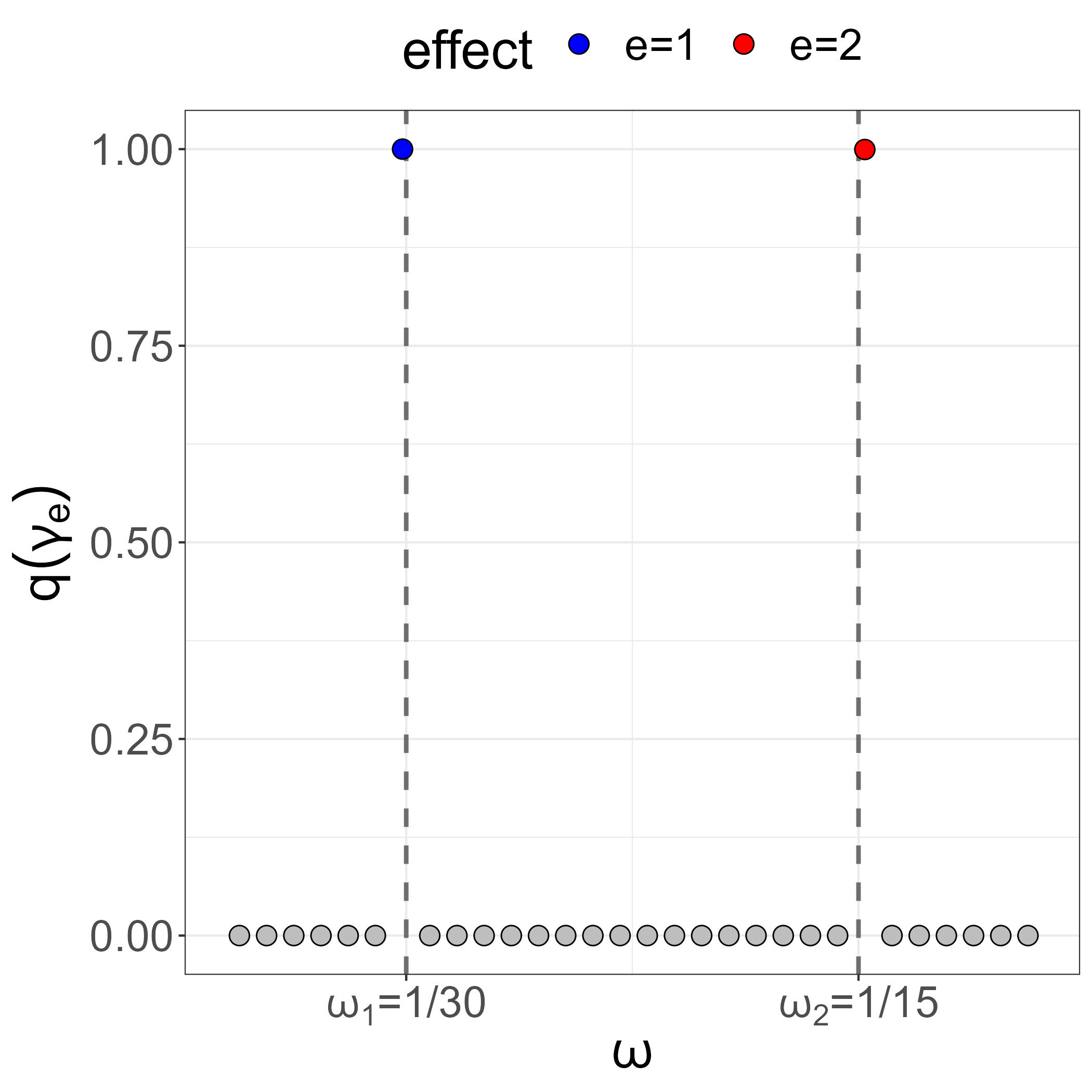}}
	\subfigure[$T=100$ and $p=500$]{\includegraphics[width=.3\textwidth]{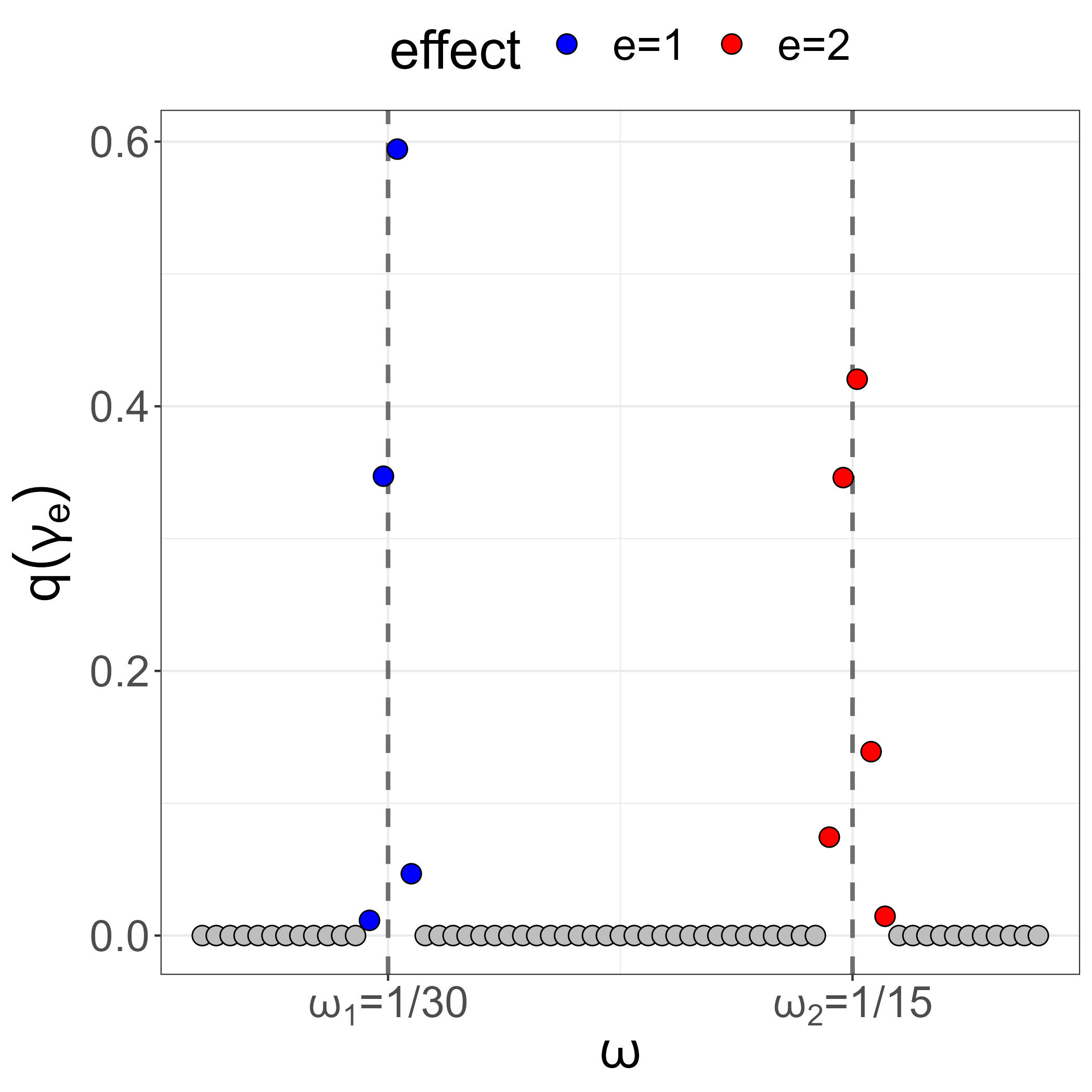}}
	\subfigure[$T=100$ and $p=1000$]{\includegraphics[width=.3\textwidth]{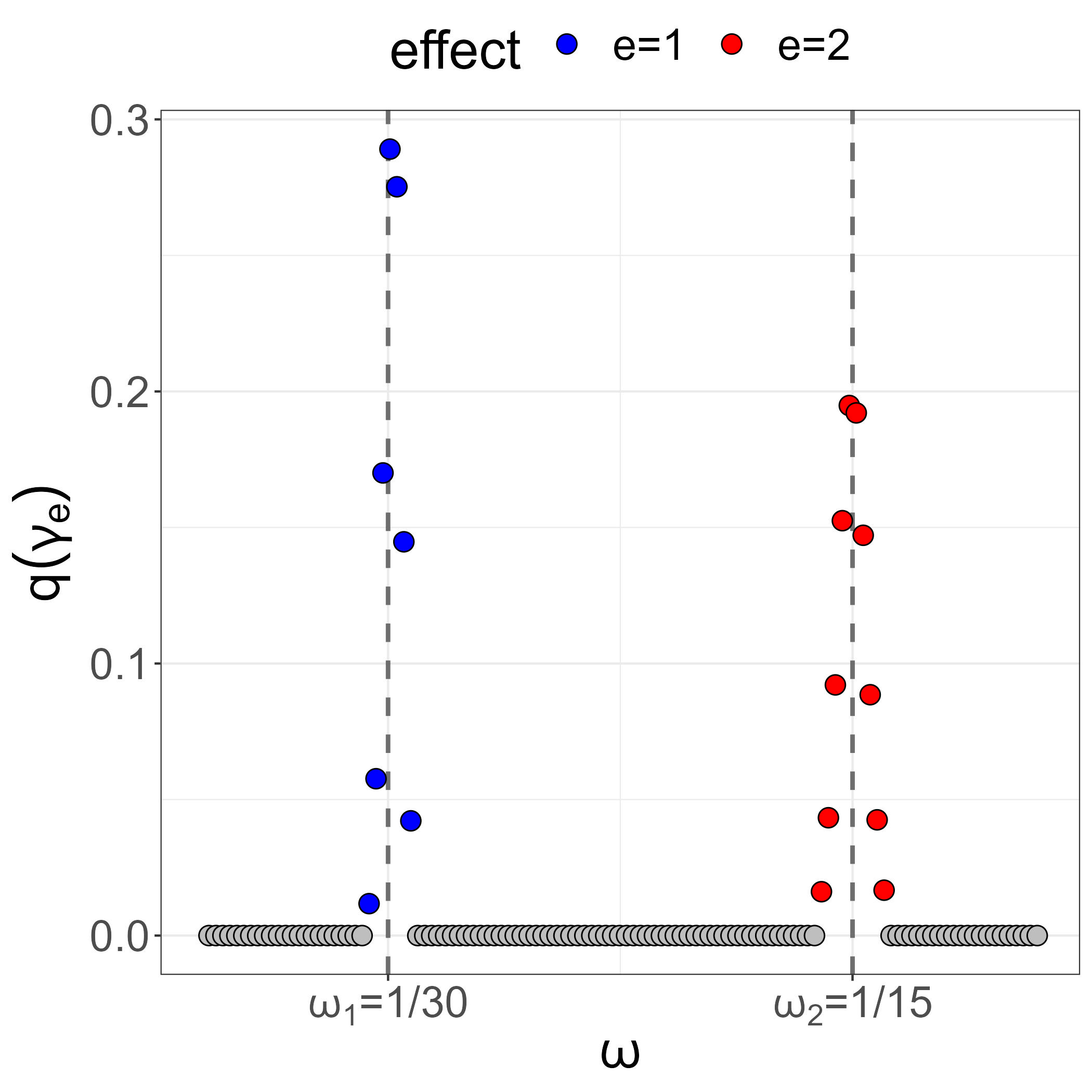}}
	\caption{{Feature selection uncertainty with SuSiE}. Uncertainty quantification around features selected using SuSiE ($N_E=2$). $T=100$ observations are generated from model \eqref{eq:model} with $d=1$,  no change points, and $\omega_1=1/30$ and $\omega_2=1/15$ ($L=2$). Panels depict $q(\gamma_e)$ obtained running the algorithms with equally spaced grid $\Omega$ for $p=250,500,1000$.}\label{fig:example_susie}
\end{figure}

We argued that SuSiE is particularly suitable for \eqref{eq:varsel}, but we reiterate that any method to perform variable selection can be used for within-segment inference. The end goal is to compute either $\max_{\bm{\theta}_{ij}} \log p (\mathbf{y}_{iI_j}| I_j,\bm{\theta}_{ij})$ needed in \eqref{eq:cp} or $p(\mathbf{y}_{I}| \mathcal{P})$ needed in \eqref{eq:marginal}; where we recall that, by independence, $p(\mathbf{y}_{I}| \mathcal{P})=\prod_{i=1}^d p(\mathbf{y}_{iI}| \mathcal{P})$. The advantage of using $p(\mathbf{y}_{I}| \mathcal{P})$ in lieu of $\max_{\bm{\theta}_{ij}} \log p (\mathbf{y}_{iI_j}| I_j,\bm{\theta}_{ij})$ is that the we are integrating over all possible models. Intuitively, this should make the procedure more robust to model selection mistakes.
On the other hand, while obtaining $\max_{\bm{\theta}_{ij}}  \log p(\mathbf{y}_{iI_j}| I_j,\bm{\theta}_{ij})$ is straightforward, computing the marginal likelihood $p(\mathbf{y}_{1:T}| \mathcal{P}) = \sum_{I\in\mathcal{P}} p(\mathbf{y}_{I}| I)$ is a notoriously difficult problem \citep{gelman2013bayesian}. SuSiE approximate posterior distribution is obtained maximizing an objective function that is a variational lower bound to the marginal likelihood. Such a relationship has justified using the lower bound as a proxy for the marginal likelihood in certain applications; for example, in mixture models \citep{ueda2002bayesian,mcgrory2007variational}. How tight the bound is depends on the specific problem, and using the objective function as a proxy for the marginal likelihood can be misleading and not justified in theory \citep{blei2017variational}, even if there are instances where theoretical support for such an approximation is available \citep{cherief2019consistency}, and the approximation can improved; for example \cite{burda2015importance}. However, as we will discuss in Section~\ref{sec:sim}, we find the algorithm robust to this approximation.

\subsection{Change point detection}\label{sec:cp}
There are various classes of algorithms to solve the optimization problem \eqref{eq:marginal} (or equivalently \eqref{eq:cp}), some based on dynamic programming \citep{auger1989algorithms,killick2012optimal}, others on approximate greedy search algorithms. Binary segmentation (BS) \citep{vostrikova1981detecting} is a widely used greedy search algorithm whose popularity is due to its inherent simplicity and the fact that it can be very fast. 
BS recursively splits an interval $(u,v]$ at a point $s\in (u,v]$ when it leads to an increase in a chosen score function called gain function, which should provide a measure of benefit of adding a split at point $s$. An example is the generalized likelihood ratio. In our problem, we will consider the following gain function:
 \begin{equation}\label{eq:gain}
 	G_{(u,v]}(s) =  \frac{ \log p(\mathbf{y}_{(u,s]}| I=(u,s])+\log  p(\mathbf{y}_{(s,v]}| I=(s,v]) }{ \log p(\mathbf{y}_{(u,v]}| I=(u,v])},
 \end{equation}
where we are rescaling the gain in marginal log-likelihood obtained from splitting the model at time $s$. The algorithm is initialized with $(u,v]=(0, T]$ and proceeds recursively until a stopping rule is met. Generally, such stopping rule is of the form $G_{(u,v]}(s) \leq \delta$, for some $\delta>1$. 

BS typically requires $\mathcal{O}(T \log T)$ evaluations of the gain function, which is not a major issue in change point questions when $G_{(u,v]}(s)$ is available in closed form (\textit{e.g.} CUSUM statistics, \cite{page1954continuous}) but it is quite costly in a set-up like ours where each marginal likelihood is approximated via a computational expensive procedure. \cite{kovacs2020optimistic} noted that the gain function in many parametric settings has a piecewise quasi-convex structure, and proposed a novel BS segmentation procedure called optimistic search (OS) that reduces the number of gain function evaluations to $\mathcal{O}(\log T)$. We deem this version of BS particularly suitable to our context. 

The stopping rule $\delta$ is crucial and it can be easily seen that the number of change point detected varies as a function of $\delta$. In our context, the gain $G_{(u,v]}(s)$ can be interpreted as a Bayes factor between the models with and without change point, and the Bayesian literature has discussion on how to set it, see for example \cite{kass1995bayes}. In the change point literature, another strategy is choose $\delta$ that optimize a certain objective function, as described by \cite{fryzlewicz2014wild} Section 3.3. Define the set of change points obtained for a given threshold $\delta$ as $\mathcal{S}(\delta)=(\hat{t}_1,\ldots,\hat{t}_{\hat{m}_\delta})$, and let $\hat{m}_\delta$ be its cardinality. By increasing the threshold $\delta^1 < \delta^2$ we have that $\hat{m}_{\delta^1} \geq \hat{m}_{\delta^2}$ and, as a consequence, $\mathcal{S}(\delta^2) \subseteq \mathcal{S}(\delta^1)$. Therefore, choosing the optimal threshold can be seen as choosing the optimal number of change points, which be done using an appropriate information criterion; for example, MBIC \citep{zhang2007modified}, Minimum Description Length \citep{rissanen1978modeling}, or Strengthened Schwarz criterion \citep{fryzlewicz2014wild}.
 
Supplementary D includes a guideline to select the algorithm parameters of Sections~\ref{sec:within} and \ref{sec:cp}.

\subsection{Theory}
\label{subsec:theory}
We study the properties of the change point estimator $\hat{t} =\max_t G_{(0,T) (t)}$ in the presence of a single change point at $t_0$. We determine a form of signal-to-noise ratio used in the context of change point detection methods and the localization error, which is a sequence  $(\epsilon_T)_{T\geq 1}$ such that
\begin{equation}\label{eq:localization_error}
\lim_{T\to \infty} \mathbb{P} (\vert \hat{t}-t_0 \vert< \epsilon_T)= 1.
\end{equation} 
See \cite{yu2020review} for a review of state-of-the-art results. Let $m_{01}$ and $m_{02}$ denote the set of true frequencies in the left and right segments, respectively, and let $\mathbf{X_{m_{01}}}$ and  $\mathbf{X_{m_{02}}}$ be  the corresponding design matrices. These matrices are constructed by merging the sine and cosine basis functions matrix as described in Section~\ref{sec:within}, keeping only the columns corresponding to the true frequencies. Without loss of generality, we assume that columns are sorted in such a way that we can define a column index set $I_0$, corresponding to the shared frequencies $m_{01}\cap m_{02}$, and index sets $I_1$ and $I_2$, corresponding to the unique frequencies $m_{01}\setminus m_{02}$ and $m_{02}\setminus m_{01}$ respectively. While any of these sets may be empty, at least one of them is guaranteed to be non-empty.
The main result is based on the following assumption.

\begin{assumption} \label{ass1_s}
	Let $(Y_t)_{1:T}$ be a sequence of independent random variables sub-Gaussian with parameters $(\sigma_t)_{1:T}$ satisfying $\sup_t \sigma_t<\infty$. Let $t_0$ be a time instance such that $Y_{t} \sim  F_1$ for $t \leq t_0$ and $Y_{t} \sim F_2$ for $t > t_0$, where $F_1$ and $F_2$ are such that 
	    \begin{align*}
		\mathbb{E}(Y_t) &= 
		\begin{cases}
 \mathbf{x}_{t01}^\intercal\bm{\beta}_{1} & \text{if } t \leq t_0, \\
	 \mathbf{x}_{t02}^\intercal\bm{\beta}_{2} , & \text{if } t > t_0.
		\end{cases} 
	\end{align*}
 In addition assume that: 
	\begin{enumerate}[(a)]
		\item (minimum spacing) $\Delta:=\min\{t_0,T-t_0+1\}> \sqrt{T \log T}$.
		\item  (signal-to-noise ratio) Let $\kappa:=\bbe^{\intercal}_{I_1 1} \bbe_{I_1 1} +\bbe^{\intercal}_{I_22} \bbe_{I_22} + (\bbe_{I_01}-\bbe_{I_02})^\intercal(\bbe_{I_01}-\bbe_{I_02})$, for some $\epsilon>0$, it holds that $\kappa \Delta>\sqrt{T \log^{1+\epsilon}T}$.
	\item (bounded signal)  $0<c<\kappa <\infty$. 
	\item (correct model mispecification) $m_{01},m_{02} \subseteq \Omega$.

        \item The hyperparameters are $0 < \sigma^2_0 < +\infty$, $N_e \geq \max\{|m_{01}|,|m_{02}|\}$, and $\pi_k>0$ for all $k$. 
	\end{enumerate}
\end{assumption}

The method was developed under Gaussian error terms. However, Assumption~\ref{ass1_s} simply requires the data to be sub-Gaussian with uniformly bounded parameters. Assumption~\ref{ass1_s} \textit{(a)} corresponds to a minimum spacing condition. $\Delta$ becomes meaningful when considered alongside the signal strength, denoted by $\kappa$ in Assumption~\ref{ass1_s} \textit{(b)}, where the signal-to-noise ratio relevant for change point detection is also defined. As a reference point, $\sqrt{T\log T}$ corresponds to the minimax detection bound for change point problems involving changes expressible by polynomial functions \citep{yu2020review}. To our knowledge, no results are currently available on the minimax rate associated with the signal structure in Assumption~\ref{ass1_s}. $\kappa$ captures, in a single quantity, three possible type of changes: the intensities can change, the frequencies can change, and the model dimension can change. Assumption~\ref{ass1_s} \textit{(d)} requires the true frequencies to be part of the discrete set $\Omega$. We make no assumptions on model dimensions, $|m_{01}|$ and $|m_{02}|$, but require $N_E$ to be sufficiently large to enable successful model selection (Assumption~\ref{ass1_s} \textit{(e)}). In addition, Assumption~\ref{ass1_s} \textit{(e)} requires that the prior distributions are well-defined, and a positive inclusion probability is assigned to every model. We are ready to state the main result of the section.

\begin{theorem}\label{th:loc_rate}
	Suppose that Assumption~\ref{ass1_s} holds. Then for any $\epsilon>0$, there exists $C^*>0$ such that, with probability approaching one, we have that 
		\begin{equation}
		\max_{t \in \{1,\ldots,T\}: |t-t_1|>C^* \log T} G_{(0,T]}(t) < G_{(0,T]}(t_1).
		\end{equation}
\end{theorem}

The proof is given in Supplementary C. Theorem~\ref{th:loc_rate} establishes that the estimator attains a localization error of order $\log T$. This is comparable to most state-of-the-art change point detection methods \citep{yu2020review}. Notably, $\epsilon_T/\Delta \to 0$ as $T\to\infty$, which is the notion of consistency for change point detection methods.

The statistic $G_{(0,T]}(t)$ based on the marginal likelihood studied in this paper is novel in the context of change point detection. The key idea of the proof is to decompose $G_{(0,T]}(t)$ in a change point detection component and a variable selection one. Conditionally on correct model recovery (\text{i.e.}, conditionally on the true frequencies), the estimator achieves the desired localization error. The proof builds on ideas in \cite{berlind2025}, but the argument differs by \cite{berlind2025} because they do not deal with linear regression. Then, we prove that the variable selection is consistent (\text{i.e.}, the true frequencies are selected in the linear model) and does not affect the rate of the estimator.

Here, the assumption that the true frequencies are in $\Omega$ appears important, and dropping this assumption should be subject of future work. Examining the proof, it seems possible to drop the Assumption that $N_E$ is larger or equal than the true model dimension. However, the argument appears more convoluted. Similarly, one may be able to weaken the signal-to-noise ratio considering one type of change at the time, in lieu of the three changes jointly.

\section{Simulation studies}
\label{sec:sim}
In this section, we present a simulation study to evaluate the performance of the proposed method, referred to as \algonamesp\ (Change Point detection based on Variable Selection within-segment). We examine different scenarios, varying $T$, $m$, change point locations, the variance, and the noise process. In three of the considered settings, the data generating process is correctly specified (\textit{i.e.}, \eqref{eq:model}), in six other is misspecified.

To assess segmentation quality, we compute the bias in the estimated number of change points as $B({m},\widehat{m})=|m-\widehat{m}|$, and two measures to describe the segmentation accuracy. We consider the coverage of the estimated partition $\hat{\mathcal{P}}=\{[0,\hat{t}_0],\ldots,[\hat{t}_{\hat{m}}+1,T]\}$ with respect to the true one $\mathcal{P}=\{[0,{t}_0],\ldots,[{t}_{{m}}+1,T]\}$ \citep{van2020evaluation}:
\begin{equation*}
	C(\mathcal{P},\hat{\mathcal{P}})=\frac{1}{T}\sum_{A\in\mathcal{P}}|A|\max_{\hat{A}\in\hat{\mathcal{P}}}\mathsf{J}(A,\hat{A}), \qquad J(A,\hat{A})=\frac{\mid A \cap \hat{A}\mid}{\mid A \cup \hat{A}\mid}.
\end{equation*}
Second, we compute the Hausdorff distance:
\begin{equation*}
H({\mathcal{P}},\hat{\mathcal{P}})=\frac{1}{T}\max\left\{dH({\mathcal{P}},\hat{\mathcal{P}}),dH(\hat{\mathcal{P}},{\mathcal{P}})\right\}, \qquad dH({\mathcal{P}},\hat{\mathcal{P}}) = \max_{p_1\in {\mathcal{P}}} \min_{p_2\in\hat{\mathcal{P}}} \Vert p_1-p_2\Vert,
\end{equation*}
where $dH({\mathcal{P}},\hat{\mathcal{P}})$ is the non-symmetric  Hausdorff distance that computes the distance from $p_1\in\mathcal{P}$ to its nearest neighbour in $\hat{\mathcal{P}}$. Finally, we report the runtime. To assess the parameter estimation accuracy, we consider the in-sample root mean squared error $RMSE(y_t,\hat{y}_t)$ and the one between the true and estimated signal $RMSE(\mu_t,\hat{\mu}_t)$. For each metric, we report the mean of $R=50$ replicates for each scenario. 

We compare three parameterizations of CPVS: $N_E=\{1,2\}$, denoted by \algoname($N_E$), and \algonamesp with the automatic procedure to set $N_E$ ($\bar{N}_E=2$, denoted by \algoname); see Supplementary D.
The grid $\Omega$ contains the $p=50$ frequencies with the highest power spectrum in the periodogram and we assign equal prior inclusion probability to each, namely $\pi_k=1/50$. Prior variance of the intensities $\sigma^2_0$ is set equal to $1$. The threshold $\delta$ is set to $1.01$ and the number of change points is estimated with the minimum description length (MDL) criterion.

We compare our method with the Automatic Nonstationary Oscillatory Modelling (ANOM) algorithm \citep{hadj2020bayesian}. We consider three parametrizations, varying the prior mean on the number of change points $\lambda_{m}=\{0.1,1,10\}$, and we set the prior mean on the number of frequencies $\lambda_\omega=2$. We denote the algorithm with a specific $\lambda_m$ as ANOM($\lambda_m$). We draw 20000 samples from the posterior, discarding the first 5000 as a burn-in period. These parameters' choice follows the authors' recommendation. We also include the piecewise autoregressive process of \cite{davis2006structural} (denoted as AP) and the adaptive spectral estimation of \cite{rosen2012adaptspec} for two different specifications of basis functions, $J=7$ and $J=15$, referred to as ADA(7) and ADA(15), respectively. We also consider the approach of \cite{korkas2017multiple} based on Locally Stationary Wavelets (hereafter denoted with LSW), and the method of \cite{wu2024frequency} that uses overlapping-block multiplier bootstrap (OBMB). AP is misspecified in the first four scenarios, ADA, LSW, and OBMB are nonparametric. For the multivariate setting, we compare CPVS with ADA-X \citep{bertolacci2022adaptspec}.

\subsection{Correctly specified model}

\noindent\textbf{Scenario 1.a (varying standard deviation).} \citep{hadj2020bayesian}. We set $T=900$, $m=2$, and change points at $\{300,650\}$. Figure E.1 in the Supplementary right column panels depicts $\mu_t$ in the three segments ($L_1=3$, $L_2=1$, $L_3=2$). We consider two standard deviations ($\sigma_j=1$, $\sigma_j=3$) to simulate different noise levels. Results are given in Supplementary E Table E.1. The segmentation performance is almost perfect for \algonamesp (coverage above $0.98$, Hausdorff distance close to $0$, and bias equal to $0$), ANOM has a comparable performance but with higher bias in the low noise case (bias bigger than $0.22$). OBMB provides good segmentation results for $\sigma=1$, while its performance deteriorates as the magnitude of the noise increases. OBMB appears to be sensitive to the choice of hyperparameters. We run it under several settings and reported the best performance. In addition, OBMB's running time is 600 times higher than \algonamesp and 8.5 times higher than ANOM. For this reason, OBMB is not considered in the next scenarios where $T$ is bigger. AP, ADA, and LSW have poorer performances, in particular, they cannot detect almost any change point in the higher noise regime. For this reason, we don't consider AP, ADA, and LSW in the following scenario, where we present a more complex simulation study.\\

\vspace{-0.3cm}\noindent\textbf{Scenario 2.a (varying mean and change point locations).} The locations of the change point and the mean functions within each segment are resampled at each replicate. See Supplementary E.1 for more details on the data generating mechanism. Here, no continuity in $\mathbb{E}(Y_t)$ at the change points is imposed (as in \eqref{eq:model}); next subsection, we consider a scenario where the mean is continuous. Within the same set-up, we explore two different cases. The first one assumes increasing both sample size and the number of change points $(T,m)=\{(500,2),(1000,4),(5000,8),(10000,12)\}$, while the second one fixes the sample size $T=1000$ and gradually increases the number of change points $m=\{2,3,4,5,6\}$. Standard deviation is constant ($\sigma=3$). Table \ref{tab:results1} presents metrics of segmentation accuracy, Table E.2 in the Supplementary summarizes parameter estimation accuracy, and Figure \ref{fig:results_time} depicts the runtime. 

In the case where both $T$ and $m$ increase, the performance \algonamesp is similar for different parameter specifications, with coverage around 0.89, Hausdorff distance around $0.05$, and a bias that varies depending on on $(T,m)$. ANOM performance is affected by the hyperparameter specification: it outperforms \algonamesp in terms of coverage and Hausdorff distance (average around $0.95$ and $0.03$) with $\lambda_m=10$, while its performance is much worse than \algonamesp when $\lambda_m=0.1$. This effect is particularly evident when looking at the bias, with ANOM failing to detect approximately half of the change points with $(T=10000,m=12)$, while \algonamesp misses one.

The case $T$ fixed and varying $m$ confirms the observed pattern: ANOM can be slightly more accurate for certain parameter specifications, but its performance deteriorates severely for other parametrizations. In addition, accuracy (bias, Hausdorff distance, and coverage) decreases as $m$ increases for both methods. However, \algoname's decay is less pronounced than ANOM. The coverage of \algonamesp is below 0.80 only in one scenario ($m=6$ and $N_E=1$) and its Hausdorff distance never exceeds 0.15, while ANOM shows coverages below 0.60 and Hausdorff distance higher than $0.3$ when $m=\{5,6\}$ and $\lambda_m=\{0.1,1\}$. The bias exhibits a similar trend.

A similar analysis applies to parameter estimation accuracy (Table E.2): \algonamesp performs comparably to ANOM across several scenarios. While ANOM—which does not rely on discretisation—can achieve higher accuracy, its performance is highly sensitive to the choice of hyperparameter; for example, when $T=10000,m=12$, ANOM(10) is the best performing model in terms of RMSE, ANOM(0.1) the worst one.

As expected, \algonamesp has a much lower average runtime than ANOM (Figure \ref{fig:results_time}). The advantage increases with the sample size (Figure \ref{fig:results_time} left panel). For instance, when considering $(T,m)=(10000,12)$, \algonamesp requires on average 7 seconds for $N_E=1$, and 20 seconds for $N_E=2$, while ANOM takes about 15 minutes when $\lambda_m=10$. ANOM's runtime is not affected by the number of change points (Figure \ref{fig:results_time} right panel), while \algonamesp's runtime increases with $m$. The advantage remains striking. As mentioned above, it is unfeasible to include OBMB for computational reasons. From preliminary results on a limited number of replicates, OBMB' runtime for a single replicate is about 8 minutes when $T=1000$, while it rapidly increases to $\sim 3$ hours and $\sim 12$ hours when $T=5000$ and $T=10000$, respectively.

% Table generated by Excel2LaTeX from sheet 'Foglio1'
\begin{table}[!t]
  \centering
    \resizebox{0.8\textwidth}{!}{
    \begin{tabular}{clrrrrrrrrrrr}
    \toprule
          &      &       & \multicolumn{4}{c}{$(T,\,m)$}        &       & \multicolumn{5}{c}{$(T=1000,\,m)$} \\
          &   Method    &       & $(500,2)$ & $(1000,4)$ & $(5000,8)$ & $(10000,12)$ &       & $2$     & $3$     & $4$     & $5$     & $6$ \\
              \cmidrule{2-2}\cmidrule{4-7}\cmidrule{9-13}
    \multirow{6}[0]{*}{\begin{turn}{90}$C(\mathcal{P},\hat{\mathcal{P}})$\end{turn}} & CPVS(1) &       & \cellcolor[rgb]{ .984,  .871,  .882}0,86 & \cellcolor[rgb]{ .984,  .851,  .859}0,84 & \cellcolor[rgb]{ .984,  .898,  .91}0,89 & \cellcolor[rgb]{ .984,  .902,  .914}0,89 &       & \cellcolor[rgb]{ .98,  .839,  .851}0,82 & \cellcolor[rgb]{ .984,  .898,  .91}0,89 & \cellcolor[rgb]{ .984,  .851,  .859}0,84 & \cellcolor[rgb]{ .98,  .831,  .843}0,82 & \cellcolor[rgb]{ .98,  .78,  .792}0,76 \\
          & CPVS(2) &       & \cellcolor[rgb]{ .984,  .914,  .922}0,90 & \cellcolor[rgb]{ .984,  .886,  .898}0,88 & \cellcolor[rgb]{ .984,  .89,  .902}0,88 & \cellcolor[rgb]{ .984,  .867,  .875}0,85 &       & \cellcolor[rgb]{ .984,  .894,  .906}0,89 & \cellcolor[rgb]{ .984,  .922,  .933}0,92 & \cellcolor[rgb]{ .984,  .886,  .898}0,88 & \cellcolor[rgb]{ .984,  .847,  .855}0,83 & \cellcolor[rgb]{ .984,  .843,  .855}0,83 \\
          & CPVS  &       & \cellcolor[rgb]{ .984,  .929,  .941}0,92 & \cellcolor[rgb]{ .984,  .894,  .906}0,88 & \cellcolor[rgb]{ .984,  .918,  .929}0,91 & \cellcolor[rgb]{ .984,  .91,  .918}0,90 &       & \cellcolor[rgb]{ .984,  .933,  .941}0,92 & \cellcolor[rgb]{ .984,  .929,  .941}0,92 & \cellcolor[rgb]{ .984,  .894,  .906}0,88 & \cellcolor[rgb]{ .984,  .855,  .867}0,84 & \cellcolor[rgb]{ .98,  .82,  .831}0,80 \\
          & ANOM(0.1) &       & \cellcolor[rgb]{ .984,  .859,  .867}0,84 & \cellcolor[rgb]{ .98,  .741,  .753}0,72 & \cellcolor[rgb]{ .98,  .753,  .765}0,73 & \cellcolor[rgb]{ .976,  .631,  .639}0,60 &       & \cellcolor[rgb]{ .984,  .961,  .973}0,96 & \cellcolor[rgb]{ .984,  .906,  .918}0,90 & \cellcolor[rgb]{ .98,  .741,  .753}0,72 & \cellcolor[rgb]{ .976,  .616,  .624}0,58 & \cellcolor[rgb]{ .973,  .412,  .42}0,36 \\
          & ANOM(1) &       & \cellcolor[rgb]{ .984,  .918,  .929}0,91 & \cellcolor[rgb]{ .984,  .89,  .898}0,88 & \cellcolor[rgb]{ .984,  .961,  .973}0,96 & \cellcolor[rgb]{ .984,  .925,  .937}0,92 &       & \cellcolor[rgb]{ .984,  .984,  .996}0,98 & \cellcolor[rgb]{ .984,  .965,  .976}0,96 & \cellcolor[rgb]{ .984,  .89,  .898}0,88 & \cellcolor[rgb]{ .98,  .82,  .831}0,80 & \cellcolor[rgb]{ .976,  .604,  .612}0,57 \\
          & ANOM(10) &       & \cellcolor[rgb]{ .984,  .949,  .961}0,94 & \cellcolor[rgb]{ .984,  .937,  .949}0,93 & \cellcolor[rgb]{ .984,  .976,  .988}0,97 & \cellcolor[rgb]{ .984,  .984,  .996}0,98 &       & \cellcolor[rgb]{ .988,  .988,  1}0,98 & \cellcolor[rgb]{ .984,  .973,  .984}0,97 & \cellcolor[rgb]{ .984,  .937,  .949}0,93 & \cellcolor[rgb]{ .984,  .875,  .882}0,86 & \cellcolor[rgb]{ .98,  .784,  .796}0,77 \\
          &       &       &       &       &       &       &       &       &       &       &       &  \\
    \multirow{6}[0]{*}{\begin{turn}{90}$H(\mathcal{P},\hat{\mathcal{P}})$\end{turn}} & \multicolumn{1}{l}{CPVS(1)} &       & \cellcolor[rgb]{ .984,  .843,  .851}0,08 & \cellcolor[rgb]{ .984,  .769,  .78}0,11 & \cellcolor[rgb]{ .988,  .886,  .898}0,06 & \cellcolor[rgb]{ .988,  .914,  .925}0,05 &       & \cellcolor[rgb]{ .984,  .831,  .843}0,15 & \cellcolor[rgb]{ .988,  .906,  .918}0,09 & \cellcolor[rgb]{ .988,  .871,  .882}0,11 & \cellcolor[rgb]{ .988,  .882,  .894}0,11 & \cellcolor[rgb]{ .988,  .871,  .882}0,12 \\
          & \multicolumn{1}{l}{CPVS(2)} &       & \cellcolor[rgb]{ .988,  .894,  .906}0,06 & \cellcolor[rgb]{ .988,  .875,  .886}0,07 & \cellcolor[rgb]{ .988,  .906,  .918}0,05 & \cellcolor[rgb]{ .988,  .906,  .914}0,05 &       & \cellcolor[rgb]{ .988,  .902,  .914}0,09 & \cellcolor[rgb]{ .988,  .957,  .965}0,04 & \cellcolor[rgb]{ .988,  .929,  .941}0,07 & \cellcolor[rgb]{ .988,  .918,  .925}0,08 & \cellcolor[rgb]{ .988,  .925,  .937}0,07 \\
          & \multicolumn{1}{l}{CPVS} &       & \cellcolor[rgb]{ .988,  .898,  .91}0,05 & \cellcolor[rgb]{ .988,  .847,  .859}0,08 & \cellcolor[rgb]{ .988,  .918,  .929}0,05 & \cellcolor[rgb]{ .988,  .922,  .933}0,04 &       & \cellcolor[rgb]{ .988,  .941,  .953}0,05 & \cellcolor[rgb]{ .988,  .941,  .953}0,05 & \cellcolor[rgb]{ .988,  .914,  .925}0,08 & \cellcolor[rgb]{ .988,  .89,  .902}0,10 & \cellcolor[rgb]{ .988,  .898,  .91}0,09 \\
          & \multicolumn{1}{l}{ANOM(0.1)} &       & \cellcolor[rgb]{ .98,  .659,  .671}0,16 & \cellcolor[rgb]{ .973,  .412,  .42}0,28 & \cellcolor[rgb]{ .98,  .573,  .58}0,20 & \cellcolor[rgb]{ .976,  .482,  .49}0,25 &       & \cellcolor[rgb]{ .988,  .941,  .953}0,05 & \cellcolor[rgb]{ .988,  .882,  .894}0,10 & \cellcolor[rgb]{ .98,  .686,  .694}0,28 & \cellcolor[rgb]{ .98,  .588,  .596}0,36 & \cellcolor[rgb]{ .973,  .412,  .42}0,52 \\
          & \multicolumn{1}{l}{ANOM(1)} &       & \cellcolor[rgb]{ .984,  .843,  .851}0,08 & \cellcolor[rgb]{ .984,  .78,  .788}0,11 & \cellcolor[rgb]{ .988,  .941,  .953}0,04 & \cellcolor[rgb]{ .988,  .875,  .886}0,07 &       & \cellcolor[rgb]{ .988,  .988,  1}0,01 & \cellcolor[rgb]{ .988,  .961,  .973}0,04 & \cellcolor[rgb]{ .988,  .878,  .89}0,11 & \cellcolor[rgb]{ .984,  .816,  .824}0,17 & \cellcolor[rgb]{ .98,  .596,  .608}0,35 \\
          & \multicolumn{1}{l}{ANOM(10)} &       & \cellcolor[rgb]{ .988,  .925,  .937}0,04 & \cellcolor[rgb]{ .988,  .89,  .902}0,06 & \cellcolor[rgb]{ .988,  .969,  .98}0,02 & \cellcolor[rgb]{ .988,  .988,  1}0,01 &       & \cellcolor[rgb]{ .988,  .988,  1}0,01 & \cellcolor[rgb]{ .988,  .98,  .992}0,02 & \cellcolor[rgb]{ .988,  .937,  .949}0,06 & \cellcolor[rgb]{ .988,  .867,  .878}0,12 & \cellcolor[rgb]{ .984,  .8,  .812}0,18 \\
          &       &       &       &       &       &       &       &       &       &       &       &  \\
    \multirow{6}[0]{*}{\begin{turn}{90}$B(m,\widehat{m})$\end{turn}} & CPVS(1) &       & \cellcolor[rgb]{ .988,  .98,  .992}0,16 & \cellcolor[rgb]{ .988,  .937,  .949}0,62 & \cellcolor[rgb]{ .988,  .922,  .933}0,78 & \cellcolor[rgb]{ .988,  .89,  .902}1,12 &       & \cellcolor[rgb]{ .988,  .953,  .965}0,44 & \cellcolor[rgb]{ .988,  .965,  .976}0,32 & \cellcolor[rgb]{ .988,  .937,  .949}0,62 & \cellcolor[rgb]{ .988,  .937,  .949}0,60 & \cellcolor[rgb]{ .988,  .898,  .91}1,02 \\
          & CPVS(2) &       & \cellcolor[rgb]{ .988,  .984,  .996}0,12 & \cellcolor[rgb]{ .988,  .961,  .969}0,38 & \cellcolor[rgb]{ .988,  .918,  .929}0,82 & \cellcolor[rgb]{ .984,  .839,  .851}1,68 &       & \cellcolor[rgb]{ .988,  .973,  .984}0,22 & \cellcolor[rgb]{ .988,  .984,  .996}0,12 & \cellcolor[rgb]{ .988,  .961,  .969}0,38 & \cellcolor[rgb]{ .988,  .945,  .957}0,54 & \cellcolor[rgb]{ .988,  .937,  .949}0,60 \\
          & CPVS  &       & \cellcolor[rgb]{ .988,  .984,  .996}0,12 & \cellcolor[rgb]{ .988,  .957,  .969}0,42 & \cellcolor[rgb]{ .988,  .929,  .941}0,68 & \cellcolor[rgb]{ .988,  .898,  .91}1,02 &       & \cellcolor[rgb]{ .988,  .98,  .992}0,14 & \cellcolor[rgb]{ .988,  .973,  .984}0,22 & \cellcolor[rgb]{ .988,  .957,  .969}0,42 & \cellcolor[rgb]{ .988,  .937,  .949}0,60 & \cellcolor[rgb]{ .988,  .922,  .933}0,80 \\
          & ANOM(0.1) &       & \cellcolor[rgb]{ .988,  .953,  .965}0,46 & \cellcolor[rgb]{ .988,  .851,  .863}1,56 & \cellcolor[rgb]{ .98,  .694,  .702}3,26 & \cellcolor[rgb]{ .973,  .412,  .42}6,30 &       & \cellcolor[rgb]{ .988,  .98,  .992}0,16 & \cellcolor[rgb]{ .988,  .953,  .965}0,46 & \cellcolor[rgb]{ .988,  .851,  .863}1,56 & \cellcolor[rgb]{ .984,  .745,  .757}2,70 & \cellcolor[rgb]{ .98,  .573,  .58}4,58 \\
          & ANOM(1) &       & \cellcolor[rgb]{ .988,  .973,  .984}0,22 & \cellcolor[rgb]{ .988,  .941,  .953}0,56 & \cellcolor[rgb]{ .988,  .953,  .965}0,46 & \cellcolor[rgb]{ .988,  .867,  .875}1,40 &       & \cellcolor[rgb]{ .988,  .988,  1}0,04 & \cellcolor[rgb]{ .988,  .98,  .992}0,16 & \cellcolor[rgb]{ .988,  .941,  .953}0,56 & \cellcolor[rgb]{ .988,  .882,  .894}1,22 & \cellcolor[rgb]{ .984,  .714,  .722}3,06 \\
          & ANOM(10) &       & \cellcolor[rgb]{ .988,  .984,  .996}0,10 & \cellcolor[rgb]{ .988,  .965,  .976}0,32 & \cellcolor[rgb]{ .988,  .976,  .988}0,20 & \cellcolor[rgb]{ .988,  .976,  .988}0,18 &       & \cellcolor[rgb]{ .988,  .988,  1}0,04 & \cellcolor[rgb]{ .988,  .988,  1}0,08 & \cellcolor[rgb]{ .988,  .965,  .976}0,32 & \cellcolor[rgb]{ .988,  .925,  .933}0,76 & \cellcolor[rgb]{ .988,  .847,  .859}1,58 \\
          \bottomrule
    \end{tabular}%
  }
  \vspace{0.3em}
\caption{Scenario 2.a: segmentation performance. CPVS is our proposal, ANOM by \cite{hadj2020bayesian}. Three metrics: segmentation coverage $C(\mathcal{P},\hat{\mathcal{P}})\in (0,1)$ (the higher, the better), Hausdorff distance $H(\mathcal{P},\hat{\mathcal{P}})\in (0,1)$ (the lower, the better), and bias $B(m,\hat{m})\in [0,m]$ (the lower, the better). The table shows the results in both varying $T$ and fixed $T=1000$ scenarios. The red scale depicts worse performances.}\label{tab:results1}%
\end{table}%

\begin{figure}[!t]
\centering
\includegraphics[width=.45\textwidth]{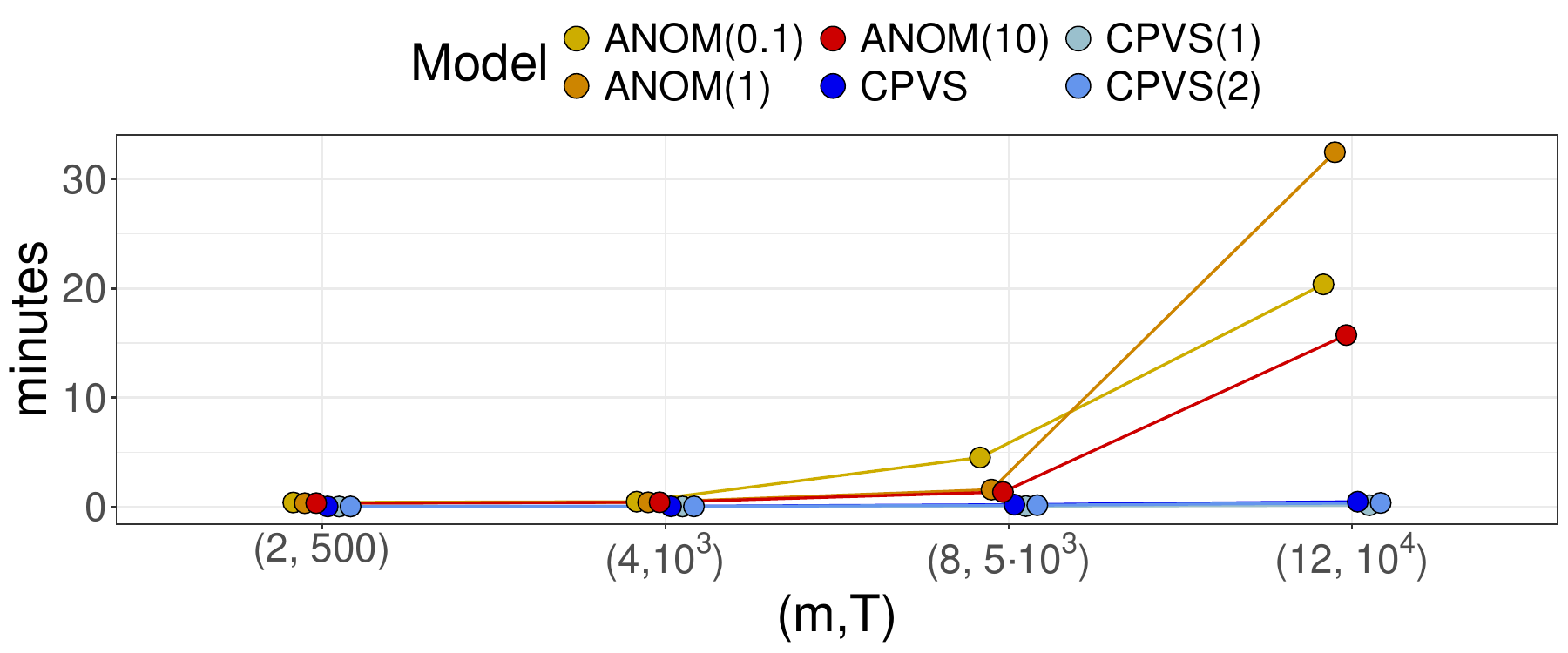}
\includegraphics[width=.45\textwidth]{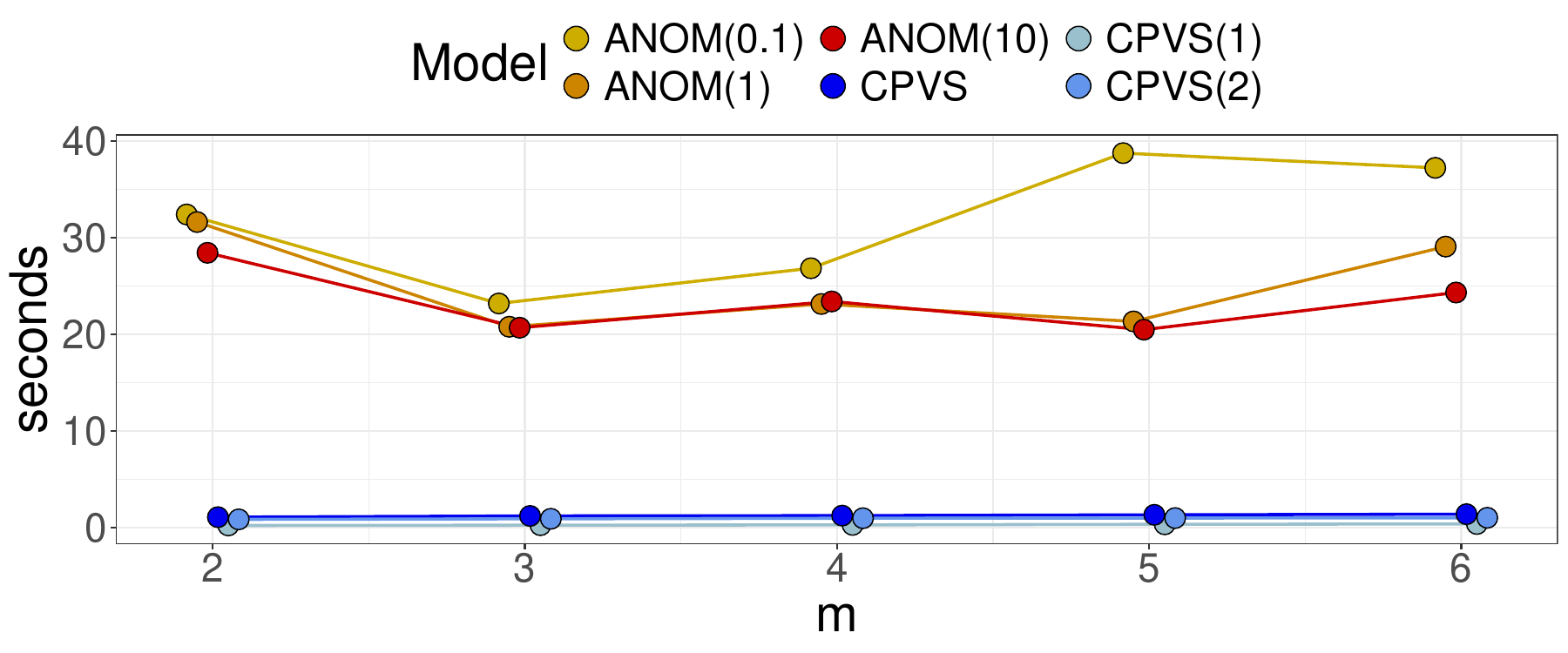}
\caption{{Scenario 2.a: Average runtime}. Left panel depicts the case with varying $(T,m)$, right panel the case with fixed $T=1000$ and varying $m$.}\label{fig:results_time}
\end{figure}

\vspace{-0.3cm}\noindent\textbf{Scenario 3 (multivariate, dense regime).} We replicate Scenario 2 for fixed $T=1000$ and $m=\{4,6\}$, now, varying dimension $d$. All series experience a change and $\sigma_i=3$ for all $i$. Table E.5 in the Supplementary presents results indicating that segmentation accuracy improves as the number of series increases. The multivariate version of \algoname, denoted by MCPVS, outperforms the multivariate nonparametric approach of \cite{bertolacci2022adaptspec} (ADA-X). We further investigate the performance of \algonamesp in presence of an increasing number of noisy series in the panel. We set $d=10$, $m=4$, and consider two noise regimes: $\sigma_{i}=3$ for $i=1,\ldots,d_1$ and $\sigma_{i}=9$ for $i=d_1+1,\ldots,10$, where $d_1\in\{9,7,4\}$. Table E.6 in the Supplementary reports the results. As expected, the performance deteriorates but not uniformly across metrics. Bias is the most severely affect and varies from $0.18$ ($d_1=9$) to $0.96$ ($d_1=4$). MCPVS outperforms ADA-X also in this case.\\

\vspace{-0.3cm}In summary, in the correctly specified scenarios, \algonamesp exhibits comparable performance to ANOM but it is order of magnitude of order faster. It is not surprising that ANOM can achieve a better performance, as it performs exact inference for \eqref{eq:model}, while the discretisation of \algonamesp introduces an approximation. However, ANOM's performance can be much worst if the RJ-MCMC parameters are not chosen well. This confirms knowns issues in the use of RJ-MCMC \citep{brooks2003efficient,fearnhead2006exact}. In these scenarios, AR and ADA are not competitive, as already illustrated by \cite{hadj2020bayesian}. 

\subsection{Misspecified model}

\noindent\textbf{Scenario 1.b (t-Student innovations).} Data-generating mechanism is identical to Scenario 1.a, except for $\varepsilon_t$ being Student-t distributed with $\nu=3$ degrees of freedom. \algoname, ANOM and OBMB segmentation coverage is essentially unaffected by the heavier tail signal (see Table E.1), while ANOM's and OBMB's Hausdorff and bias are increasing, with the latter being larger than $0.3$ (recall that $m=2$). Again, AP, ADA, and LSW exhibit the poorest performance.\\

\vspace{-0.3cm}\noindent\textbf{Scenario 1.c (non-stationary innovations).} Data-generating mechanism for the signal is the same as Scenario 1.a but $\varepsilon_t$ is a non-stationary process (defined by model M1 in Section 5 of  \cite{wu2024frequency}). Methods' performances are similar to the ones in Scenario 1.b (see Table E.1).\\

\vspace{-0.3cm}\noindent\textbf{Scenario 2.b (varying mean and change point locations).} We replicate Scenario 2.a, with the key difference that continuity in the mean signal is enforced; see Section E for details. Tables E.3 and E.4 in the Supplementary summarize the segmentation and parameter estimation accuracy. The relative performance between \algonamesp and ANOM remains consistent. Notably, segmentation accuracy is identical to that in Scenario 2.a—even though the model is effectively misspecified. The mean RMSE matches that of Scenario 2.a when the number of change points is low ($m \leq 3$), but deteriorates thereafter, with increases of up to $50\%$. Prediction RMSE remains unaffected.

In the next two scenarios, AP is correctly specified (the two scenarios are taken from \cite{davis2006structural}), while \algonamesp  and ANOM are misspecified. ADA and LSW are nonparametric.\\

\vspace{-0.3cm}\noindent\textbf{Scenario 4 (piece-wise autoregressive process).} \citep{davis2006structural}. The data-generating mechanism is detailed in Supplementary E.1 and the segmentation results are reported in Table E.7. In this setting, AP and ADA outperform \algoname, ANOM, and LSW in terms of coverage, Hausdorff distance and bias. However, \algonamesp performance is similar to LSW, only slightly worse than AP and ADA, and much better than ANOM (irrespective of the parametrization). \\

\vspace{-0.3cm}\noindent\textbf{Scenario 5 (slowly varying autoregressive process).} \citep{davis2006structural}
The time-varying parameter autoregressive process assumes no structural breaks, but a continuously changing spectrum over time. 
The data-generating mechanism is detailed in Supplementary E.1. As there are no change points, we compute the mean squared error between the continuous true log spectrum $f$ and the estimate $\hat{f}$ given by each model. \algonamesp approximates the time-varying spectrum better than ANOM and ADA, while the performance is comparable to AP.\\

\vspace{-0.3cm}\noindent\textbf{Scenario 6 (no-signal and no-change point).} The data-generating mechanism is a white noise process and no change-point are assumed. In this setting, \algonamesp estimates one false change-point in the one replicate out of 50 when $N_E=1$ and two replicates out of 50 when $N_E=2$.\\

\vspace{-0.3cm}In summary, in the misspecified scenarios, \algonamesp remains competitive. Its performance is slightly worse than AP and ADA in the piece-wise autoregressive scenario, but it is better than ADA with time-varying parameter. When only the noise process is misspecified, it remains the best performing model.

\section{Applications to real data}
\label{sec:appl}
\subsection{Multichannel EEG sleep data}\label{app:eeg_main}
Multichannel electroencephalogram (EEG) is often used to study brain activity during sleep stages, and it can provide insights into patients' neuropsychiatric conditions \citep{manoach2019abnormal}. The DREAMS database \citep{devuyst2005dreams} provides 8 excerpts of 30 minutes from the whole night recording for different subjects. The excerpts contain three EEG channels (FP1-A1, C3-A1 or CZ-A1, and O1-A1), and brain activities such as sleep spindles are annotated by two experts. Spindles are 
characterized by short rhythmic oscillations in the brain activity visible on an EEG during non-rapid eye movement sleep (NREM). They occur at an higher frequency range than normal NREM brain activity and are of interest in neuropsychiatric research since studies link their abnormal presence to neuropsychiatric conditions, such as schizophrenia, prion disease, and epilepsy \citep{ferrarelli2007reduced}. The development of automatic segmentation method is an active research area \citep{parekh2017multichannel}. 

In this analysis, we consider a sub-period of an excerpt of a patient chosen at random. Figure F.1 depicts the three series considered. We model the first-order differences to reduce the temporal trend, which is clearly visible. We consider MCPVS with the grid $\Omega$ defined according to the $100$ frequencies with the higher power spectrum in the periodogram, and $N_E=2$ (lowest MDL, comparing $N_E=1,2$ and the automatic procedure with $\bar{N}_E=6$). In addition, the prior probability for each frequency in the grid $\Omega$ is set to $\pi_k=0.01$, and the prior variance for $\beta_e^{(1)},\beta_e^{(2)}$ is set to $\sigma^2_0=1$. 
As a comparison, we also include in the analysis ADA-X \citep{bertolacci2022adaptspec} with $J=7$ basis functions. Finally, we also estimate the change points using the univariate methods (namely CPVS, AP, and ADA) independently on each series. AP is run with default parameters. Similar to the multivariate counterpart, $J=7$ basis functions are set for ADA and $N_E=2$ for CPVS. ANOM \citep{hadj2020bayesian} is not included as we were not able to get convergence runs ($100000$ iterations, $\lambda_{m}=10$, $\lambda_\omega=2$): every run resulted in a different segmentation. 

Figure~\ref{fig:sleep_cp} depicts the change points estimated by MCPVS (blue) and ADA-X (yellow). The segmentation estimated by MCPVS partially overlaps with the sleep spindles events marked by experts (see Figure \ref{fig:sleep_cp}), whereas ADA-X detects only one segment linked to spindle activity. Univariate methods (analysing each series separately) fail to consistently detect change points around the sleep spindles events (Figure F.2): AP doesn't detect any change point, while ADA can identify a single spindle activity only for the FP1-A1 channel.

Figure \ref{fig:sleep_spectra} illustrates the maximum a posteriori frequencies and relative intensities within each segment estimated by MCPVS. Both segments that overlap with the spindles have the same highest intensity frequency (in absolute value) across the three channels. This is consistent with the literature that suggests that spindles should occur at 12-16 Hz. Figure F.3 in the Supplementary depicts the posterior distribution $q(\gamma_e)$ in each segment. The figure suggests that there is little uncertainty around the selected frequency within each segment, with $95\%$ credible sets including at most four frequencies, with often a single frequency being selected with posterior probability larger that $0.95$.

\begin{figure}[!t]
\centering
\includegraphics[width=1\textwidth]{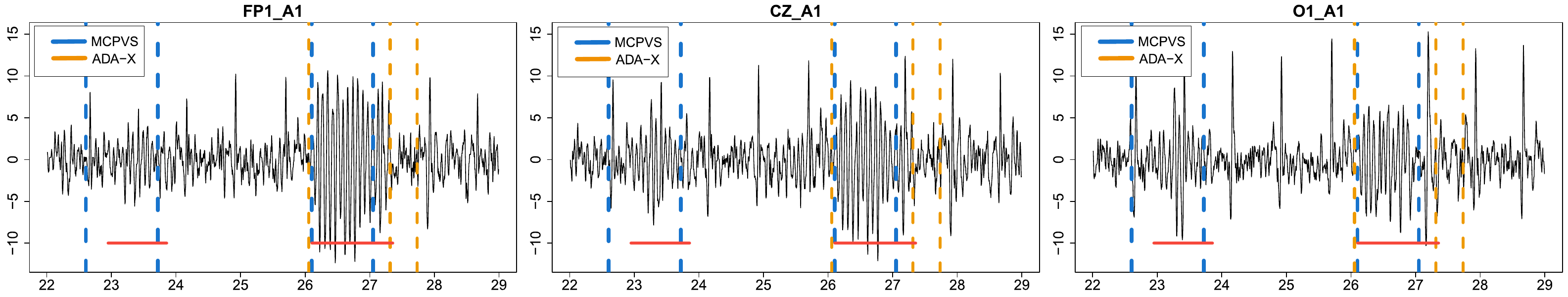}
\caption{{EEG sleep data: segmentation by MCPVS and ADA-X}. Segmentation provided by MCPVS and the ADA-X of \cite{bertolacci2022adaptspec}. Three panels refer to distinct EEG channels (FP1-A1, C3-A1, O1-A1). Black lines depict the first-order difference of EEG recordings. Vertical lines represent the change point estimated by MCPVS (blue) and ADA-X (yellow). Sleep spindles annotated by the experts are highlighted in red.}\label{fig:sleep_cp}
\end{figure}

\begin{figure}[!t]
\centering
{\includegraphics[width=.45\textwidth]{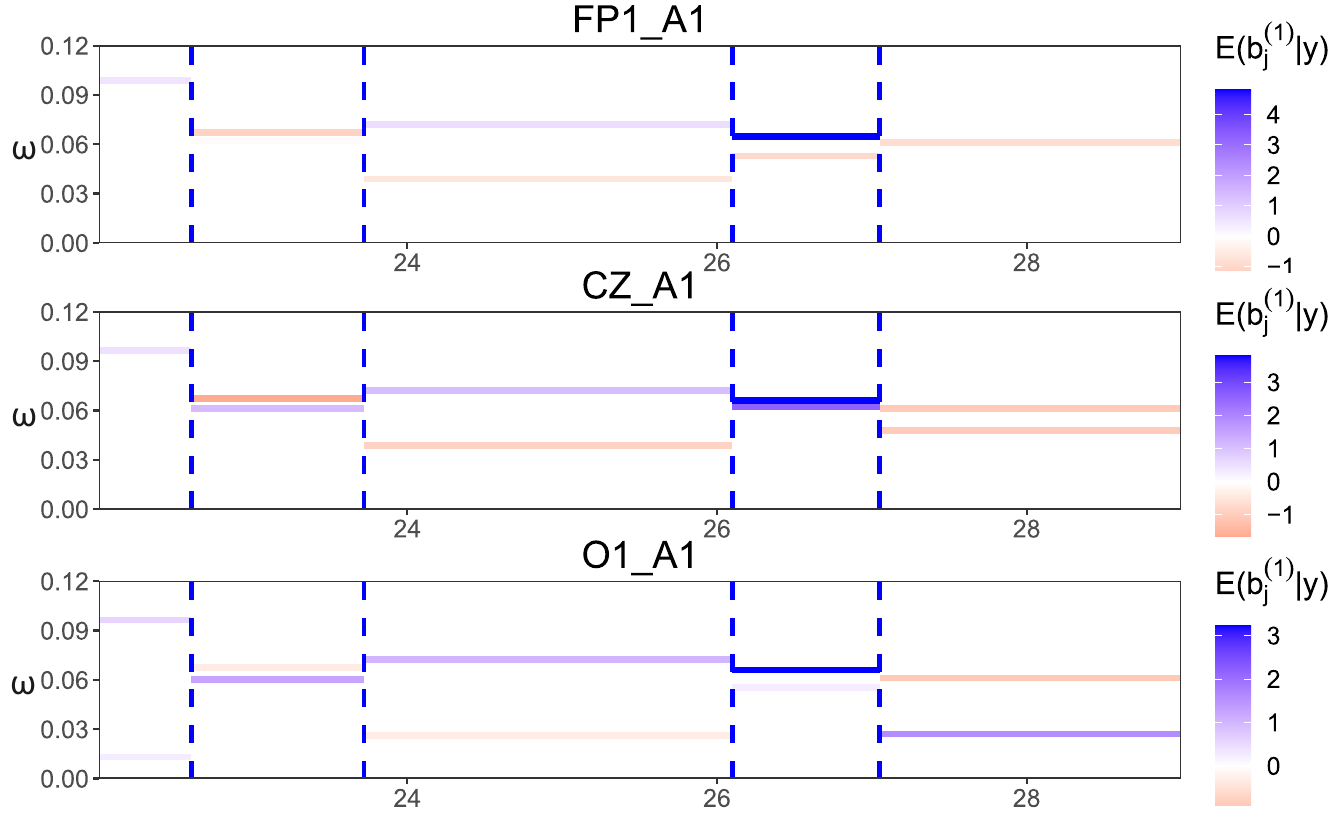}}
{\includegraphics[width=.45\textwidth]{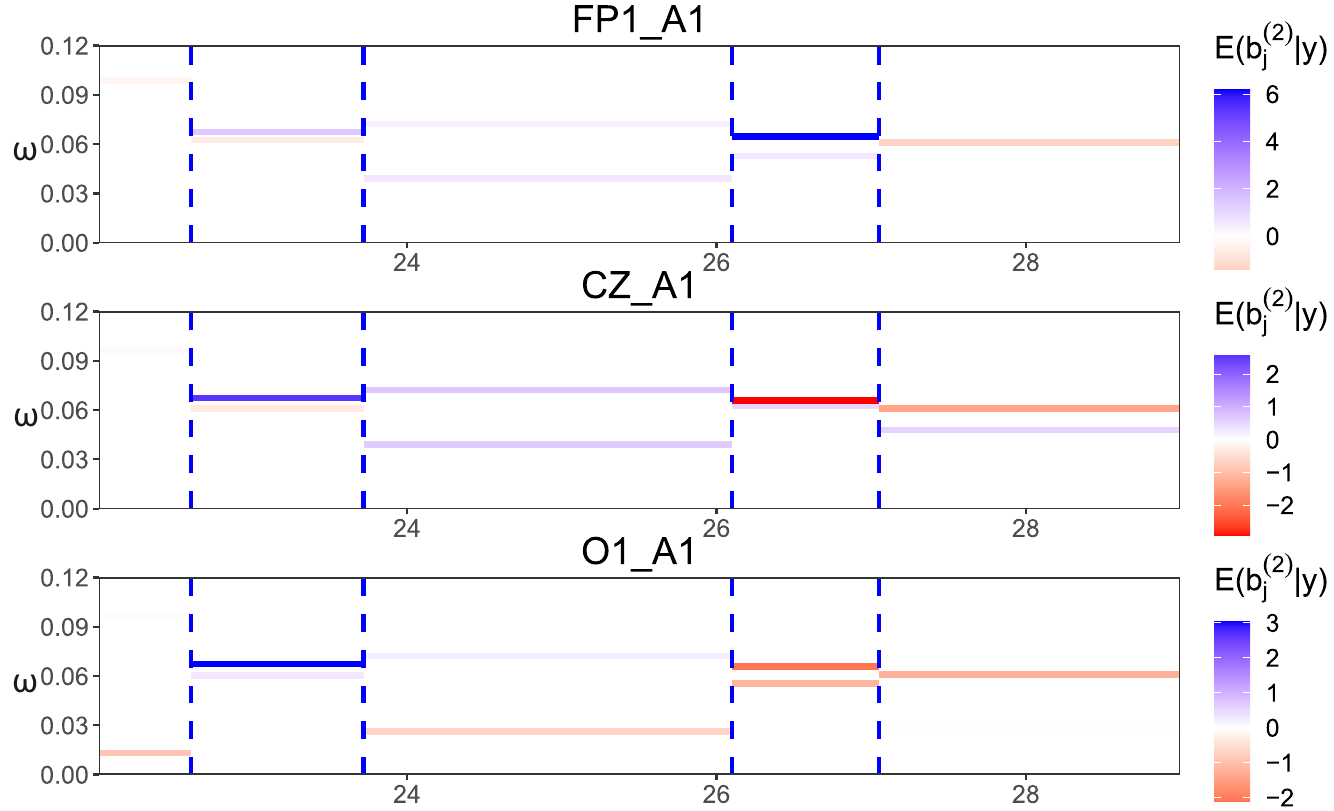}}
\caption{{EEG sleep data: estimated frequencies and intensities by segment ($N_E=2$)}. Vertical lines depict the change points estimated by MCPVS. Within each segment, horizontal bars depict the maximum a posteriori frequency, the colour gradient depicts the corresponding intensity.}\label{fig:sleep_spectra}
\end{figure}

\subsection{El Ni\~no-Southern Oscillation (ENSO)}\label{app:nino}
The El Ni\~no-Southern Oscillation (ENSO) is a periodic climate pattern linked to changes in the eastern tropical Pacific Ocean water temperature, and it affects the global weather. Detecting possible changes in ENSO dynamics is an active research topic. The National Oceanic and Atmospheric Administration (NOAA) and other researchers claim that changes in ENSO characteristics of ENSO are expected over time \citep[see, e.g.,][]{wang2017continued}, while other studies do not show evidence for that \citep[see][]{nicholls2008recent}. 

Part of the difficulty in studying ENSO is the large availability of different indexes that track ENSO dynamics. \cite{rosen2012adaptspec} studied El Ni\~no3.4 index, a commonly used index for studying ENSO measuring the sea surface temperature averaged across the region 5S-5N and 170-120W, and found no statistical evidence for a change in ENSO's frequency spectrum in the period from 1950 to 2011. However, \cite{trenberth2001indices} suggests that for a comprehensive understanding of the ENSO, El Ni\~no3.4 should be analysed jointly with a second index, the Trans-Ni\~no Index (TNI), which is given by the difference of the sea surface temperature between Ni\~no-1+2 and Ni\~no-4 regions. 
Here, we apply MCPVS to study the two indexes jointly.

The data is provided by the NSF National Center for Atmospheric Research (available for download at \url{https://climatedataguide.ucar.edu/climate-data}). The series consists of monthly observations spanning a period between January 1950 and November 2022. Moreover, we use the variables in the first difference. Data are depicted in Figure F.4. For the MCPVS setting, we consider both $N_E=1$ and $N_E=2$, and build $\Omega$ based on the $100$ frequencies with the higher power spectrum according to the periodogram. The prior probability for each frequency is set to $\pi_k=0.01$. We then apply the automatic procedure for the choice of $N_E$, which suggests $N_E=2$; therefore the following results are shown for this setting. Finally, the prior variance for the coefficients $\beta_e^{(1)},\beta_e^{(2)}$ is set to $\sigma^2_0=1$.

Figure~\ref{fig:nino_spectra2} depicts the estimated change points (dashed vertical line), the maximum a posteriori frequencies (horizontal bars) and intensities (colour gradient) within each segment. Our method estimates a change point in March 1996. TNI shows a change in the frequencies before and after March 1996, namely $\hat{\bm{\omega}}_1=(0.055,0.061)$ and $\hat{\bm{\omega}}_2 = (0.068,0.083)$. In contrast, the maximum a posteriori frequencies for El Ni\~no3.4 index do not change in the two segments ($\widehat{\bm{\omega}}_1 = \widehat{\bm{\omega}}_2 = (0.083,0.166)$). This is consistent with \cite{rosen2012adaptspec}, who analysed this index alone. However, we can observe a change in the posterior distribution of the intensities, see Figure F.5 where we compare the SuSiE posteriors $q(b_{je}^{(1)})$ and $q(b_{je}^{(2)})$ between the two segments $j=1,2$ (before and after the change point) and for $e=1,2$ (each component in SuSiE). The change in the intensity of the process is evident for $b_{j1}^{(1)}$ and $b_{j2}^{(2)}$ in the two segments $j=1,2$.

Figures F.6 and F.7 depict the posterior distribution $q(\gamma_{je})$ for the two segments of the estimated partition. For both El Ni\~no3.4 index and TNI, the plot suggests that \algonamesp selects one frequency for each effect with inclusion probability larger than $0.95$.

\begin{figure}[!ht]
\centering
{\includegraphics[width=.45\textwidth]{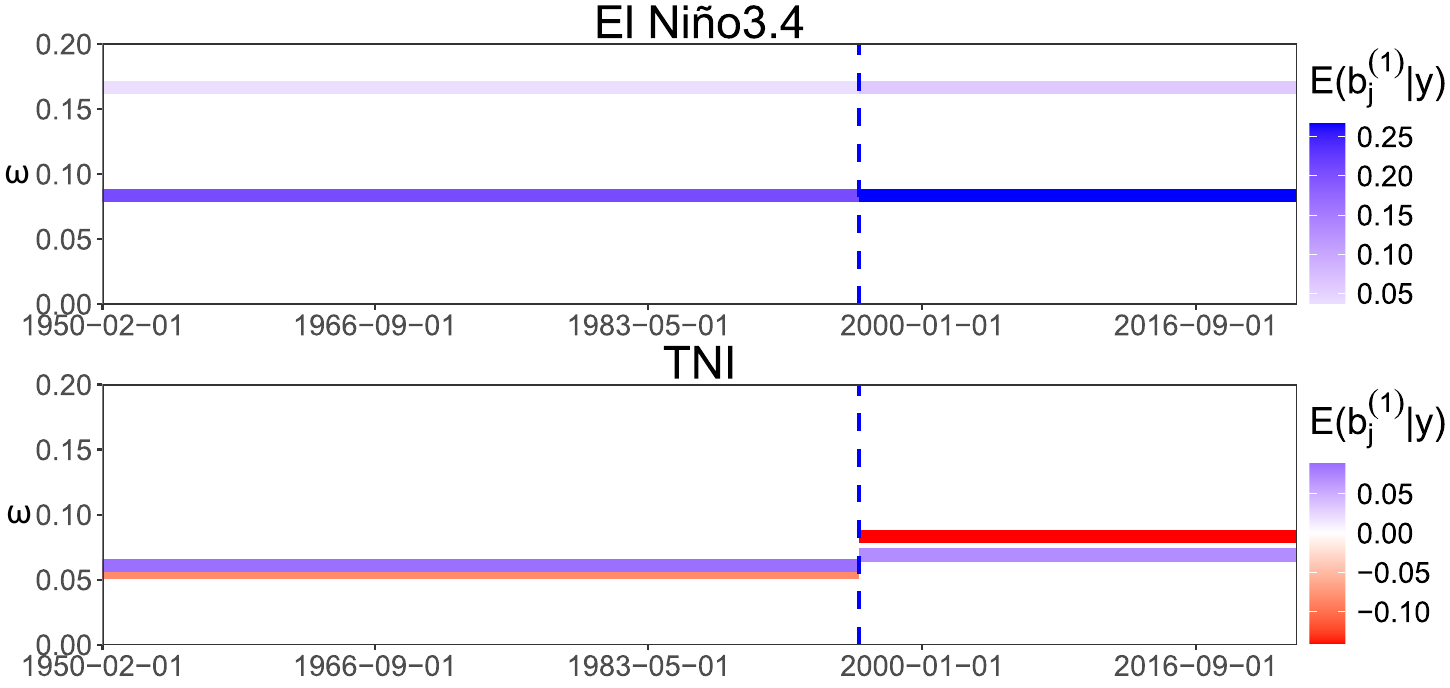}}
{\includegraphics[width=.45\textwidth]{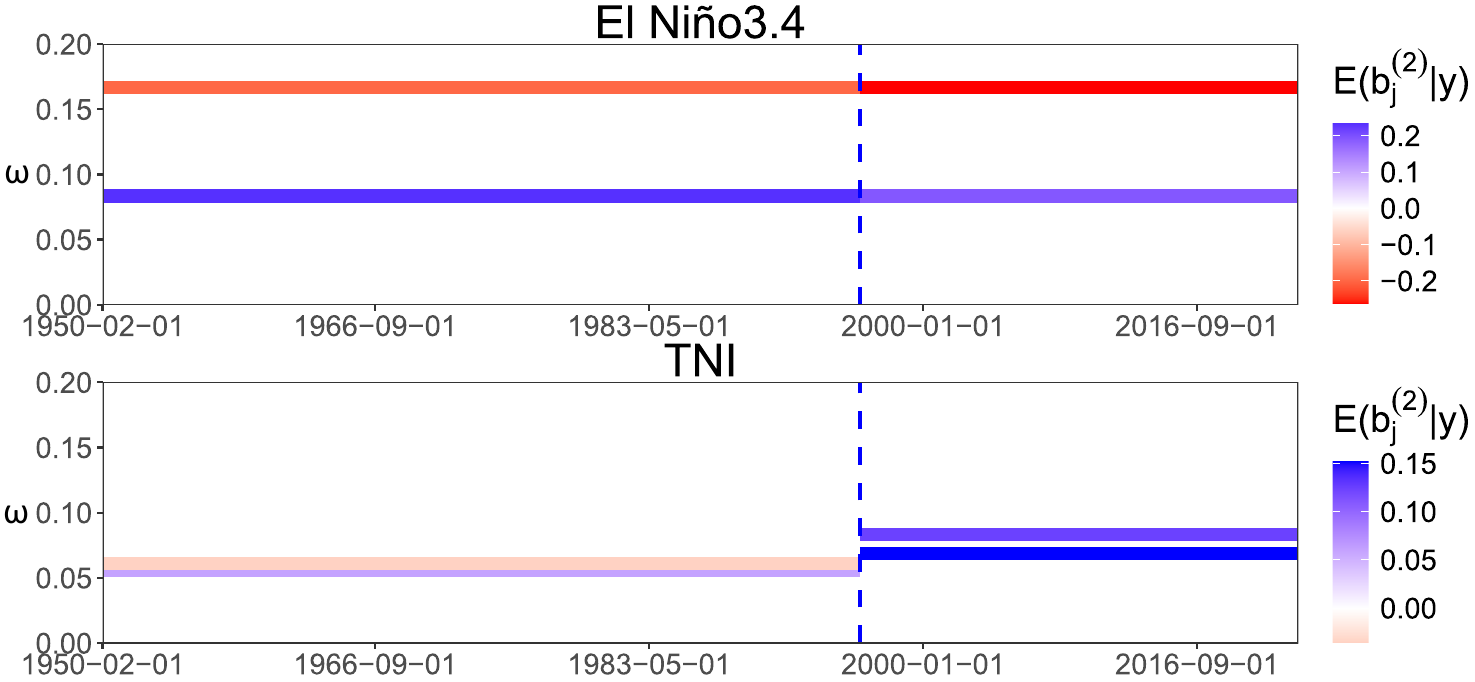}}
\caption{{ENSO data: estimated frequencies and intensities ($N_E=2$)}. Vertical lines depict the change points estimated by MCPVS. Within each segment, horizontal bars depicts the maximum a posterior frequency, the color gradient depicts the corresponding intensity.}\label{fig:nino_spectra2}
\end{figure}

\section{Continuous mean signal}\label{sec:cont}

Model \eqref{eq:model} does not assume continuity of $\mu_{it}$ at change points. In Section~\ref{sec:sim}, we showed that CVPS is capable of recovering the true change points with high accuracy, at least numerically, when continuity is enforced — the accuracy is in line with the simulation where we sample from model \eqref{eq:model} and do not impose this constraint. Both the sleep spindle and ENSO applications appear to exhibit continuous mean functions, suggesting it is natural to consider whether a change point method tailored to oscillatory series can be developed to explicitly enforce continuity.

Change point detection with continuity constraints has only recently gained attention, with studies focusing on continuous piecewise linear \citep{maidstone2017optimal,baranowski2019narrowest} and polynomial trends \citep{yu2022localising}. A key challenge is that continuity induces parameter dependencies across segments, potentially imposing constraints on the parameter space—as is the case in our model, which we illustrate below using a simplified version of model~\eqref{eq:model}.

\begin{proposition}\label{prop:continuity} Let $(Y_t)_{t\geq 1}$ be a sequence of independent random variables with $\mathbb{E}(Y_t)= (\beta^{(1)}\sin (2 \pi \omega_1 t) + \beta^{(2)} \cos (2 \pi\omega_1 t)) 1(t \leq t_0) + (\beta^{(1)}\sin (2 \pi\omega_2 t) + \beta^{(2)}\cos (2 \pi \omega_2 t)) 1(t > t_0)$. Then, if $\omega_1 \in (0,1/2)$, we have that $\mu_{t_0^-}=\mu_{t_0^+}$ if and only if $\omega_2\in \Omega_A \cup \Omega_B$ where
\begin{align*}
    \Omega_A &:= \left\{ \omega_1 + \frac{k}{t_0} : k \in \mathbb{Z},\ 0 < \omega_1 + \frac{k}{t_0} < \frac{1}{2} \right\}, \\ 
    \Omega_B&:=\left\{ \frac{1}{2} - \omega_1 + \frac{k}{t_0} : k \in \mathbb{Z},\ 0 < \frac{1}{2} - \omega_1 + \frac{k}{t_0} < \frac{1}{2} \right\}.
\end{align*}
\end{proposition}

Proposition~\ref{prop:continuity} examines a simplified version of model~\eqref{eq:model} in which the signal is one-dimensional and undergoes a single change in frequency. If we fix an arbitrary frequency in the left segment, only a finite set of frequencies in the right segment are compatible with the assumption of continuity in the mean signal. The same conclusion holds when the roles of the left and right segments are reversed. A further undesirable feature is that the cardinality of $\Omega_A \cup \Omega_B$ depends on the change point $t_0$; i.e., the state space of the parameter is a function of the change point. It is easy to see that this constraint can be reversed: if $\omega_1$ and $\omega_2$ were known, certain times $t$ could be ruled out as candidate change points. All these limitations are in place without even requiring smoothness at the point of change: no conditions are imposed on its first or second derivatives.

These observations suggest that enforcing continuity would necessitate a model other than \eqref{eq:model}, which is at the moment the standard choice for oscillatory signals in the literature~\citep{shumway2019time,hadj2020bayesian,wu2024frequency}. We leave this future work.

\section{Discussion}
\label{sec:concl}

The proposed procedure requires choosing a variable selection method and change point detection algorithm. Although we argue that SuSiE \citep{wang2020simple} is particularly suitable for the former task and OS \citep{kovacs2020optimistic} for the latter, the whole method is flexible to these choices. We settled on SuSiE mostly for its computational efficiency. However, we highlighted that the segmentation is sensitive to the number of variables to include in each segment. Future research will involve assessing robust strategies for this selection, or the use of alternative methods. For change point detection, we considered exact methods like dynamic programming, but there is a lack of pruning rules  \citep[\textit{e.g.},][]{killick2012optimal} when working with the marginal likelihood. Developing functional pruning for marginal likelihood seems a promising research question and could help in studying how much the localization rate can be improved for this class of models since dynamic programming can attain optimal rates.

There are important generalizations to consider. The first case, involving continuous signals, was discussed in the previous section. Our preliminary results suggest that model~\eqref{eq:model} exhibits undesirable features when continuity is assumed. Nevertheless, segmentation remains accurate in the continuous setting, at least numerically. One potential strategy to address this limitation is to use \algonamesp for segmentation, followed by estimation of a continuous mean signal given the detected change points. A second one is to consider generalized linear models. The extension is trivial as one would simply substitute the algorithm used for variable selection -- we are not aware of an implementation for SuSiE in GLMs, but there are alternatives \citep{chen2008bayesian,johnson2012bayesian,rossell2013high}.

Other classes of problems mimic the structure considered here: change point detection plus challenging (or computationally expensive) within-segment inference. For example, soil moisture dynamics exhibits recurring patterns of exponential decay \citep{gong2023changepoint}. In all these examples, researchers are interested both in the segmentation and the parameters within each segment. We believe that the ideas in this paper, in particular the simplification of the within-segment inference, will be valuable in those contexts as well.

\bibliographystyle{agsm.bst}
\bibliography{biblio_CPVS}

\clearpage
\appendix

\renewcommand\thefigure{\thesection.\arabic{figure}}    
\renewcommand\thetable{\thesection.\arabic{table}}    
\numberwithin{equation}{section}
\numberwithin{figure}{section}
\numberwithin{table}{section}

\begin{center}
    {\Large \textbf{Supplementary material for:\\}}
    \vspace{1em}
    {\LARGE \textbf{Computationally efficient segmentation for non-stationary time series with oscillatory patterns}}
\end{center}
\vspace{2em}

\noindent This supplementary material provides the derivation of the SuSiE posteriors under the sparsity constraints on the coefficients, the proof of Theorem 1, practical guidance for choosing the tuning parameters of the proposed procedure, some simulations' additional results and information, and more illustrations from the empirical analyses. An implementation of the proposed method is available as an R package for download at \url{https://github.com/whitenoise8/CPVS}

\setcounter{figure}{0}    
\setcounter{table}{0}

\section{Background: Bayesian linear regression}

Consider a Bayesian linear regression model 
\begin{align}
\mathbf{y}= \bm{X}\bm{b}  + \bm{\epsilon},
\end{align}
where $\bm{\epsilon}\sim N_n (\bm{0},\sigma^2 \bm{I}_n)$ and $\bm{b}\sim N_p (\bm{0},\sigma_0^2 \bm{I}_p)$. It is a standard result that the marginal likelihood can be written in closed form; see for example \cite{gelman2013bayesian}. We recall here the closed form expression as it will be used extensively. 

\begin{align} \label{eq:marg}
   p(\mathbf{y}|\mathbf{X},\sigma^2,\sigma^2_0)
    &= \int_{\mathbb{R}^{2}} (2\pi\sigma^2)^{-n/2}\exp\left\{-\frac{1}{2\sigma^2}(\mathbf{y}-\mathbf{X}\boldsymbol{b})^\intercal(\mathbf{y}-\mathbf{X}\boldsymbol{b})\right\} (2\pi\sigma^2_0)^{-1}\exp\left\{-\frac{1}{2\sigma^2_0}\boldsymbol{b}^\intercal\boldsymbol{b}\right\}\,d\boldsymbol{b} \nonumber \\
    &= (2\pi\sigma^2)^{-n/2}(\sigma^2_0)^{-1}\exp\left\{-\frac{1}{2\sigma^2}\mathbf{y}^\intercal\mathbf{y}\right\} \vert\mathbf{\Sigma}\vert^{1/2} \exp\left\{\frac{1}{2}\boldsymbol{\mu}^\intercal\mathbf{\Sigma}^{-1}\boldsymbol{\mu}\right\} \nonumber \\
    &= {(2\pi\sigma^2)^{-n/2}(\sigma^2_0)^{-1}\exp\left\{-\frac{1}{2\sigma^2}\mathbf{y}^\intercal\mathbf{y}\right\} \vert\mathbf{\Sigma}\vert^{1/2} \exp\left\{\frac{1}{2}\mathbf{y}^\intercal \, \mathbf{X} \, \mathbf{\Sigma}\, \mathbf{X}^\intercal \, \mathbf{y}\right\}},
\end{align}
where $\boldsymbol{\mu}$ and $\mathbf{\Sigma}$ are the posterior mean and covariance of $\boldsymbol{b}$, also available in closed form:
\begin{equation*}
    \mathbf{\Sigma}=\left(\frac{1}{\sigma^2}\mathbf{X}^\intercal\mathbf{X}+\frac{1}{\sigma^2_0}\mathbf{I}_p\right)^{-1}, \qquad \boldsymbol{\mu}=\frac{1}{\sigma^2}\mathbf{\Sigma}\mathbf{X}^\intercal\mathbf{y}.
\end{equation*}

\section{SuSiE posteriors under sparsity constraints}
\label{appendix:group_susie}

The Sum of Single Effect (SuSiE) model introduced by \cite{wang2020simple} is defined as the sum of $N_E$ baseline models, labelled Single Effect regression (SER), each of them designed to do variable selection of one regressor. In the paper, we described a simple variation of \cite{wang2020simple} SER model where variable selection involves two covariance matrices, $\bm{X}^{(1)}$ and $\bm{X}^{(2)}$. We describe it in more details below. Throughout the section, we ignore the change-point detection task: we assume to be in a generic segment as in (5) and in the following we drop the indices $i$, $j$, and $I_j$, and $n=|I_j|$ to lighten the notation.

The SER model corresponds to (7)--(9) with $N_E=1$. It is a linear regression model
\begin{equation*}
\mathbf{y}= \bm{X}^{(1)}\bm{b}^{(1)} + \bm{X}^{(2)}\bm{b}^{(2)}  + \bm{\epsilon},
\end{equation*}
with $\bm{\epsilon}\sim N_n (\bm{0},\sigma^2 \bm{I}_n)$ and the following prior distributions
\begin{alignat}{3}\label{eq:ser}
	\bm{b}^{(1)} &= \bm{\gamma}b^{(1)}, \qquad &&\bm{b}^{(2)} = \bm{\gamma}b^{(2)}, \\
	\bm{\gamma} &\sim Mult_p(1,\bm{\pi}), \qquad &&\bm{b} \sim N_2(0,\sigma^2_0\mathbf{I}_2),
\end{alignat}
where {$\bm{b} = (b^{(1)},b^{(2)})^\intercal$ and} $\bm{\gamma}\in\{0,1\}^p$ is p-vector indicator variable with one non-zero entry describing which regressor to include in the model,  $\sigma^2_0>0$ and $\bm{\pi}\in (0,1)^p$ are fixed hyperparameters and represent the prior variance of the coefficients and the prior inclusion probability for each covariate. The only difference with respect to \cite{wang2020simple} is that the latent indicators $\bm{\gamma}$ are the same for $b^{(1)}$ and $b^{(2)}$ to enforce {\it group} selection (see (6)). {Note that the two coefficients  $(b^{(1)},b^{(2)})$ are not necessarily identical.}

Let $\mathbf{x}^{(1)}_{k}$ and $\mathbf{x}^{(2)}_{k}$ denote the $k$-th column vector of $\mathbf{X}^{(1)}$ and $\mathbf{X}^{(2)}$ respectively, and $\mathbf{X}_k = [\mathbf{x}^{(1)}_{k},  \mathbf{x}^{(2)}_{k}]$. The posterior distribution $p(\bm{\gamma}, \bm{b}| \mathbf{y},\sigma^2_0,\sigma^2)$ factorizes as:
\begin{align}\label{eq:post1}
\bm{b}|\gamma_k=1,\mathbf{y},\mathbf{X}_k,\sigma^2,\sigma^2_0&\sim N_2(\boldsymbol{\mu}_k,\mathbf{\Sigma}_k),\\
 \boldsymbol{\gamma}|\mathbf{y},\mathbf{X}_k,\sigma^2,\sigma^2_0 &\sim Mult(1,\boldsymbol{\alpha}),
\end{align}
where:
\begin{equation*}
    \mathbf{\Sigma}_k=\left(\frac{1}{\sigma^2}\mathbf{X}_k^\intercal\mathbf{X}_k+\frac{1}{\sigma^2_0}\mathbf{I}_2\right)^{-1}, \qquad \boldsymbol{\mu}_k=\frac{1}{\sigma^2}\mathbf{\Sigma}_k\mathbf{X}_k^\intercal\mathbf{y}.
\end{equation*}
Then, the posterior inclusion probabilities (PIP) are defined as:
\begin{align}\label{eq:alpha_PIP}
\alpha_k=\mathbb{P}(\gamma_k=1|\mathbf{y},\mathbf{X}_{k},\sigma^2,\sigma^2_0 ) = \frac{\pi_k\,p(\mathbf{y}|\mathbf{X}_{k},\sigma^2,\sigma^2_0)}{\sum_{s=1}^p \pi_s \,p(\mathbf{y}|\mathbf{X}_{s},\sigma^2,\sigma^2_0)}.
\end{align}

The numerator is the marginal likelihood of the linear model with the $k$-th covariate and is available in closed form {(see \eqref{eq:marg}).}

The iterative algorithm proposed by \cite{wang2020simple} to fit SuSiE solves at each iteration $N_E$ SER model, fitting them with the model residuals from the previous steps rather than $\mathbf{y}$. {The modification described here does not affect the feasibility of the algorithm as the SER model remains conjugate.} We refer to \cite{wang2020simple} for more details on the construction of the residuals and the exact algorithm. 

\section{Proofs}
\label{app:theo_proof}

\subsection{Notation}

{\textit{Mathematical notation.}} For two functions $f$ and $h$, we will make use of the standard $\mathcal{O}$ and $o$ notation.  We write $f(T) = \mathcal{O}(h(T))$ to mean that there exists  $M > 0$ and $T^* > 0$ such that for all $T > T^*$:
\begin{equation*}
    |f(T)| \leq M h(T).
\end{equation*}
Similarly, we also use $f(T) = o(h(T))$ to mean:
\begin{equation*}
    \lim_{T\to\infty} \frac{f(T)}{h(T)} = 0.
\end{equation*}
Last, we use $f(T) \gtrsim h(T)$ to indicate that there exists $M > 0$ so that for $T$ large enough we have $h(T) \leq M f(T).$ For $\mathbf{x}\in\mathbb{R}^n$ and $p \geq 1$, we denote the $p$-norm $\lVert \mathbf{x} \rVert_p := (\sum_{i=1}^d |x_i|^p)^{1/p}$ and $\infty$-norm as $\lVert \mathbf{x} \rVert_\infty := \max_{1\leq i \leq n}  |x_i|$. For a matrix $A\in\mathbb{R}^{r\times c}$, we denote the the $\infty$-norm $\lVert A \rVert_\infty := \max_{1\leq i \leq r} \sum_{j=1}^c |a_{ij}|$, which is the maximum absolute row sum of the matrix. We write $A=\mathcal{O}(T)$ to say that the entries of $A$ are $\mathcal{O}(T)$.  We denote the indicator function as $1(\cdot)$.\\

\textit{Notations and definitions of terms used extensively in the proof.} Recall that $\Omega$ is the discrete set of frequencies we use as input for our method. Given $\Omega$, we defined $\mathbf{X}=[\mathbf{X}^{(1)},\mathbf{X}^{(2)}]$, the design matrix obtaining by merging the matrices with entries the Fourier basis defined in Section 2.1.

We denote by $\mathcal{M}=\mathcal{P}(\Omega)$ the set of the $2^{p}$ possible linear regression models that the frequencies in $\Omega$ induce, and let $m$ be an elements of $\mathcal{M}$, and $\bm{X}_{m}$ the design matrix corresponding to model $m$, \textit{i.e.}, it is obtained taking a subset of the columns of $\mathbf{X}$ corresponding to the frequencies $m$. We denote by $\mathbf{m}$ the pair $(m_1,m_2)\in \mathcal{M}\times \mathcal{M}$, where $m_1,m_2 \subseteq \Omega$ refer to a set of frequencies to consider in the left and right segments; \textit{i.e.}, which variables to include in each segment.

The true model is denoted by $\mathbf{m}_0=(m_{01},m_{02})$. Let $\mathbf{X}_{m_{01}}$ and $\mathbf{X}_{m_{02}}$ be the corresponding design matrices, $\mathbf{x}_{tm_{01}}$ and $\mathbf{x}_{tm_{02}}$ row vectors of these matrices. Without loss of generality, $\mathbf{X}_{m_{01}}$ and $\mathbf{X}_{m_{02}}$ are assumed to be constructed following a precise column ordering. The first $2|m_{01}\cap m_{02}|$ columns of  $\mathbf{X}_{m_{01}}$ and $\mathbf{X}_{m_{02}}$ are the sine and cosine basis of the frequencies in $m_{01}\cap m_{02}$. $I_0$ is the index set corresponding these columns and $\mathbf{X}_{I_0}$ the corresponding sub-matrix. %By construction  $\mathbf{X}_{I_0m_{01}} = \mathbf{X}_{I_0m_{02}}$. 
 Then, $\mathbf{X}_{m_{01}}$ has $2|m_{01}\setminus m_{02}|$ extra columns, corresponding to the sine and cosine basis with frequencies that are in $m_{01}$ but not in $m_{02}$. $I_1$ is the corresponding index set. Similarly, we can construct $I_2$. If $I_i=\emptyset$, for $i=0,1,2$, we set $\mathbf{X}_{I_i}=0$. Note that having same column ordering in $\mathbf{X}_{m_{01}}$ and $\mathbf{X}_{m_{02}}$ is not strictly necessary, it makes the notation more compact as there is a common index shared across  the two segments.

The true intensities are $\bbe_{1}$ and $\bbe_2$. The index sets $I_0,I_1$ and $I_2$ apply to these vectors as well: $\bbe_{I_1 1}$ is the vectors of frequencies associated to sine and cosine basis of frequencies unique to model $m_{01}$. Similarly, if $I_i=\emptyset$, $\bbe_{I_i}=0$. We denote $\beta_*:=\max\{||\bbe_1||_{\infty},||\bbe_2||_{\infty}\}$. $\beta_*<\infty$ by Assumption 1. We denote by $d:=\max\{2|m_{01}|, 2| m_{02}|\}$.

\subsection{Trigonometric identities}

The following lemma summarizes several trigonometric identities we use in the proofs. We believe most of them should be known, we added a proof when we could not find a reference for them.

\begin{lemma}\label{lemma:trig_approx} Let $\delta^-=\omega_1-\omega_2$ and $\delta^+=\omega_1+\omega_2$, the following holds:
\begin{align*}
    &1.\,\sum_{t=1}^n \sin^2(\omega t) = \frac{n}{2} + \frac{1}{4} - \frac{\sin ((2n+1)\omega)}{4\sin(\omega)} \\
    &2.\,\sum_{t=1}^n \cos^2(\omega t)  =\frac{n}{2} - \frac{1}{4} + \frac{\sin ((2n+1)\omega)}{4\sin(\omega)} \\
    &3.\,\sum_{t=1}^n \sin(\omega t)\cos(\omega t) = \frac{\cos(\omega)-\cos((2n+1)\omega)}{4\sin(\omega)} \\
    &4.\,\sum_{t=1}^n \sin  \omega_1 t \sin  \omega_2 t = \frac{1}{4}\left(\frac{\sin ((n+1/2)\delta^-)}{\sin(\delta^-/2)}-\frac{\sin ((n+1/2)\delta^+)}{\sin(\delta^+/2)}\right)\\
    &5.\,\sum_{t=1}^n \cos  \omega_1 t \cos  \omega_2 t = \frac{1}{4}\left(\frac{\sin ((n+1/2)\delta^-)}{\sin(\delta^-/2)}+\frac{\sin ((n+1/2)\delta^+)}{\sin(\delta^+/2)}-2\right) \\
    &6.\,\sum_{t=1}^n \sin  \omega_1 t \cos  \omega_2 t = \frac{1}{4}\left(\frac{\cos (\delta^+/2)-\cos ((n+1/2)\delta^+)}{\sin(\delta^+/2)}+\frac{\cos (\delta^-/2)-\cos ((n+1/2)\delta^-)}{\sin(\delta^-/2)}\right) \\
    &7.\,\sum_{t=1}^n \cos  \omega_1 t \sin  \omega_2 t = \frac{1}{4}\left(\frac{\cos (\delta^+/2)-\cos ((n+1/2)\delta^+)}{\sin(\delta^+/2)}-\frac{\cos (\delta^-/2)-\cos ((n+1/2)\delta^-)}{\sin(\delta^-/2)}\right).
\end{align*}
\end{lemma}
To prove Lemma \ref{lemma:trig_approx}, we need some well-known trigonometric results. In particular, recall the power-reduction formulas (a-b) and the double-angle formula (c):
\begin{align*}
    a)\,\sin^2(\omega t) = \frac{1-\cos(2\omega t)}{2}, \qquad b)\,\cos^2(\omega t) = \frac{1+\cos(2\omega t)}{2}, \qquad
    c)\,\sin(2\omega t) = 2\sin(\omega t)\cos(\omega t),
\end{align*}
and Lagrange's trigonometric identities:
 \begin{align*}
     \sum_{t=0}^n \sin(\omega t) &= \frac{\cos(\omega/2)-\cos((n+1/2)\omega)}{2\sin(\omega/2)}, \qquad
     \sum_{t=0}^n \cos(\omega t) = \frac{\sin(\omega/2)+\sin((n+1/2)\omega)}{2\sin(\omega/2)},
 \end{align*}
 or, analogously, starting from $t=1$:
 \begin{align*}
     \sum_{t=1}^n \sin(\omega t) &= \frac{\cos(\omega/2)-\cos((n+1/2)\omega)}{2\sin(\omega/2)}, \qquad
     \sum_{t=1}^n \cos(\omega t) = \frac{\sin(\omega/2)+\sin((n+1/2)\omega)}{2\sin(\omega/2)} -1.
 \end{align*}
\begin{proof} We start with 1-2:
\begin{align*}
    1.\quad \sum_{t=1}^n \sin^2(\omega t) &= \sum_{t=1}^n \frac{1-\cos(2\omega t)}{2} = \frac{n}{2}-\frac{1}{2}\sum_{t=1}^n \cos(2\omega t) \qquad\text{(power-reduction formula)}\\
    &=\frac{n}{2}-\frac{1}{2}\left( \frac{\sin(\omega)+\sin((2n+1)\omega)}{2\sin(\omega)} -1\right) \qquad\text{(Lagrange's trigonometric identity)}\\
    &=\frac{n}{2}+\frac{1}{4}-\frac{\sin((2n+1)\omega)}{4\sin(\omega)} \\
    2.\quad \sum_{t=1}^n \cos^2(\omega t) &= \sum_{t=1}^n \frac{1+\cos(2\omega t)}{2} = \frac{n}{2}+\frac{1}{2}\sum_{t=1}^n \cos(2\omega t) \qquad\text{(power-reduction formula)}\\
    &=\frac{n}{2}+\frac{1}{2}\left( \frac{\sin(\omega)+\sin((2n+1)\omega)}{2\sin(\omega)} -1\right) \qquad\text{(Lagrange's trigonometric identity)}\\
    &=\frac{n}{2}-\frac{1}{4}+\frac{\sin((2n+1)\omega)}{4\sin(\omega)},
\end{align*}
and 3:
\begin{align*}
    3.\quad \sum_{t=1}^n \sin(\omega t)\cos(\omega t) &= \frac{1}{2}\sum_{t=1}^n \sin(2\omega t) &&\qquad\text{(double-angle formula)}\\
    &=\frac{1}{2}\left(\frac{\cos(\omega)-\cos((2n+1)\omega)}{2\sin(\omega)} \right) &&\qquad\text{(Lagrange's trigonometric identity)}\\
    &=\frac{\cos(\omega)-\cos((2n+1)\omega)}{4\sin(\omega)}.
\end{align*}
Now, in order to prove 4-7, we need some additional identities, known as product-to-sum formulas. Denote with $\delta^-=\omega_1-\omega_2$ and $\delta^+=\omega_1+\omega_2$, it holds:
 \begin{align*}
     (i)\quad\sum_{t=1}^n \sin  \omega_1 t \sin  \omega_2 t &= \sum_{t=1}^n \frac{\cos((\omega_1 -\omega_2) t)-\cos((\omega_1+\omega_2) t)}{2} = \sum_{t=1}^n \frac{\cos(\delta^- t)-\cos(\delta^+ t)}{2} \\
    (ii)\quad\sum_{t=1}^n \cos  \omega_1 t \cos  \omega_2 t &= \sum_{t=1}^n \frac{\cos((\omega_1 -\omega_2) t)+\cos((\omega_1+\omega_2) t)}{2} = \sum_{t=1}^n \frac{\cos(\delta^- t)+\cos(\delta^+ t)}{2} \\
       (iii)\quad\sum_{t=1}^n \sin  \omega_1 t \cos  \omega_2 t &= \sum_{t=1}^n \frac{\sin((\omega_1 +\omega_2) t)+\sin((\omega_1-\omega_2) t)}{2} = \sum_{t=1}^n \frac{\sin(\delta^+ t)+\sin(\delta^- t)}{2}\\
    (iv)\quad\sum_{t=1}^n \cos  \omega_1 t \sin  \omega_2 t &= \sum_{t=1}^n \frac{\sin((\omega_1 +\omega_2) t)-\sin((\omega_1-\omega_2) t)}{2} = \sum_{t=1}^n \frac{\sin(\delta^+ t)-\sin(\delta^- t)}{2}.
 \end{align*}
We are now ready to prove the identities 4-7 in Lemma \ref{lemma:trig_approx}:
 \begin{align*}
     4.\quad\sum_{t=1}^n \sin  \omega_1 t \sin  \omega_2 t &= \sum_{t=1}^n \frac{\cos(\delta^- t)-\cos(\delta^+ t)}{2} \hspace{3.5cm}\text{(product-to-sum formula $(i)$)}\\
     &= \frac{1}{2}\underbrace{\left(\frac{\sin(\delta^-/2)+\sin((n+1/2)\delta^-)}{2\sin(\delta^-/2)}-\frac{\sin(\delta^+/2)+\sin((n+1/2)\delta^+)}{2\sin(\delta^+/2)}\right)}_{\text{(Lagrange's trigonometric identity)}} \\
     &= \frac{1}{4}\left(\frac{\sin((n+1/2)\delta^-)}{\sin(\delta^-/2)}-\frac{\sin((n+1/2)\delta^+)}{\sin(\delta^+/2)}\right);\\
     5.\quad\sum_{t=1}^n \cos  \omega_1 t \cos  \omega_2 t &= \sum_{t=1}^n \frac{\cos(\delta^- t)+\cos(\delta^+ t)}{2} \hspace{3.5cm}\text{(product-to-sum formula $(ii)$)}\\
     &= \frac{1}{2}\underbrace{\left(\frac{\sin(\delta^-/2)+\sin((n+1/2)\delta^-)}{2\sin(\delta^-/2)}+\frac{\sin(\delta^+/2)+\sin((n+1/2)\delta^+)}{2\sin(\delta^+/2)}-2\right)}_{\text{(Lagrange's trigonometric identity)}} \\
     &= \frac{1}{4}\left(\frac{\sin((n+1/2)\delta^-)}{\sin(\delta^-/2)}+\frac{\sin((n+1/2)\delta^+)}{\sin(\delta^+/2)}-2\right);\\
     6.\quad\sum_{t=1}^n \sin  \omega_1 t \cos  \omega_2 t &= \sum_{t=1}^n \frac{\sin(\delta^+ t)+\sin(\delta^- t)}{2} \hspace{3.25cm}\text{(product-to-sum formula $(iii)$)}\\
     &= \frac{1}{4}\underbrace{\left(\frac{\cos(\delta^+/2)-\cos((n+1/2)\delta^+)}{\sin(\delta^+/2)}+\frac{\cos(\delta^-/2)-\cos((n+1/2)\delta^-)}{\sin(\delta^-/2)}\right)}_{\text{(Lagrange's trigonometric identity)}}; \\
     7.\quad\sum_{t=1}^n \cos  \omega_1 t \sin  \omega_2 t &= \sum_{t=1}^n \frac{\sin(\delta^+ t)-\sin(\delta^- t)}{2} \hspace{3.25cm}\text{(product-to-sum formula $(iv)$)}\\
     &= \frac{1}{4}\underbrace{\left(\frac{\cos(\delta^+/2)-\cos((n+1/2)\delta^+)}{\sin(\delta^+/2)}-\frac{\cos(\delta^-/2)-\cos((n+1/2)\delta^-)}{\sin(\delta^-/2)}\right)}_{\text{(Lagrange's trigonometric identity)}}.
 \end{align*}
\end{proof}

Given $\Omega$, we will define 
	\begin{align}\label{eq:c_omega}
	C_\Omega:=\max_{\omega\in \Omega, (\omega_1, \omega_2)\in \Omega \times\Omega}&\left\{ \sum_{j=1}^t \sin^2  \omega j -\frac{t}{2}, \sum_{j=1}^t \cos^2  \omega j -\frac{t}{2}, \sum_{j=1}^t \sin  \omega j \cos  \omega j, \right\}  \cup \\
&\left\{	\sum_{j=1}^t \sin  \omega_1 j \cos  \omega_2 j  ,\sum_{j=1}^t \cos  \omega_1 j \cos  \omega_2 j ,\sum_{j=1}^t \sin  \omega_1 j \sin  \omega_2 j \right\}.
	\end{align} 
Under Assumption 1, $|\Omega|<\infty$ and so it is easy to check that $C_\Omega<\infty$. 

\subsection{Sub-Gaussian Random Variables}

\begin{definition}\label{def:subgaus}
	A random variable $X\in\mathbb{R}$ is sub-Gaussian with mean $\mu=\mathbb{E}(X)$  if there exists a positive number $\sigma$ such that the following holds:
	\begin{equation*}
		\mathbb{E}(\exp(tX)) \leq \exp(\sigma^2t^2/2), \forall t\in\mathbb{R}.
	\end{equation*}
\end{definition}

We say that $X$ sub-Gaussian with parameter $\sigma$ and denote it by $X\sim SG(\sigma)$. Following Definition~\ref{def:subgaus}, if $X\sim SG(\sigma)$, then $X\sim SG(\sigma')$ for all $\sigma'>\sigma$. There are a few results for sub-Gaussian random variables that we will use extensively. First, if $X\sim SG(\sigma)$, then $a X\sim SG(|a| \sigma)$ with $a\in\mathbb{R}$. Also, sub-Gaussianity is preserved in linear combinations. Let $X_1\sim SG(\sigma_1)$ and $X_2 \sim SG(\sigma_2)$ be independent random variables, then
\begin{equation}\label{eq:sub_combination}
X_1 + X_2 \sim SG\left(\sqrt{\sigma_1^2+\sigma_2^2}\right).
\end{equation}
Last, for a sub-Gaussian random variable, we can use the Chernoff bound:
\begin{equation}\label{eq:chernoff}
	\mathbb{P}\left(|X-\mu| >t \right)\leq \exp \left\{-\frac{t^2}{2  \sigma^2}\right\}.
\end{equation}
See \cite{Boucheron_etal2013book} for a reference on the topic.

\begin{lemma}\label{lemma:event_bounds} Suppose that Assumption 1 holds and you are given $\Omega$ and the corresponding $\mathbf{X}$. Let $(\epsilon_t)_{t \geq 1}$ defined by $\epsilon_t=Y_t - \mathbb{E}(Y_t)$. Then for some constant $C>0$ for all $i=1,\ldots,2p$  we have that 
	\[
	\mathbb{P} \left( \left\vert \sum_{j=1}^t \epsilon_j x_{ji} \right\vert < \sqrt{C t \log T} \,\,\,\, \forall 1\leq t \leq T \right) \geq  1- \frac{2}{T},
	\] 
		\[
	\mathbb{P} \left( \left\vert \sum_{j=t+1}^T \epsilon_j x_{ji} \right\vert < \sqrt{C (T-t+1) \log T} \,\,\,\, \forall 1\leq t \leq T \right) \geq  1- \frac{2}{T},
	\]
	and 
		\[
	\mathbb{P} \left( \left\vert \sum_{j=\min\{t_0,t\}}^{\max\{t_0,t\} -1}  \epsilon_j x_{ji} \right\vert < \sqrt{C| t -t _0| \log T} \,\,\,\, \forall 1\leq t \leq T \right) \geq  1- \frac{2}{T}. 
	\]
\end{lemma}

\begin{proof}
By Assumption 1, $(Y_t)_{t \geq 1}$ is a sub-Gaussian sequence of random variables with parameters $(\sigma_t)_{t\geq 1}$. Let $K:=\sup_t \sigma_t.$ By Assumption 1, $K<\infty$. By property of sub-Gaussian random variables, $(\epsilon_t)_{t \geq 1}$ is also a sub-Gaussian sequence of random variables with $\mathbb{E}[\epsilon_t]=0$ and parameters $(\sigma_t)_{t\geq 1}$. In addition, for all $t\geq 1$, it holds that: $\epsilon_t \sim SG(K)$, $\epsilon_t x_{ji} \sim SG(K |x_{ji}| )$, by~\eqref{eq:sub_combination}, $ \sum_{j=1}^t \epsilon_j x_{ji} \sim SG\left(K \sqrt{\sum_{i=1}^t x_{ji}^2}\right)$, and $\mathbb{E}(\sum_{j=1}^t \epsilon_j x_{ji})=0$. For $C>0$, we apply Chernoff bound
\[
\mathbb{P}\left(\left\vert \sum_{j=1}^t \epsilon_j x_{j}  \right\vert  > \sqrt{C t \log T}\right) \leq 2 \exp\left\{ - \frac{C t  \log T }{2  K^2 \sum_{i=1}^t x_j^2 }  \right \}\leq   2 \exp\left\{ - \frac{C   \log T }{ K^2 (1+2 C_{\Omega}) } \right \},
\]
where the last equality follows by Lemma~\ref{lemma:trig_approx} and $t \geq 1$. 
Then for $C>2K^2 (1+2C_{\Omega})$, we have that
\begin{align*}
	\mathbb{P} \left( \left\vert \sum_{j=1}^t \epsilon_j x_{ji} \right\vert < \sqrt{C t \log T} \,\,\,\, \forall 1\leq t \leq T \right) &= \mathbb{P} \left(\bigcap_{1\leq t \leq T} \left\{\left\vert \sum_{j=1}^t \epsilon_j x_{ji} \right\vert < \sqrt{C t \log T} \right\}\right) \\ 
	&\leq  \sum_{1\leq t \leq T} \mathbb{P} \left( \left\{\left\vert \sum_{j=1}^t \epsilon_j x_{ji} \right\vert < \sqrt{C t \log T} \right\}\right)\\
	&\leq \sum_{1\leq t \leq T}  \frac{2}{T^2} = \frac{2}{T}.
\end{align*}
The argument is identical for the other two events. 
\end{proof}

\subsection{Closed-form marginal likelihoods}
\label{subsec:marglik}

	In the proof we will make ample use of $\log p(\mathbf{y}|t,{\mathbf{m}})$, the marginal log-likelihood where we condition on a single candidate change point $t$ and a pair of models $\mathbf{m}$. We omit the conditioning on $\sigma^2$ and $\sigma^2_0$ and assume they are given. It decomposes as 
\begin{align*}
\log p(\mathbf{y}|t,{\mathbf{m}})=&\log p(\mathbf{y}_{1:t}|{m_1})+\log p(\mathbf{y}_{t+1:T}|{m}_2)\\
\end{align*}
where $\log p(\mathbf{y}_{1:t}|{m_1})$ and $\log p(\mathbf{y}_{t+1:T}|{m_2})$ can be computed in closed-form through \eqref{eq:marg}, {as $\log p(\mathbf{y}_{1:t}|{m_1})$ stands for $\log p(\mathbf{y}_{1:t}|{X_{m_1}})$}.\\

The marginal likelihood for the model considered here can be further simplified. Let us consider a generic model $m$ (design matrix $\mathbf{X}_m$). We can further simplify $\mathbf{\Sigma}_m$ in \eqref{eq:marg}. As $\sigma_0^2 \to \infty$, $\mathbf{\Sigma}_m \to \left(\frac{1}{\sigma^2}\mathbf{X}_m^\intercal\mathbf{X}_m\right)^{-1}$. Now, $\mathbf{X}_m^\intercal\mathbf{X}_m$ is a square matrix with diagonal entries equal to
\[
\text{diag}(\mathbf{X}_m^\intercal\mathbf{X}_m) =\left(\sum_{s=1}^t \sin^2 (2 \pi \omega_{1} s),\ldots, \sum_{s=1}^t \sin^2 (2 \pi \omega_{|m|} s), \sum_{s=1}^t\cos^2 (2 \pi \omega_{1} s),\ldots, \sum_{s=1}^t \cos^2 (2 \pi \omega_{|m|} s)\right),
\]
and off diagonal entries for $i\neq j$
\begin{align*}
[\mathbf{X}_m^\intercal\mathbf{X}_m]_{ij} &= \sum_{s=1}^t \sin (2 \pi \omega_{i} s) \sin (2 \pi \omega_{j} s) \,\ \,\  i\neq j  \,\ \,\ \text{and} \,\ \,\  i,j=1,\ldots |m|\\
[\mathbf{X}_m^\intercal\mathbf{X}_m]_{ij} &= \sum_{s=1}^t \cos (2 \pi \omega_{i} s) \cos (2 \pi \omega_{j} s) \,\ \,\  i\neq j  \,\ \,\ \text{and} \,\ \,\ i,j=|m|+1,\ldots 2|m| \\
[\mathbf{X}_m^\intercal\mathbf{X}_m]_{ij} &= \sum_{s=1}^t \sin (2 \pi \omega_{i} s) \cos (2 \pi \omega_{j} s) \,\ \,\ \text{otherwise}.
\end{align*}

By Lemma \ref{lemma:trig_approx}, the diagonal elements are $\mathcal{O}(t)$, while the off-diagonal ones $\mathcal{O}(1)$. Let $D=\mathbf{X}_m^\intercal\mathbf{X}_m-\frac{t}{2}\mathbf{I}$. $D$ is a symmetric matrix with entries $\mathcal{O}(1)$. We rewrite $\mathbf{\Sigma}_m$ as 
\[
\mathbf{\Sigma}_m=\sigma^2\left(\mathbf{X}_m^\intercal \mathbf{X}_m\right)^{-1}= \sigma^2 \left(\frac{t}{2}\mathbf{I} +D\right)^{-1}= \sigma^2 \left( \frac{t}{2}\mathbf{I} \left(\mathbf{I}  +\frac{2}{t}\mathbf{I} D\right)\right)^{-1}= \sigma^2 \frac{2}{t}  \left(\mathbf{I}  +\frac{2}{t} D\right)^{-1}.
\]
Now, the Neumann series of the matrix above is
\[
\left(\mathbf{I}  +\frac{2}{t} D\right)^{-1}=\sum_{k=0}^{\infty} \left(-\frac{2}{t} D\right)^k = \mathbf{I} -\frac{2}{t} D + Re(2),
\]
where $Re(2)$ is the remainder of order two of the Neumann series and $Re(2)=\mathcal{O}(t^{-2})$. For large $t$ enough and finite, $\lVert \frac{2}{t} D\rVert_{\infty}<1$, which implies convergence of the Neumann series. We rewrite $\Sigma_m$ as
\[
\mathbf{\Sigma}_m= \sigma^2  \frac{2}{t}\left(\mathbf{I} -\frac{2}{t} D + Re(2) \right) =  \sigma^2  \left( \frac{2}{t}\mathbf{I}-\frac{2}{t^2} D+\frac{2}{t}Re(2)\right)=   \frac{2\sigma^2}{t}\mathbf{I}+R_m,
\]
where $R_m=\mathcal{O}(t^{-2})$ by construction. The determinant can be rewritten as 
\[
|\mathbf{\Sigma}_m|^{1/2}=\left \vert\frac{2\sigma^2}{t}\mathbf{I}+R_m\right\vert^{1/2}= \left(\frac{2\sigma^2}{t}\right)^{|m_k|/2}\left \vert \mathbf{I}+\frac{2}{t} R_m\right\vert^{1/2}.
\]

The marginal likelihood can be rewritten as:
\begin{align*}
\log p(\mathbf{y}|t,{\mathbf{m}) \propto}& {    -\frac{|m_1|}{2}\log t  + \frac{1}{2} \log \left\vert \mathbf{I}+\frac{2}{t} R_{m_1}\right\vert +  \frac{\sigma^2}{t}\left\lVert \sum_{j=1}^{t} y_j \mathbf{x}_{j}\right\rVert_2^2 + \left(\sum_{j=1}^{t} y_j \mathbf{x}^\intercal_{j}\right) R_{m_1} \left(\sum_{j=1}^{t} y_j \mathbf{x}_{j}\right) } \nonumber \\
&-\frac{|m_2|}{2}\log (T-t+1)  + \frac{1}{2} \log \left\vert \mathbf{I}+\frac{2}{t} R_{m_2}\right\vert + \frac{\sigma^2}{T-t+1}\left\lVert\sum_{j=t+1}^{T} y_j \mathbf{x}_{j}\right\rVert_2^2 \nonumber \\
& + \left(\sum_{j=t+1}^{T} y_j \mathbf{x}_{j}^\intercal\right) R_{m_2} \left(\sum_{j=t+1}^{T} y_j \mathbf{x}_{j}\right).  
\end{align*}
where the above holds up to proportionality constant that do not depend on $t$ and $\mathbf{m}$. Note that, $|m_1|$ is not necessarily equal to $|m_2|$, hence, the dimension of vectors and matrices in the first row is not necessarily equal to that in the second row. In addition, the last terms in each row are dominated for large $t$ by the terms that precede them. Last, as $t\to \infty$, $\left \vert \mathbf{I}+\frac{2}{t} R\right\vert^{1/2}\to 1$, making those terms negligible for large $t$. In the proof, we will use the following:
\begin{align}\label{eq:marg_like}
	\log p(\mathbf{y}|t,{\mathbf{m}) \propto}& {    -\frac{|m_1|}{2}\log t  +  \frac{\sigma^2}{t}\left\lVert \sum_{j=1}^{t} y_j \mathbf{x}_{j}\right\rVert_2^2 -\frac{|m_2|}{2}\log (T-t+1)   + \frac{\sigma^2}{T-t+1}\left\lVert\sum_{j=t+1}^{T} y_j \mathbf{x}_{j}\right\rVert_2^2, }
\end{align}
{where we ignore lower order terms.}

\subsection{Proof of Theorem 1}

The statement
\[
	\max_{t {\in \{1,\ldots,T\}:} |t-t_0|> {C^*\log T}} G_{(0,T]}(t) < G_{(0,T]}(t_0) ,
\]
is equivalent to 
\begin{equation}\label{eq:statement}
	\max_{t {\in \{1,\ldots,T\}:} |t-t_0|> {C^*\log T}} \log p(\mathbf{y}|t) < \log p(\mathbf{y}|t_0),
\end{equation}
where 
\[
\log p(\mathbf{y}|t)=\sum_{{\mathbf{m}\in \mathcal{M} \times \mathcal{M}}}\log p(\mathbf{y}|t,{\mathbf{m}}) p({\mathbf{m}}), 
\]
where {$\log p(\mathbf{y}|t,\mathbf{m})$} is given in Section~\ref{subsec:marglik}, and $p({\mathbf{m}})$ is the prior probability of the model {implied by SuSiE}. SuSiE does not assign equal probability to each model, but since $|\mathcal{M}\times\mathcal{M}|$ is finite, we have that $p({\mathbf{m}})=\mathcal{O}(1)$. This follows directly by Assumption 1 \textit{(e)}.

The statement of theorem is thus equivalent to showing that that there exists a constant $C^*>0$ such that it holds that
	\[
	\lim_{T\to \infty} \mathbb{P}\left( \left\{ t \in\left\{1 \leq t \leq T : |t-t_0|> C^* \log T \right\}:   \frac{p(\mathbf{y}|t_0)}{p(\mathbf{y}|t)} > 0 \right\}	\right) =1,
	\] 
	which implies \eqref{eq:statement}.

The ratio of marginal likelihoods can be rewritten as follows:
\begin{equation}\label{eq:deco}
	\frac{p(\mathbf{y}|t_0)}{p(\mathbf{y}|t)}=\frac{\sum_{{\mathbf{m}}} p({\mathbf{m}}) p(\mathbf{y}|t_0,{\mathbf{m}})}{\sum_{{\mathbf{m}}} p({\mathbf{m}}) p(\mathbf{y}|t,{\mathbf{m}})}= \frac{ p(\mathbf{y}|t_0,{\mathbf{m}}_0)}{p(\mathbf{y}|t,{\mathbf{m}}_0)}\frac{\left(1+\sum_{{\mathbf{m}}\neq {\mathbf{m}}_0}\frac{ p(\mathbf{y}|t_0,{\mathbf{m}})p({\mathbf{m}})}{ p(\mathbf{y}|t_0,{\mathbf{m}}_0)p({\mathbf{m}}_0)}\right)}{\left(1+\sum_{{\mathbf{m}}\neq {\mathbf{m}}_0}\frac{ p(\mathbf{y}|t,{\mathbf{m}})p({\mathbf{m}})}{ p(\mathbf{y}|t,{\mathbf{m}}_0)p({\mathbf{m}}_0)}\right)},
\end{equation}
where ${\mathbf{m}}_0$ is the true model, $\sum_{{\mathbf{m}}\neq {\mathbf{m}}_0}$  denotes the sum over all pairs in $\mathcal{M}\times\mathcal{M}$ except ${\mathbf{m}}_0$. By Assumption 1 \textit{(d)}, ${\mathbf{m}}_0 \in \mathcal{M}\times\mathcal{M}$.

In equation \eqref{eq:deco}, we decomposed the ratio $p(\mathbf{y}|t_0)/p(\mathbf{y}|t)$ into two components: the first is the ratio of marginal likelihoods under the correct model, conditional on different candidate change points; the second involves sums of ratios where we fix the change point and compare alternative models to the true one in both the numerator and denominator. Inference for the change point appears only in the first term, which can therefore be interpreted as reflecting change point detection consistency. The second term can be seen as a kind of variable selection consistency ratio, where models are compared conditional on the same change point but under different models.

Define the set 
\[
\mathcal{T}_{C^*}:=\left\{1 \leq t \leq T : |t-t_0|> C^* \log T \right\}. 
\]
Consider the event $\Omega=\Omega_1\cap\Omega_2\cap\Omega_3$ defined by
\begin{align*}
	\Omega_1:=& \bigcap_{t\in \mathcal{T}_{C^*}} \left\{\left\vert\sum_{j=1}^t \epsilon_j x_{ji} \right\vert < \sqrt{C t \log T} \right\} \\
	\Omega_2:=& \bigcap_{t\in \mathcal{T}_{C^*}}\left\{ \left\vert \sum_{j=t+1}^T \epsilon_j x_{j} \right\vert < \sqrt{ C(T-t+1) \log T }\right\} \\
	\Omega_3:=& \bigcap_{t\in \mathcal{T}_{C^*}} \left\{ \left\vert \sum_{j=\min\{t_0,t\}}^{\max\{t_0,t\} -1} \epsilon_j x_{j}\right\vert < \sqrt{ C|t_0-t| \log T }\right\}
\end{align*}
By union bound and Lemma~\ref{lemma:event_bounds}:
\begin{equation*}
	\mathbb{P}(\Omega) \geq  1- \frac{2}{T}- \frac{2}{T}- \frac{2}{T}.
\end{equation*}

Suppose that $\Omega$ holds, and assume $\sigma^2=1$.

\noindent \textbf{Step 1:} Fix $t \in \{1,\ldots,T\}$, we want to show $\delta :=\log P(\mathbf{y}|t_0,\mathbf{m}_0)- \log P(\mathbf{y}|t,\mathbf{m}_0) > 0 $ for $|t-t_0|>C^*  \log T $. Let's assume without loss of generality that $t<t_0$. Since in $\delta$ we are conditioning on the same true model, throughout this step we drop the subscript $0$ in $\mathbf{m}_0$ to make the notation lighter (the subscript will be necessary in Step 2 when comparing different model). Hence,  $\mathbf{m}=(m_1,m_2)$, with $m_1$ and $m_2$ denoting the true left and right models in lieu of $m_{01}$ and $m_{02}$, and $\mathbf{X}_1$ and $\mathbf{X}_2$ the corresponding design matrices in lieu of $\mathbf{X}_{m_{01}}$ and $\mathbf{X}_{m_{02}}$ respectively.

Using \eqref{eq:marg_like}, $\delta$ can be rewritten as :

\begin{align} \label{eq:delta1_step1_new}
	\delta =& \frac{1}{t_0} \left\lVert\sum_{j=1}^{t_0} y_j \mathbf{x}_{j1}\right\rVert_2^2 - \frac{2|m_{1}|}{2} \log t_0 + \frac{1}{T-t_0+1}\left\lVert\sum_{j=t_0+1}^{T} y_j \mathbf{x}_{j2}\right\rVert_2^2 -\frac{2|m_{2}|}{2} \log (T-t_0) \nonumber \\
	& - \frac{1}{t}\left\lVert\sum_{j=1}^{t} y_j \mathbf{x}_{j1}\right\rVert_2^2 + \frac{2|m_{1}|}{2}\log t - \frac{1}{T-t+1}\left\lVert\sum_{j=t+1}^{T} y_j \mathbf{x}_{j2}\right\rVert_2^2 + \frac{2|m_{2}|}{2}\log (T-t+1).
\end{align}
We follow the decomposition of inner products suggested by \cite{berlind2025} -- the argument itself is completely different as that reference deals with a high-dimensional change in mean. The first as 
\begin{align*}
    \left\lVert\sum_{j=1}^{t_0} y_j \mathbf{x}_{j1}\right\rVert_2^2&=\left\lVert\sum_{j=1}^{t_0} (y_j - \bm{\beta}^\intercal_{1} \mathbf{x}_{j1} + \bbe^\intercal_{1} \bx_{j1}) \bx_{j1} \right\rVert_2^2 \\ &=  \left\lVert\sum_{j=1}^{t_0} \epsilon_j \bx_{j1} \right\rVert_2^2+ \left\lVert\sum_{j=1}^{t_0} (\bbe^\intercal_{1} \bx_{j1}) \bx_{j1} \right\rVert_2^2 + 2 \left(\sum_{j=1}^{t_0} \epsilon_j \bx_{j1}\right)^\intercal \left(\sum_{j=1}^{t_0} (\bbe^\intercal_{1} \bx_{j1}) \bx_{j1}\right).
\end{align*}
Now, note that for $i=1,\ldots, 2 |m_1|$
\begin{align*}
	\left(\sum_{j=1}^{t_0} (\bbe^\intercal_{1} \bx_{j1}) x_{ji1} \right)^2&= 	\left(\sum_{j=1}^{t_0} \beta_{i1} x_{ji1}^2 +\sum_{z\neq i} \beta_{z1} \sum_{j=1}^{t_0}  x_{ji1} x_{jz1} \right)^2 \\
    &\geq \beta_{i1}^2\left(\sum_{j=1}^{t_0}  x_{ji1}^2\right)^2 - 2 \left\vert\sum_{z\neq i} \beta_{z1} \sum_{j=1}^{t_0}  x_{ji1} x_{jz1} \right\vert  \left\vert \sum_{j=1}^{t_0} \beta_{i1} x_{ji1}^2 \right\vert \\
	&\geq \beta_{i1}^2 \left(\frac{t_0}{2}-C_{\Omega} \right)^2 -2 C_{\Omega}2|m_1| \beta_* \beta_{i1} \left(\frac{t_0}{2}+C_{\Omega} \right)\\
	&\geq \beta_{i1}^2\frac{t_0^2}{4} + C_{\Omega}^2 \beta_{i1}^2 - 2  \beta_{i1}^2\frac{t_0}{2} - 2 C_{\Omega}d \beta_*^2 \left(\frac{t_0}{2}+C_{\Omega} \right).
\end{align*}
Similarly,
\begin{align*}
\left(\sum_{j=1}^{t_0} \epsilon_j x_{ji1}\right) \left(\sum_{j=1}^{t_0} (\bbe^\intercal_{1} \bx_{j1}) x_{ji1}\right)=& \left(\sum_{j=1}^{t_0} \epsilon_j x_{ji1}\right) \left(\sum_{j=1}^{t_0} \beta_{i1} x_{ji1}^2 +\sum_{z\neq i} \beta_{z1} \sum_{j=1}^{t_0}  x_{ji1} x_{jz1} \right) \\ \geq&   \left(\sum_{j=1}^{t_0} \epsilon_j x_{ji1}\right)\beta_{i1} \frac{t_0}{2}-\left \vert\sum_{j=1}^{t_0} \epsilon_j x_{ji1}\right\vert \left(C_{\Omega}2|m_{1}| \beta_* + C_{\Omega} \right),
\end{align*}
 Putting everything together, the following lower bound holds
\[
\left\lVert\sum_{j=1}^{t_0} y_j \mathbf{x}_{j1}\right\rVert_2^2 \geq \left\lVert\sum_{j=1}^{t_0} \epsilon_j \mathbf{x}_{j1}\right\rVert_2^2 + \bbe^\intercal_{1}\bbe_{1} \frac{t_0^2}{4} + R_1(t)+ 2 \sum_{i=1}^{2|m_{1}|}\left[\left(\sum_{j=1}^{t_0} \epsilon_j x_{ji1}\right)\beta_{i1} \frac{t_0}{2}-\left \vert\sum_{j=1}^{t_0} \epsilon_j x_{ji1}\right\vert C_{\Omega}2|m_{1}| \beta_*  \right].
\]
where $R_1(t_0)= 2 d[ -C_{\Omega}^2 \beta_*^2 - 2  \beta_*^2\frac{t_0}{2} - 2 C_{\Omega}2|m_{1}| \beta_*^2 \left(\frac{t_0}{2}+C_{\Omega} \right)]$.  By a similar argument, we can lower bound the second and third inner products in \eqref{eq:delta1_step1_new}
\begin{align*}
\left\lVert\sum_{j=t_0+1}^{T} y_j \mathbf{x}_{j2}\right\rVert_2^2 \geq& \left\lVert\sum_{j=t_0+1}^{T} \epsilon_j \mathbf{x}_{j2}\right\rVert_2^2 + \bbe^\intercal_{2}\bbe_{2} \frac{(T-t_0+1)^2}{4} + R_2(T-t_0+1) \\ &+ 2\sum_{i=1}^{2|m_{2}|}\left(\sum_{j=t_0+1}^{T} \epsilon_j x_{ji2}\right) \beta_{i2}\frac{T-t_0+1}{2} \\
& - 2\sum_{i=1}^{2|m_{2}|}\left\vert \sum_{j=t_0+1}^{T} \epsilon_j x_{ji2}\right\vert C_{\Omega}2|m_{2}| \beta_*  ,
\end{align*}
\[
\left\lVert\sum_{j=1}^{t} y_j \mathbf{x}_{j1}\right\rVert_2^2 \geq \left\lVert\sum_{j=1}^{t} \epsilon_j \mathbf{x}_{j1}\right\rVert_2^2 + \bbe^\intercal_{1}\bbe_{1} \frac{t^2}{4} + R_3(t)+ 2\sum_{i=1}^{2|m_{1}|} \left[\left(\sum_{j=1}^{t} \epsilon_j x_{ji1}\right) \beta_{i1}\frac{t}{2} - \left\vert\sum_{j=1}^{t} \epsilon_j x_{ji1}\right\vert C_{\Omega}2|m_{1}| \beta_* \right],
\]
where $R_2$ and $R_3$ can be defined analogously. Recalling that $t<t_0$, the fourth inner product requires an extra step because they depend on the true parameters of models $m_1$ and $m_2$. We rewrite it as 
\begin{align*}
\left\lVert\sum_{j=t+1}^{T} y_j \mathbf{x}_{j2}\right\rVert_2^2 =&\left\lVert\sum_{j=t+1}^{t_0} (y_j - \bm{\beta}^\intercal_{1} \mathbf{x}_{j1} + \bbe^\intercal_{1} \bx_{j1}) \bx_{j2}+\sum_{j=t_0+1}^{T} (y_j - \bm{\beta}^\intercal_{2} \mathbf{x}_{j2} + \bbe^\intercal_{2} \bx_{j2}) \bx_{j2}\right\rVert_2^2\\
=&\left\lVert\sum_{j=t+1}^{T} \epsilon_j  \mathbf{x}_{j2}\right\rVert_2^2 + \left\lVert\sum_{j=t+1}^{t_0} (\bbe^\intercal_{1} \bx_{j1})  \mathbf{x}_{j2}\right\rVert_2^2 +  \left\lVert\sum_{j=t_0+1}^{T} (\bbe^\intercal_{2} \bx_{j2})  \mathbf{x}_{j2}\right\rVert_2^2 \\&+2\left(\sum_{j=t+1}^{T} \epsilon_j  \mathbf{x}_{j2}\right)^\intercal \left(\sum_{j=t_0+1}^{T} (\bbe^\intercal_{2} \bx_{j2})  \mathbf{x}_{j2}\right)  \\
&+2\left(\sum_{j=t+1}^{T} \epsilon_j  \mathbf{x}_{j2}\right)^\intercal\left(\sum_{j=t+1}^{t_0} (\bbe^\intercal_{1} \bx_{j1})  \mathbf{x}_{j2}\right) \\&+ 2\left(\sum_{j=t+1}^{t_0} (\bbe^\intercal_{1} \bx_{j1})  \mathbf{x}_{j2}\right)^\intercal\left(\sum_{j=t_0+1}^{T} (\bbe^\intercal_{2} \bx_{j2})  \mathbf{x}_{j2}\right).
\end{align*}
The first, third and fourth terms on the second row can be lower bounded via arguments identical to the ones used so far. The second term on the second row and the terms in the third row include cross-terms between models $m_{1}$ and $m_{2}$.  We start by the second row term
\begin{align*}
  \left\lVert\sum_{j=t+1}^{t_0} (\bbe^\intercal_{1} \bx_{j1})  \mathbf{x}_{j2}\right\rVert_2^2 &= \sum_{i=1}^{2|m_{2}|} \left(\sum_{j=t+1}^{t_0} \left(\sum_{z=1}^{2|m_{1}|}\beta_{z1} x_{jz1}\right)    x_{ji2}\right)^2 = \sum_{i=1}^{2|m_{2}|} \left( \sum_{z=1}^{2|m_{1}|}\beta_{z1} \sum_{j=t+1}^{t_0} x_{jz1}   x_{ji2}\right)^2\\
  &=\sum_{i\in I_0} \left(\beta_{i1} \sum_{j=t+1}^{t_0}   x^2_{ji2} +  \sum_{z\neq i} \beta_{z1} \sum_{j=t+1}^{t_0} x_{jz1}   x_{ji2}\right)^2 + \sum_{i\in I_2} \left(\sum_{z=1}^{2|m_{1}|} \beta_{z1} \sum_{j=t+1}^{t_0} x_{jz1}   x_{ji2}\right)^2\\
  &\geq \sum_{i\in I_0} \left(\beta_{i1} \frac{t_0-t}{2} -C_{\Omega} 2 |m_1|\beta_*\right)^2\\
  &\geq \bbe_{I_01}^\intercal \bbe_{I_01}\frac{(t_0-t)^2}{4} +H_1(t_0-t),
\end{align*}
where $H_1(t_0-t)=4C_{\Omega}|m_1|\beta_* \sum_{i\in I_0}\beta_{i1} (t_0-t)$. The second equality follows from noticing that in $I_0$ (if not empty), there is a covariate $x_{jz1}$ that is equal to $x_{ji2}$. If empty, the sum is equal to zero. The first inequality follows by Lemma~\ref{lemma:trig_approx} and removing the positive terms. The other terms in the third row that include parameters from $m_1$ and $m_2$ can be bound similarly
\begin{align*}
\left(\sum_{j=t+1}^{T} \epsilon_j  \mathbf{x}_{j2}\right)\left(\sum_{j=t+1}^{t_0} (\bbe^\intercal_{1} \bx_{j1})  \mathbf{x}_{j2}\right) \geq& \sum_{i=1}^{2|m_{2}|} \left(\sum_{j=t+1}^{T} \epsilon_j  x_{ji2}\right)\left(\sum_{z=1}^{2|m_{1}|} \beta_{z1} \sum_{j=t+1}^{t_0} x_{jz1} x_{ji2} \right)\\
\geq& \sum_{i \in I_0} \left(\sum_{j=t+1}^{T} \epsilon_j  x_{ji2}\right)\left(\beta_{i1} \frac{t_0-t}{2} - 2|m_1| C_{\Omega} \beta_*  \right) \\&- \sum_{i \in I_2} \left\vert\sum_{j=t+1}^{T} \epsilon_j  x_{ji2}\right\vert2|m_1| C_{\Omega} \beta_*\\
\geq&\sum_{i \in I_0} \left(\sum_{j=t+1}^{T} \epsilon_j  x_{ji2}\right)\left(\beta_{i1} \frac{t_0-t}{2} \right) \\&-\sum_{i=1}^{2|m_{2}|} \left\vert\sum_{j=t+1}^{T} \epsilon_j  x_{ji2}\right\vert 2|m_1| C_{\Omega} \beta^*, 
\end{align*}
where in the second inequality we split the summation into indexes $I_0$ and $I_2$ and used Lemma~\ref{lemma:trig_approx}. Similarly
\begin{align*}
\left(\sum_{j=t+1}^{t_0} (\bbe^\intercal_{1} \bx_{j1})  \mathbf{x}_{j2}\right)\left(\sum_{j=t_0+1}^{T} (\bbe^\intercal_{2} \bx_{j2})  \mathbf{x}_{j2}\right)\geq& \sum_{i=1}^{2|m_{2}|} \left(\sum_{z=1}^{2|m_{1}|} \beta_{z1} \sum_{j=t+1}^{t_0} x_{jz1} x_{ji2} \right) \left(\beta_{i2}\frac{T-t_0+1}{2}-C_{\Omega}2|m_{2}| \beta_*  \right) \\
\geq& \sum_{i\in I_0} \left(\beta_{i1} \frac{t_0-t}{2} - 2|m_1| C_{\Omega} \beta_* \right) \left(\beta_{i2}\frac{T-t_0+1}{2}-C_{\Omega}2|m_{2}| \beta_* \right)\\
&  - \sum_{i\in I_2} 2|m_1| C_{\Omega} \beta^*  \left(\beta_{i2}\frac{T-t_0+1}{2}-C_{\Omega}2|m_{2}| \beta_*  \right)\\
\geq & \sum_{i\in I_0} \beta_{i1} \frac{t_0-t}{2} \beta_{i2}\frac{T-t_0+1}{2} + H_2(t_0-t,T-t_0+1),
\end{align*}
where the second inequality holds because the quantities inside the round bracket are positive for $T$ large enough, and
 $H_2(t_0-t,T-t_0+1)$ is defined collecting all the remaining term. Note that $H_2(t_0-t,T-t_0+1)/(T-t+1)=\mathcal{O}(1)$ since it includes all lower order terms. 
 
 \noindent Putting everything together, we obtain the following lower bound for the fourth inner product
\begin{align*}
\left\lVert\sum_{j=t+1}^{T} y_j \mathbf{x}_{j2}\right\rVert_2^2 \geq& \left\lVert\sum_{j=t+1}^{T} \epsilon_j  \mathbf{x}_{j2}\right\rVert_2^2+ \bbe^\intercal_{I_01}\bbe_{I_01} \frac{(t_0-t)^2}{4} + H_1(t_0-t) +\bbe^\intercal_{2}\bbe_{2} \frac{(T-t_0+1)^2}{4} + R_2(T-t_0+1) \\
&+ 2\sum_{i=1}^{2|m_{2}|}\left[\left(\sum_{j=t_0+1}^{T} \epsilon_j x_{ji2}\right) \beta_{i2}\frac{T-t_0+1}{2}- \left\vert\sum_{j=t_0+1}^{T} \epsilon_j x_{ji2}\right\vert C_{\Omega}2|m_{2}| \beta_* \right] \\&+2\sum_{i \in I_0} \left(\sum_{j=t+1}^{T} \epsilon_j  x_{ji2}\right)\left(\beta_{i1} \frac{t_0-t}{2} \right)\\
&-2\sum_{i=1}^{2|m_{2}|} \left\vert\sum_{j=t+1}^{T} \epsilon_j  x_{ji2}\right\vert 2|m_1| C_{\Omega} \beta_*  \\&+  2\sum_{i\in I_0} \beta_{i1} \frac{t_0-t}{2} \beta_{i2}\frac{T-t_0+1}{2} + 2H_2(t_0-t,T-t_0+1).
\end{align*}
Let $C_1=C_{\Omega}2 d \beta_*$, it holds the following lower bound for $\delta$
\begin{align}\label{eq:delta_2}
	\delta \geq&\frac{\left\lVert\sum_{j=1}^{t_0} \epsilon_j \mathbf{x}_{j1}\right\rVert_2^2}{t_0}-\frac{\left\lVert\sum_{j=1}^{t} \epsilon_j \mathbf{x}_{j1}\right\rVert_2^2}{t}+ \frac{\left\lVert\sum_{j=t_0+1}^{T} \epsilon_j \mathbf{x}_{j2}\right\rVert_2^2}{T-t_0+1}-\frac{\left\lVert\sum_{j=t+1}^{T} \epsilon_j \mathbf{x}_{j2}\right\rVert_2^2}{T-t+1}\nonumber \\
&	+\frac{ \bbe^\intercal_{1}\bbe_{1} \frac{t_0^2}{4}}{t_0}  - \frac{\bbe^\intercal_{1}\bbe_{1} \frac{t^2}{4}}{t} + \frac{\bbe^\intercal_{2}\bbe_{2} \frac{(T-t_0+1)^2}{4} }{T-t_0+1}- \frac{\bbe^\intercal_{I_01}\bbe_{I_01} \frac{(t_0-t)^2}{4}  +\bbe^\intercal_{2}\bbe_{2} \frac{(T-t_0+1)^2}{4} }{T-t+1} \nonumber\\
&	+  \sum_{i=1}^{2|m_{1}|}\left(\sum_{j=1}^{t_0} \epsilon_j x_{ji1}\right) \beta_{i1}- \sum_{i=1}^{2|m_{1}|}\left(\sum_{j=1}^{t} \epsilon_j x_{ji1}\right) \beta_{i1}  - \frac{t_0-t}{T-t+1} \sum_{i \in I_0} \left(\sum_{j=t+1}^{t_0} \epsilon_j  x_{ji2}+\sum_{j=t_0+1}^{T} \epsilon_j  x_{ji2} \right)\beta_{i1}   \nonumber \\
&+ \sum_{i=1}^{2|m_{2}|}\left(\sum_{j=t_0+1}^{T} \epsilon_j x_{ji2}\right) \beta_{i2} - \frac{T-t_0+1}{T-t+1}\sum_{i=1}^{2|m_{2}|}\left(\sum_{j=t_0+1}^{T} \epsilon_j x_{ji2}\right) \beta_{i2} -  \frac{2 C_1}{t_0} \sum_{i=1}^{2|m_{1}|}\left\vert \sum_{j=1}^{t_0} \epsilon_j x_{ji1}\right\vert   \\
&- \frac{2 C_1}{t}\sum_{i=1}^{2|m_{1}|}\left\vert\sum_{j=1}^{t} \epsilon_j x_{ji1}\right\vert- \frac{2 C_1}{T-t_0+1}\sum_{i=1}^{2|m_{2}|}\left\vert \sum_{j=t_0+1}^{T} \epsilon_j x_{ji2}\right\vert \nonumber\\
&- \frac{2 C_1}{T-t+1}\sum_{i=1}^{2|m_{2}|}\left[\left\vert \sum_{j=t_0+1}^{T} \epsilon_j x_{ji2}\right\vert+ \left\vert\sum_{j=t+1}^{T} \epsilon_j  x_{ji2}\right\vert  \right]  \nonumber \\
&-\frac{2}{T-t+1}  \left[\sum_{i\in I_0} \beta_{i1} \frac{t_0-t}{2} \beta_{i2}\frac{T-t_0+1}{2} + H_2(t_0-t,T-t_0+1) \right] \nonumber \\
&+ \frac{R_1(t_0)}{t_0} -\frac{R_3(t)}{t} +\frac{R_2(T-t_0+1)}{T-t_0+1} -\frac{ H_1(t_0-t)+ R_2(T-t_0+1)}{T-t+1} \nonumber \\ &+\frac{2|m_{1}|}{2} \log \frac{t}{t_0} + \frac{2|m_{2}|}{2}\log \frac{T-t+1}{T-t_0+1}. \nonumber
\end{align}
Now, note that $\bbe_{I_01}^\intercal \bbe_{I_01} \geq 0$ and $-\bbe_{I_01}^\intercal \bbe_{I_01} \geq -\bbe_{1}^\intercal \bbe_{1}$, so we can substitute this lower bound in the second row of \eqref{eq:delta_2}. We consider the second and second-last rows in \eqref{eq:delta_2}, which include terms quadratic in $\bbe_{1}$ and $\bbe_{2}$.  They can be rewritten as follows
\begin{align*}
&\frac{T-t_0+1}{T-t+1} \frac{(t_0-t)}{4} \left(\bbe_{1}^\intercal \bbe_{1} +\bbe_{2}^\intercal \bbe_{2} -2 \sum_{i\in I_0} \beta_{i1}  \beta_{i2}  \right) \\
&\qquad\qquad=\frac{T-t_0+1}{T-t+1} \frac{t_0-t}{4} [\bbe_{I_11}^\intercal \bbe_{I_11} +\bbe_{I_22}^\intercal \bbe_{I_22} + (\bbe_{I_01}-\bbe_{I_02})^\intercal(\bbe_{I_01}-\bbe_{I_02})] \\
&\qquad\qquad=\frac{t_0-t}{ T-t+1}\frac{T-t_0+1}{4} \kappa.
\end{align*}
In the rest of proof, we will lower bound the remaining terms, collect $(t_0-t)/(T-t+1)$, and finally show that $\delta>0$.

The first two terms  in the first row of \eqref{eq:delta_2} can be lower bounded as follows
\begin{align*}
	&\frac{\left\lVert\sum_{j=1}^{t_0} \epsilon_j \mathbf{x}_{j1}\right\rVert_2^2}{t_0}-\frac{\left\lVert\sum_{j=1}^{t} \epsilon_j \mathbf{x}_{j1}\right\rVert_2^2}{t} \\
    &\qquad=	\frac{\left\lVert\sum_{j=1}^{t} \epsilon_j \mathbf{x}_{j1}\right\rVert_2^2 + \left\lVert\sum_{j=t+1}^{t_0} \epsilon_j \mathbf{x}_{j1}\right\rVert_2^2 + 2\left(\sum_{j=1}^{t} \epsilon_j \mathbf{x}_{j1}\right)^\intercal\left(\sum_{j=t+1}^{t_0} \epsilon_j \mathbf{x}_{j1}\right)}{t_0}-\frac{\left\lVert\sum_{j=1}^{t} \epsilon_j \mathbf{x}_{j1}\right\rVert_2^2}{t} \\
	&\qquad=-\frac{(t_0-t)}{t_0}\left[\frac{\left\lVert\sum_{j=1}^{t} \epsilon_j \mathbf{x}_{j1}\right\rVert_2^2}{t} -\frac{\left\lVert\sum_{j=t+1}^{t_0} \epsilon_j \mathbf{x}_{j1}\right\rVert_2^2}{t_0-t} - 2 \frac{\sqrt{t}}{\sqrt{t_0-t}} \frac{\left(\sum_{j=1}^{t} \epsilon_j \mathbf{x}_{j1}\right)^\intercal\left(\sum_{j=t+1}^{t_0} \epsilon_j \mathbf{x}_{j1}\right)}{\sqrt{t_0-t}\sqrt{t}} \right].
\end{align*}	
In $\Omega$, there exists some $C \geq 1$ such that
\begin{align*}
\left\vert \frac{\left\lVert\sum_{j=1}^{t} \epsilon_j \mathbf{x}_{j1}\right\rVert_2^2}{t} \right \vert=\left \vert \frac{\sum_{i=1}^{2|m_{1}| }(\sum_{j=1}^{t} \epsilon_j x_{ji1})^2}{t} \right\vert <& 2|m_{1}|  C \log T\\
\left\vert \frac{\left\lVert\sum_{j=t+1}^{t_0} \epsilon_j \mathbf{x}_{j1}\right\rVert_2^2}{t_0-t} \right \vert=\left \vert \frac{\sum_{i=1}^{2|m_{1}| }(\sum_{j=t+1}^{t_0} \epsilon_j x_{ji1})^2}{t_0-t} \right\vert <& 2|m_{1}|  C \log T\\
\left\vert \frac{\left(\sum_{j=1}^{t} \epsilon_j \mathbf{x}_{j1}\right)^\intercal\left(\sum_{j=t+1}^{t_0} \epsilon_j \mathbf{x}_{j1}\right)}{\sqrt{t_0-t}\sqrt{t}}\right\vert=\left\vert \frac{\sum_{i=1}^{2|m_{1}| } \left(\sum_{j=1}^{t} \epsilon_j x_{ji1}\right)\left(\sum_{j=t+1}^{t_0} \epsilon_j x_{ji1}\right)}{\sqrt{t_0-t}\sqrt{t}}\right\vert <& 2|m_{1}|  C \log T
\end{align*}
which gives 
\begin{align*}
	&\frac{\left\lVert\sum_{j=1}^{t_0} \epsilon_j \mathbf{x}_{j1}\right\rVert_2^2}{t_0}-\frac{\left\lVert\sum_{j=1}^{t} \epsilon_j \mathbf{x}_{j1}\right\rVert_2^2}{t} \geq -2|m_{1}|  (t_0-t) C\left[  \frac{ \log T}{t_0} + \frac{ \log T \sqrt{t} }{t_0 \sqrt{t_0-t}} \right]\\
	&\qquad\geq - \frac{t_0-t}{T-t+1}2|m_{1}| C\left\{(T-t_0+1)\left[ \frac{ \log T}{t_0} + \frac{ \log T \sqrt{t} }{t_0 \sqrt{t_0-t}} \right] +(t_0-t)\left[ \frac{ \log T}{t_0} + \frac{ \log T \sqrt{t} }{t_0 \sqrt{t_0-t}} \right]  \right\}\\
	&\qquad\geq  - \frac{t_0-t}{T-t+1}2|m_{1}| C\left[ \frac{2(T-t_0+1)}{t_0}  +\frac{(T-t_0+1) \log T}{\sqrt{t_0} \sqrt{C^*\log T}}\frac{\sqrt{t}}{\sqrt{t_0}}  +\log T  +\frac{\sqrt{t}}{\sqrt{t_0}} \frac{\sqrt{t_0-t}\log T}{\sqrt{t_0}} \right]\\
	&\qquad\geq  - \frac{t_0-t}{T-t+1}2|m_{1}| C\left[ \frac{2(T-t_0+1)}{t _0} + 2 \log T \right],
\end{align*}
where, in the third equality we collected a common term $(t_0-t)/(T-t+1)$, in the third inequality we used $t/t_0<1$, $(t_0-t)/t_0<1$, $\log T/t_0 <1$, and $(t_0-t)>C^* \log T$ with $C^*>1$. We will be use these identities extensively below.

Similarly, we can lower bound the last two terms of \eqref{eq:delta_2} first row:
\begin{align*}
&\frac{\left\lVert\sum_{j=t_0+1}^{T} \epsilon_j \mathbf{x}_{j2}\right\rVert_2^2}{T-t_0+1}-\frac{\left\lVert\sum_{j=t+1}^{T} \epsilon_j \mathbf{x}_{j2}\right\rVert_2^2}{T-t+1} =\\
&\quad=\frac{\left\lVert\sum_{j=t_0+1}^{T} \epsilon_j \mathbf{x}_{j2}\right\rVert_2^2}{T-t_0+1} \\
&\qquad\qquad- \frac{\left\lVert\sum_{j=t+1}^{t_0} \epsilon_j \mathbf{x}_{j2}\right\rVert_2^2 + \left\lVert\sum_{j=t_0+1}^{T} \epsilon_j \mathbf{x}_{j2}\right\rVert_2^2 + 2\left(\sum_{j=s}^{t_0} \epsilon_j \mathbf{x}_{j2}\right)^\intercal\left(\sum_{j=t_0+1}^{T} \epsilon_j \mathbf{x}_{j2}\right)}{T-t+1} \\
&\quad=\frac{t_0-t}{T-t+1}\Bigg[\frac{\left\lVert\sum_{j=t_0+1}^{T} \epsilon_j \mathbf{x}_{j2}\right\rVert_2^2}{T-t_0+1} \\
&\qquad\qquad-\frac{\left\lVert\sum_{j=t+1}^{t_0} \epsilon_j \mathbf{x}_{j2}\right\rVert_2^2}{t_0-t}- 2 \frac{\sqrt{T-t_0+1}}{\sqrt{t_0-t}} \frac{\left(\sum_{j=t_0+1}^{T} \epsilon_j \mathbf{x}_{j2}\right)^\intercal\left(\sum_{j=t+1}^{t_0} \epsilon_j \mathbf{x}_{j2}\right)}{\sqrt{t_0-t}\sqrt{T-t_0+1}} \Bigg].
\end{align*}	
Applying identical bounds to before, we have that in $\Omega$ there exists $C \geq 1$ such that
\begin{align*}
	\frac{\left\lVert\sum_{j=t_0+1}^{T} \epsilon_j \mathbf{x}_{j2}\right\rVert_2^2}{T-t_0+1}-\frac{\left\lVert\sum_{j=t+1}^{T} \epsilon_j \mathbf{x}_{j2}\right\rVert_2^2}{T-t+1}  \geq& -   \frac{t_0-t}{T-t+1} 2 |m_{2}| C\left[  \log T  + \frac{  \log T \sqrt{T-t_0+1}}{\sqrt{t_0-t}} \right]\\
	\geq& -  \frac{t_0-t}{T-t+1} 2 |m_{2}| C \left[  \log T  + \frac{  \sqrt{\log T} \sqrt{T-t_0+1}}{\sqrt{C^*}} \right].
\end{align*}
Similarly, noting that 
\[
 \sum_{i=1}^{2|m_{1}|}\left(\sum_{j=1}^{t_0} \epsilon_j x_{ji1}\right) \beta_{i1}- \sum_{i=1}^{2|m_{1}|}\left(\sum_{j=1}^{t} \epsilon_j x_{ji1}\right) \beta_{i1} =  \sum_{i\in I_0}\left(\sum_{j=t+1}^{t_0} \epsilon_j x_{ji1}\right) \beta_{i1}+ \sum_{i\in I_1}\left(\sum_{j=t+1}^{t_0} \epsilon_j x_{ji1}\right) \beta_{i1},
\]
 in $\Omega$ it holds that 
 \begin{align*}
\left\vert \sum_{i\in I_0}\left(\sum_{j=t+1}^{t_0} \epsilon_j x_{ji1}\right) \beta_{i1} - \frac{t_0-s}{T-t+1} \sum_{i \in I_0} \left(\sum_{j=t+1}^{t_0} \epsilon_j  x_{ji2}\right) \beta_{i1} \right\vert \geq&   -2|I_0|  \beta_* C \frac{T-t_0+1}{T-t+1} \sqrt{(t_0-t) \log T } \\
\geq&   -2|m_{1}|  \beta_* C\frac{t_0-t}{T-t+1}   \frac{T-t_0+1}{\sqrt{C^*}},
\end{align*}
and
\begin{align*}
\left\vert \sum_{i\in I_1}\left(\sum_{j=t+1}^{t_0} \epsilon_j x_{ji1}\right) \beta_{i1}\right\vert \geq&  -2|I_1|  \beta_* C \frac{t_0-t}{T-t+1} \left(\frac{T-t_0+1}{\sqrt{C^*}} + \sqrt{(t_0-t) \log T} \right).
 \end{align*}
 We can lower bound the other terms with similar arguments
\begin{align*}
	&\left\vert \sum_{i=1}^{2|m_{2}|}\left(\sum_{j=t_0+1}^{T} \epsilon_j x_{ji2}\right) \beta_{i2} - \frac{T-t_0+1}{T-t+1}\sum_{i=1}^{2|m_{2}|}\left(\sum_{j=t_0+1}^{T} \epsilon_j x_{ji2}\right) \beta_{i2}\right\vert \geq \\
    &\qquad\qquad\qquad\qquad\qquad\qquad\qquad\qquad\qquad\qquad- \frac{t_0-t}{T-t+1}2|m_{2}|  \beta_* C \sqrt{(T-t_0+1)\log T},
\end{align*}
and
\begin{align*}
	\left\vert \frac{t_0-t}{T-t+1} \sum_{i \in I_0} \left(\sum_{j=t_0+1}^{T} \epsilon_j  x_{ji2}\right)\beta_{i1} \right\vert \geq&- \frac{t_0-t}{T-t+1}2|I_0|  \beta_* C \sqrt{(T-t_0+1)\log T},
    \end{align*}
and
\begin{align*}
& -  \frac{2 C_1}{t_0} \sum_{i=1}^{2|m_{1}|}\left\vert \sum_{j=1}^{t_0} \epsilon_j x_{ji1}\right\vert - \frac{2 C_1}{t}\sum_{i=1}^{2|m_{1}|}\left\vert\sum_{j=1}^{t} \epsilon_j x_{ji1}\right\vert \geq \\
&\qquad\qquad\qquad\qquad\qquad\qquad-\frac{t_0-t}{T-t+1} 2 C_1 2 |m_{1}| C\frac{T-t_0+1 + t_0-t}{t_0-t}\left( \frac{\sqrt{\log T}}{ \sqrt{t_0}} +\frac{\sqrt{\log T}}{\sqrt{t}} \right)\\
 &\qquad\qquad\qquad\qquad\qquad\qquad\qquad\qquad\qquad\qquad\geq-\frac{t_0-t}{T-t+1} 4 C_1 2 |m_{1}| C \left(\frac{T-t_0+1}{C^*\sqrt{\log T}} + \sqrt{\log T}\right),
 \end{align*}
and
\begin{align*}
- \frac{2 C_1}{T-t_0+1}\sum_{i=1}^{2|m_{2}|}\left\vert \sum_{j=t_0+1}^{T} \epsilon_j x_{ji2}\right\vert \geq& -2 C_1 2 |m_{2}|  C \frac{\sqrt{\log T}}{\sqrt{T-t_0+1}} \\
\geq& -\frac{t_0-t}{T-t+1} 2 C_1 2 |m_{2}|  C\left( \frac{\sqrt{T-t_0+1}}{C^* \sqrt{\log T} } + \frac{\sqrt{\log T}}{C^*\sqrt{T-t_0+1}}\right),
\end{align*}
and
\begin{align*}
& - \frac{2 C_1}{T-t+1}\sum_{i=1}^{2|m_{2}|}\left[\left\vert \sum_{j=t_0+1}^{T} \epsilon_j x_{ji2}\right\vert+ \left\vert\sum_{j=t+1}^{T} \epsilon_j  x_{ji2}\right\vert  \right] \geq \\
&\qquad\qquad\qquad\qquad\qquad\qquad-\frac{t_0-t}{T-t+1} 2 C_1 2 |m_{2}| C \left(\frac{\sqrt{(T-t_0+1)}}{\sqrt{C^*\log T}} +\frac{\sqrt{ (T-t+1)}}{\sqrt{C^*\log T}}\right). 
%\geq& -2 C_1 2 |m_{2}| C \left(2\frac{\sqrt{T-t_0+1}}{C^* \sqrt{\log T} } +\frac{1}{\sqrt{C^*}}\right).
\end{align*}

\noindent Finally, the log terms in the last row can be lower bounded:
\[
\frac{2|m_{1}|}{2} \log \frac{t}{t_0} + \frac{2|m_{2}|}{2}\log \frac{T-t+1}{T-t_0+1} \geq -  |m_{1}| \frac{t_0-t}{{\sqrt{t_0}\sqrt{t}}}=-  |m_{1}| \frac{t_0-t}{T-t+1} \left(\frac{T-t_0+1}{\sqrt{t_0}}  + \sqrt{t_0-t} \right).
\]
where the inequality uses $\log (1-x)\geq -x/\sqrt{1-x}$ when $x\in [0,1)$. The remaining terms are of lower order.

Let $C_2:= 2 C d$. We have a new lower bound for $\delta$:
\begin{align*}
\delta \geq& \frac{t_0-t}{T-t+1} \Bigg[ \kappa \frac{T-t_0+1}{4} -2 C_2 \beta_*\frac{T-t_0+1}{\sqrt{C^*}} \\
&-C_2 \Bigg( 2 \beta_* \sqrt{(T-t_0+1) \log T} + 3 \log T + 4 C_1 \frac{\sqrt{T-t_0+1}}{C^* \sqrt{\log T}}\\
& + (1+\beta_*) \sqrt{(t_0-t)\log T} +2C_1\frac{\sqrt{\log T}}{\sqrt{T-t_0+1}}+ 4C_1\frac{T-t_0+1}{C^* \sqrt{\log T}}+ \frac{\sqrt{ (T-t+1)}}{\sqrt{C^*\log T}} \\
&+4 C_1 \sqrt{\log T} + 2\frac{T-t_0+1}{\sqrt{t_0}}+ \sqrt{t_0-t} \Bigg)  \Bigg].
\end{align*}
\noindent By Assumption 1 \textit{(b)}, $\kappa (T-t_0+1)>\sqrt{T \log^{1+\epsilon}T}$. Hence, for large $T$, $\kappa(T-t_0+1)/8$ is larger than what is between the round brackets. By Assumption 1 \textit{(c)} $\kappa>c$, it follows that for $T$ large enough
\[
\delta \geq \frac{t_0-t}{T-t+1} (T-t_0+1)\left(\frac{c}{8} - \frac{2C_2 \beta_*}{\sqrt{C^*}}\right).
\]
The term $(t_0-t)(T-t_0+1)/(T-t+1)$ decreases as a function of $t$, and it is minimized with $t=t_0-C^* \log T$
\[
\delta \geq \frac{C^* \log T}{T-t+1+C^* \log T} (T-t+1)\left(\frac{c}{8} - \frac{2C_2 \beta^*}{\sqrt{C^*}}\right) = \frac{C^* \log T}{1+\frac{C^* \log T}{T-t+1}} \left(\frac{c}{8} - \frac{2C_2 \beta^*}{\sqrt{C^*}}\right).
\]
For $C^*> 256 C_2^2 (\beta^*)^2/c^2$ and $T$ large enough, $\delta>0$.

\noindent \textbf{Step 2:} In this step, we restore the subscript $\mathbf{m}_0$ to denote the true model. We want to show that 
\begin{equation}\label{eq:step2}
    \frac{ p(\mathbf{y}|t,\mathbf{m})p(\mathbf{m})}{ p(\mathbf{y}|t,\mathbf{m}_0)p(\mathbf{m}_0)} \rightarrow 0, \quad\text{as} \quad T\rightarrow +\infty,
\end{equation} 
for $\mathbf{m}\neq \mathbf{m}_0$ and for all $t=1,\dots, T$ and $\mathbf{m} \in \mathcal{M}\times\mathcal{M}$. Recall that, given a change point, the left and right segments are independent. Therefore:
\begin{align}\label{eq:bf_ratio}
    \frac{ p(\mathbf{y}|t,\mathbf{m})p(\mathbf{m})}{ p(\mathbf{y}|t,\mathbf{m})p(\mathbf{m}_0)} &= \frac{p(\mathbf{m})}{p(\mathbf{m}_0)}\underbrace{\frac{ p(\mathbf{y}_{1:t}|m_1)}{ p(\mathbf{y}_{1:t}|m_{01})}}_{BF(m_1,m_{01})}  \underbrace{\frac{ p(\mathbf{y}_{t+1:T}|m_2)}{ p(\mathbf{y}_{t+1:T}|m_{02})}}_{BF(m_2,m_{02})},
\end{align} 
where the ratio of priors on models is bounded and independent of $T$ by Assumption 1 \textit{(e)}. The remaining terms in  \eqref{eq:bf_ratio} are Bayes factors. By Assumption 1 \textit{(d)}, the model is correctly specified. Bayes factors consistency in this case is a standard result \cite[e.g.][]{johnson2012bayesian}. As a consequence, \eqref{eq:bf_ratio} goes to $0$ at a parametric rate in the worst case scenario.

\noindent This concludes the proof.

\subsection{Proof of Proposition 1}

First, note that  for all $\omega\in(0,1/2)$, there exits $A,
\phi$ such that $\beta^{(1)}\sin (2 \pi \omega t) + \beta^{(2)} \cos (2 \pi\omega t)=A\sin(2 \pi \omega t+\phi)$. To see this, rewrite
\begin{align*}
    \beta^{(1)}\sin (2 \pi \omega t) + \beta^{(2)} \cos (2 \pi\omega t) &= \sqrt{( \beta^{(1)})^2+( \beta^{(2)})^2}\left(\frac{\beta^{(1)}\sin (2 \pi \omega t)}{\sqrt{( \beta^{(1)})^2+( \beta^{(2)})^2}} + \frac{\beta^{(2)}\cos (2 \pi\omega t)}{\sqrt{( \beta^{(1)})^2+( \beta^{(2)})^2}}  \right).
\end{align*}
Then, there exits $\phi$ such that
\[
\cos(\phi)=\frac{\beta^{(1)}}{\sqrt{( \beta^{(1)})^2+( \beta^{(2)})^2}} \text{  and  } \sin(\phi)=\frac{\beta^{(2)}}{\sqrt{( \beta^{(1)})^2+( \beta^{(2)})^2}}.
\]
Noting that $\tan(\phi)= \beta^{(2)}/ \beta^{(1)}$, then $\phi=\tan^{-1}(\beta^{(2)}/ \beta^{(1)})$. Using the trigonometric identity $\sin(\alpha+\beta)=\sin(\alpha)\cos(\beta)+\sin(\beta)\cos(\alpha)$, we have that 
\begin{align*}
 \sqrt{( \beta^{(1)})^2+( \beta^{(2)})^2}\left(\cos(\phi)\sin (2 \pi \omega t)+\cos (2 \pi\omega t) \sin(\phi) \right)&=\sqrt{( \beta^{(1)})^2+( \beta^{(2)})^2} \sin (2 \pi\omega t+\phi),
\end{align*}
and $A=\sqrt{( \beta^{(1)})^2+( \beta^{(2)})^2}$.
Now, given $\omega_1\in(0,1/2)$, we have to determine the set of $\omega_2$ where
\[
A\sin (2 \pi\omega_1 t_0+\phi)=A\sin (2 \pi\omega_2 t_0+\phi).
\]
The sine function is such that that $\sin(\theta)=\sin(\theta+2n\pi)$ with $n\in\mathbb{Z}$. $\Omega_A$ define such set of all shifts by $2\pi$. It also holds that 
$\sin(\theta)=\sin(\pi-\theta)$, $\Omega_B$ defines the sets of all those $\omega_2$. Both $\Omega_A$ and $\Omega_B$ are intersected with $(0,1/2)$ to give a valid parameter.

\clearpage
\section{Practicalities}
\label{sec:practical}
\textbf{Choice of the grid $\Omega$}. The choice of the grid affects \algonamesp performance, both in terms of accuracy and runtime. Here, we compare two possible grid construction. The first one assumes an equally spaced grid between $0$ and $1/2$ of length $p$. The second one builds a grid aided by the periodogram, which is an empirical estimator of the spectral density. The periodogram allows for a thinner grid in some intervals than others. The idea is to compute the periodogram on the entire series and define $\Omega$ as the set of $p$ frequencies with the highest power spectrum or the set of frequencies with a power spectrum greater than a threshold.

We compare the two grid constructions (equally spaced, and one based on the periodogram considers the $p$ frequencies with the highest power spectrum) on one replicate of Scenario 2 of Section 3 (the replicate is considered representative), and fix all the other parameters ($N_E$ and prior hyperparameters).

Table \ref{tab:om_grid} shows the results. The grid construction through the periodogram achieves better performance than the equally spaced grid, already for moderate $p$. Both strategies have comparable performance when $p\geq 400$. Small $p$ is preferable since it is faster.\\

\begin{table}[!ht]
  \centering
    \resizebox{0.65\textwidth}{!}{
    \begin{tabular}{lrlrlrl}
    \toprule
    \multicolumn{3}{c}{Periodogram} &       &  \multicolumn{3}{c}{Equally spaced} \\
    \cmidrule{1-3}\cmidrule{5-7}
    $p$ & & Change points & & $p$ & & Change points \\
    \cmidrule{1-1}\cmidrule{3-3}\cmidrule{5-5}\cmidrule{7-7}
    $10$  & & $(165, 624, 893)$ &       & $100$ & & $(181, 292, 461, 596, 896)$ \\
    $50$  & & $(167, 624, 896)$ &       & $200$ & & $(181, 292, 462, 624, 898)$ \\
    $100$ & & $(167, 624, 897)$ &       & $300$ & & $(163, 402, 627, 896)$ \\
    $200$ & & $(167, 624, 898)$ &       & $400$ & & $(165, 624, 896)$ \\
    $500$ & & $(167, 624, 899)$ &       & $500$ & & $(165, 624, 897)$ \\
    \bottomrule
    \end{tabular}}
    \vspace{0.5em}
    \caption{{Choice of the grid $\Omega$}. Segmentation obtained with different grid specification and size. One replicate of Scenario 2 in Section 3 with $T=1000$ and $m=3$. The true change point locations are $(174, 622, 892)$. We compare the performance of the proposed change point detection procedure increasing the dimension of the grid ($p$), for the equally spaced construction and the one based on the $p$-highest power spectrum frequencies in the periodogram.}\label{tab:om_grid}
\end{table}

\textbf{Choice of the number of effects $N_E$}. $N_E$ is a parameter of SuSiE that upper bound $L_{j}$ within a segment. The accuracy of the change point detection is affected by $N_E$. We consider one replicate of Scenario 2 of Section 3 (the same used for Table \ref{tab:om_grid}). In this example the true within model dimensions are $L_1=2$, $L_2=2$, $L_3=3$, and $L_4=2$. We construct $\Omega$ using the $p$ frequencies with the highest power spectrum in the periodogram and apply \algonamesp for varying $N_E$. Figure \ref{fig:choice_Ne} depicts the results. First, when $N_E<L_j$, the segmentation is nearly perfect but the estimation of the within-segment parameters (\textit{e.g.}, the mean) is not (see Figure \ref{fig:choice_Ne}(b) for $N_E=1$). On the other hand, if $N_E$ is set too large, even a complex signal can be well approximated by a combination of sinusoidal functions, leading to a poor segmentation. (see Figure \ref{fig:choice_Ne}(e)-(f) for $N_E=5,10$).

Our recommendation is to select a low $N_E$ (\textit{e.g.}, $1,2$ or $3$) as it leads to the best segmentation performance. This is likely not an issue in applications, as we see that \cite{andrieu1999joint} and \cite{hadj2020bayesian} always considers low model dimensions. However, once the segmentation step is done and a set of estimated change points is available, it is possible to re-estimate the within segment parameters to improve the inference on the underlying signal. In doing so, one can increase the number of effects $N_E$ and even consider specific $N_E$ for each segment.

An alternative is to have an automatic procedure that chooses $N_E$ with an information criterion. Indeed, we can see the estimated partition as a function of $N_E$, \textit{i.e.} $\hat{\mathcal{P}}(N_E)$, and choose $\hat{\mathcal{P}}(N_E)$ maximizing a certain information criterion (\textit{e.g.}, mBIC, MDL). Directly comparing $\hat{\mathcal{P}}(N_E)$ is not viable since the marginal likelihoods are largely driven by the within segment fit (or lack of), while we want to focus on the segmentation performance. Therefore, we suggest to estimate $\hat{\mathcal{P}}(N_E)$ for $N_E$ ranging from one to $\bar{N}_E$ (an arbitrary upper bound), then, given each partition $\hat{\mathcal{P}}(N_E)$, we re-fit the same within segment model with $\bar{N}_E$ effects in each segment. Finally, we choose $\hat{\mathcal{P}}(N_E)$ that maximizes a certain information criterion.

The automatic procedure applied to the series depicted in Figure \ref{fig:choice_Ne} with $\bar{N}_E=4$ returns $N_E=1$, which is equal to the largest within-segment model dimension.

\begin{figure}[!ht]
\centering
\subfigure[True mean function and change-points]{\includegraphics[width=.48\textwidth]{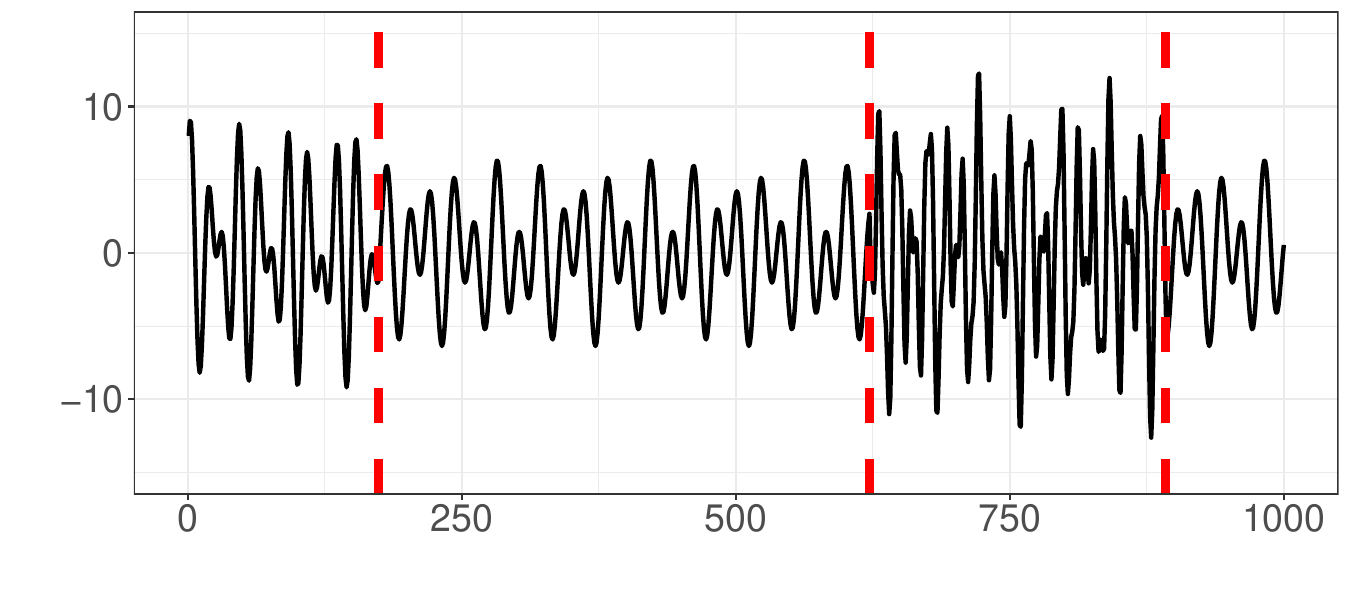}}
\subfigure[$N_E=1$]{\includegraphics[width=.47\textwidth]{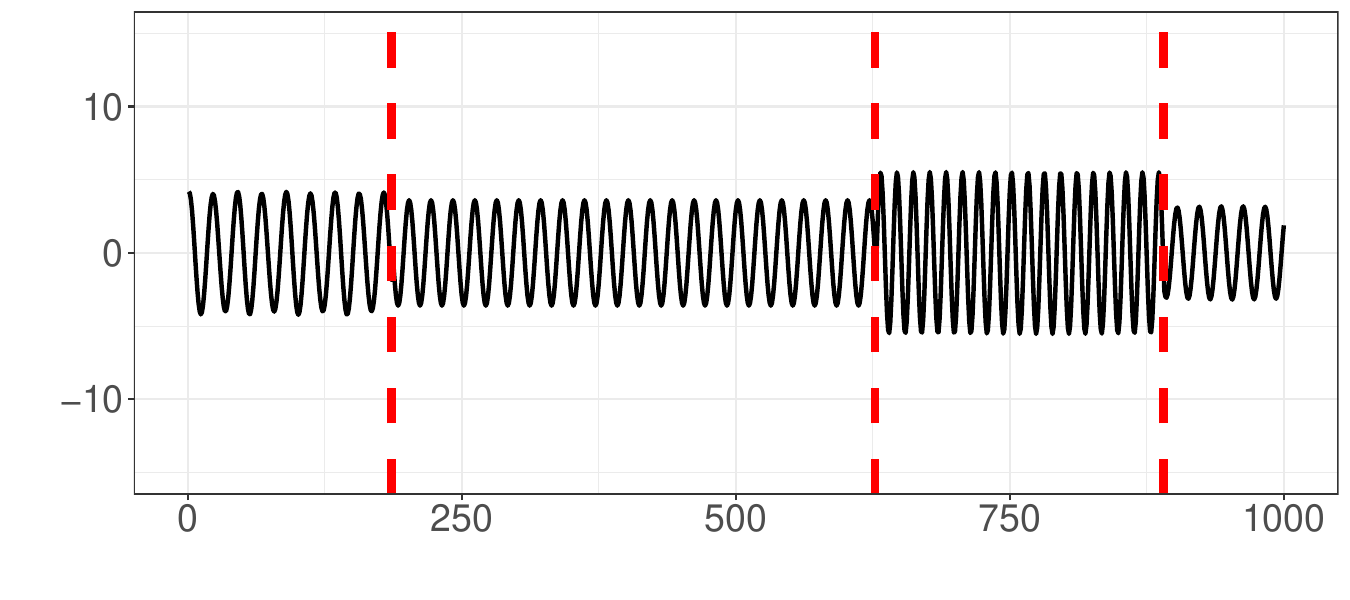}}\\
\subfigure[$N_E=2$]{\includegraphics[width=.47\textwidth]{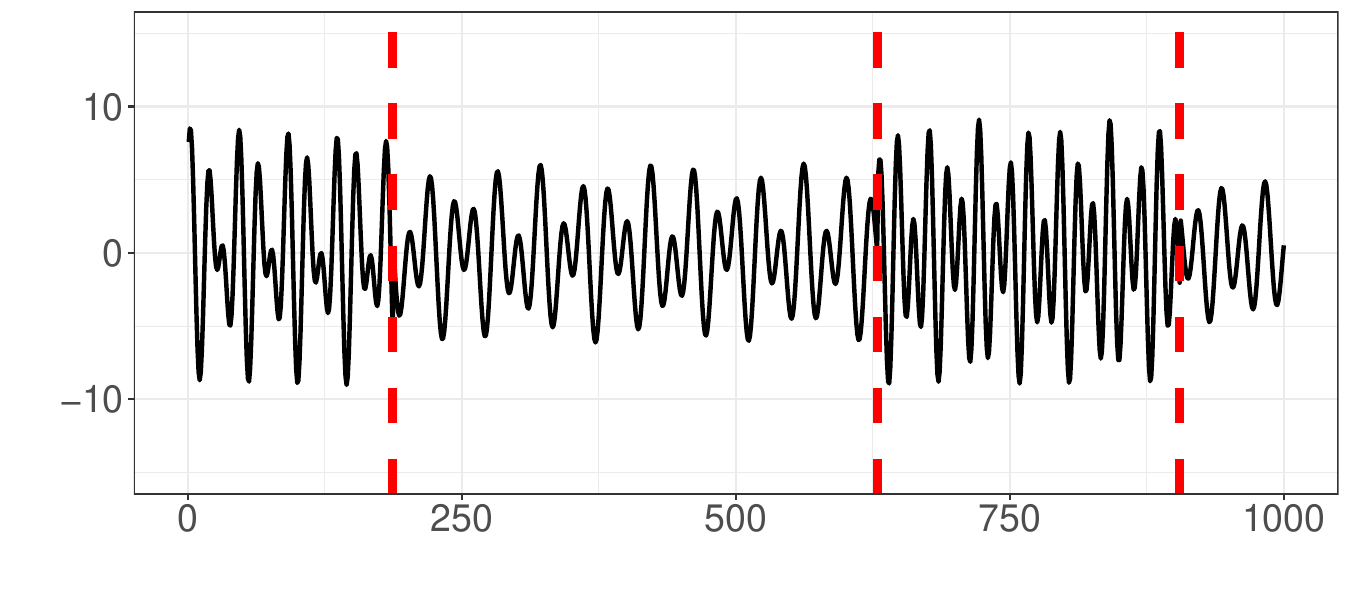}}
\subfigure[$N_E=3$]{\includegraphics[width=.47\textwidth]{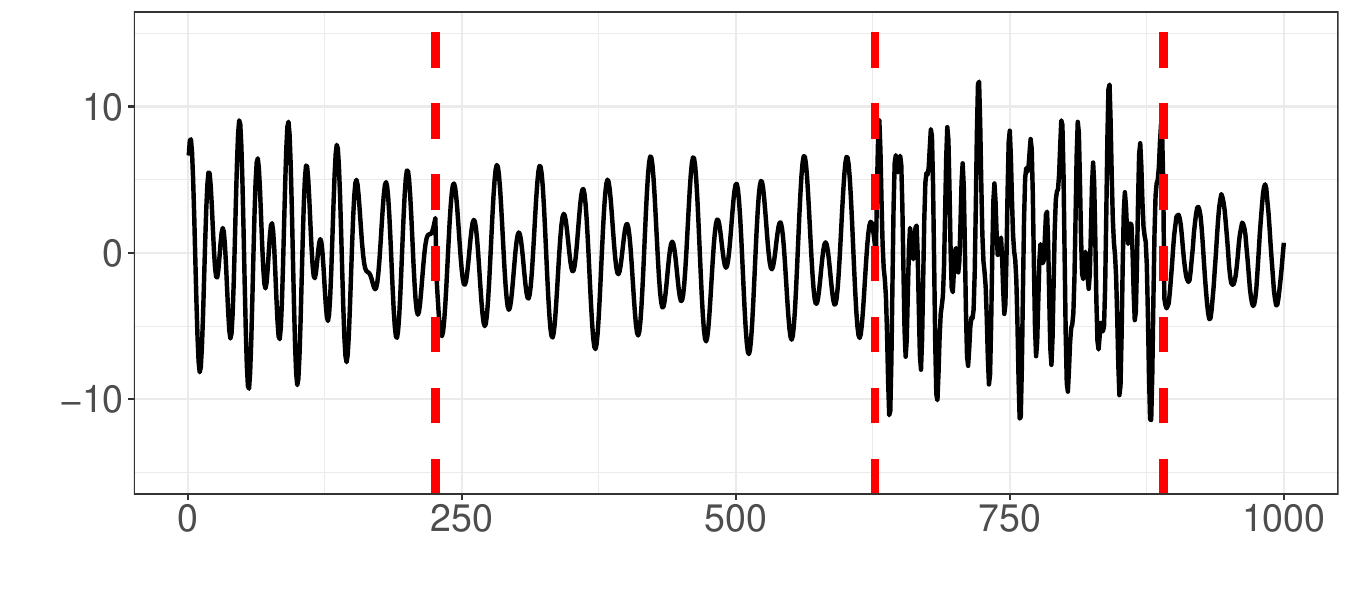}}\\
\subfigure[$N_E=5$]{\includegraphics[width=.47\textwidth]{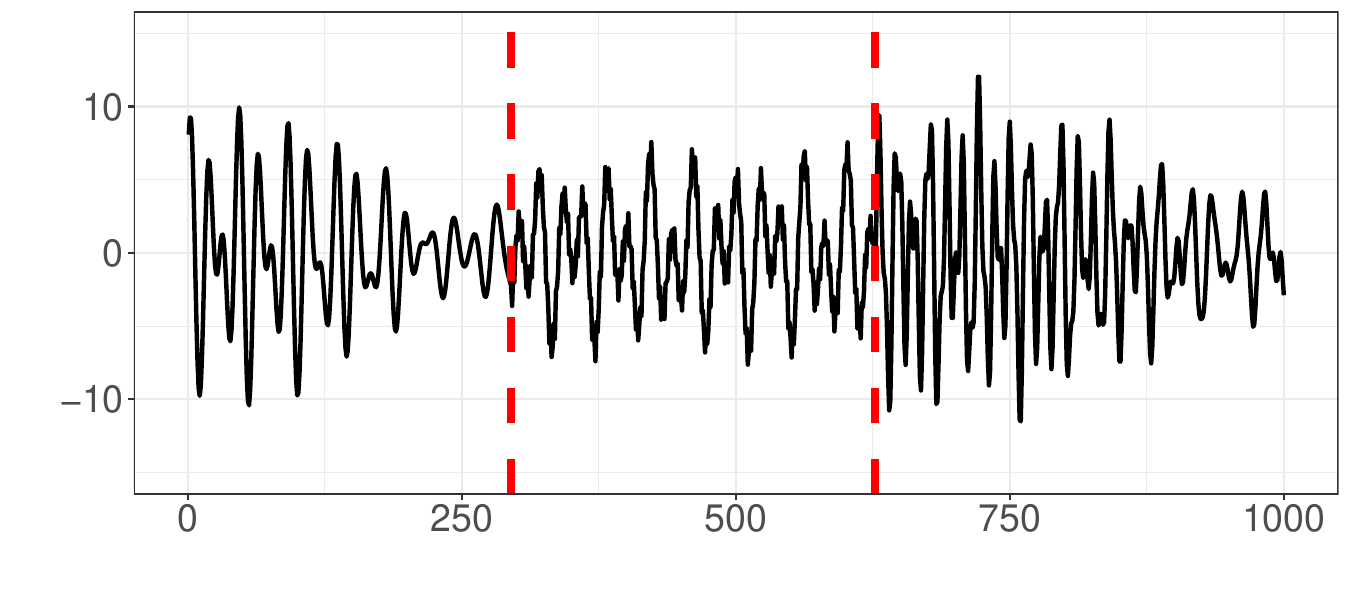}}
\subfigure[$N_E=10$]{\includegraphics[width=.47\textwidth]{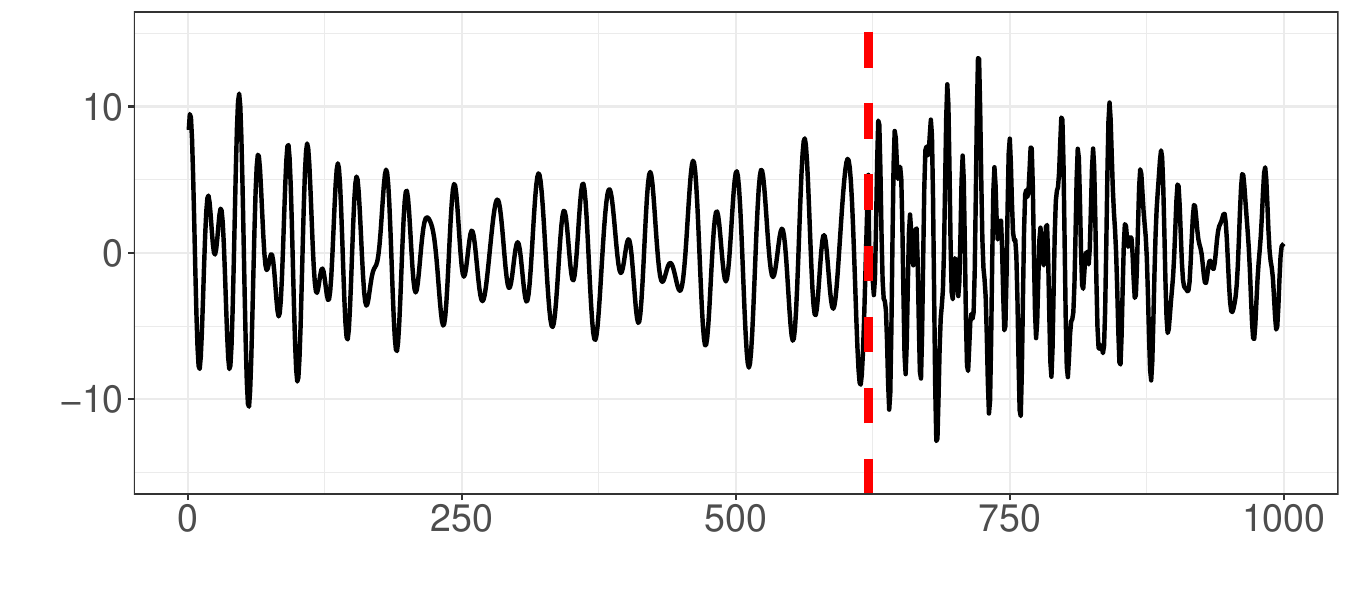}}
\vspace{0.5em}
\caption{{Choice of SuSiE effects $N_E$}. Segmentation obtained with different $N_E$. One replicate of Scenario 2 in Section 3 with $T=1000$ and $m=3$. The true change point locations are $(174, 622, 892)$ and $L_1=2$, $L_2=2$, $L_3=3$, and $L_4=2$. Panel (a) depicts the true signal along with the change points, the other panels depict \algonamesp estimated change points and means for growing $N_E$.}\label{fig:choice_Ne}
\end{figure}

\clearpage
\section{Simulations: additional information and results}
\label{app:more_sim}
We complement Section 3 with additional results and details of the scenarios not included in the main manuscript. \\

\subsection{Details on the data-generating mechanisms}
\label{sec:details_scenarios}
In this section, we provide the details on the data-generating mechanisms used in Section 4. The mean signal in Scenarios 1-3 is constructed starting from 6 baseline functions that are presented in Figure \ref{fig:signals}. Then, the scenarios differ for the number of baseline functions used and the way they are connected at the change-points. In addition, different scenarios also assume different noise. 
\begin{figure}[!ht]
\centering
    \includegraphics[width=0.8\textwidth]{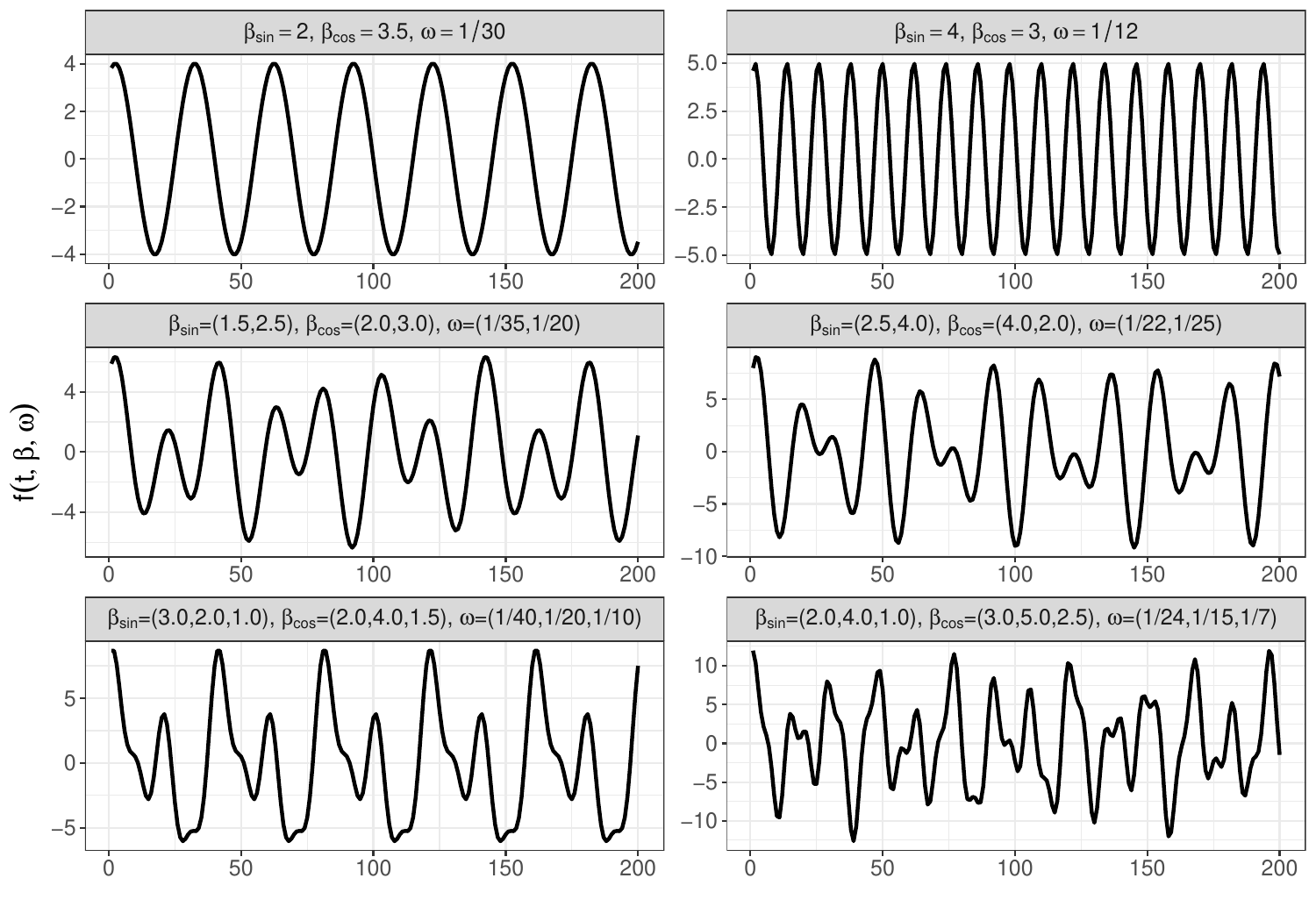}\\
    \caption{{Scenarios 1,2,3: true mean signals}. The right column depicts the signal in Scenario 1. These baseline functions are the same used in \cite{hadj2020bayesian}. All six mean signals are used to construct the mean function in Scenarios 2 and 3.}
    \label{fig:signals}
\end{figure}
A detailed description of each setting follows:
\begin{itemize}\setlength\itemsep{1em}
    \item \textbf{Scenario 1}. This Scenario follows the structure of the simulation study in \citep{hadj2020bayesian} to build the mean signal. The length of the time series $T=900$ and the number $m=2$ and locations of change-points $\{300,650\}$ are fixed. The data generating process is defined as $y_t=\mu_t+\varepsilon_t$ where:
   \begin{equation*}
    \resizebox{\textwidth}{!}{$
        \mu_t =\begin{cases}
            2\sin(2\pi \frac{1}{24} t)+3\cos(2\pi \frac{1}{24} t) + 4\sin(2\pi \frac{1}{15} t)+5\cos(2\pi \frac{1}{15} t) + \sin(2\pi \frac{1}{7} t)+2.5\cos(2\pi \frac{1}{7} t) &1\leq t \leq 300 \\
            4\sin(2\pi \frac{1}{12} t)+3\cos(2\pi \frac{1}{12} t) &300 < t \leq 650\\
            2.5\sin(2\pi \frac{1}{22} t)+4\cos(2\pi \frac{1}{22} t) + 4\sin(2\pi \frac{1}{25} t)+2\cos(2\pi \frac{1}{25} t) &650 < t \leq 900\\
        \end{cases},
    $}\end{equation*}
    where the size of the three segments are $L_1=3$, $L_2=1$, $L_3=2$. The choice of $\varepsilon_t$ lead to sub-scenarios:\\
    \begin{itemize}\setlength\itemsep{0.5em}
        \item \textbf{Scenario 1.a} Assumes stationary and Gaussian errors $\varepsilon_t\sim N(0,1)$ and $\varepsilon_t\sim N(0,3^2)$.
        \item \textbf{Scenario 1.b} Assumes stationary and t-student errors $\varepsilon_t\sim t_3$.
        \item \textbf{Scenario 1.c} Assumes the non-stationary model in \cite{wu2024frequency} (M1) such that $\varepsilon_t=0.5\cos(t/900)\varepsilon_{t-1}+e_t+0.3(t/900)e_{t-1}$, where $e_t\sim N(0,1)$.
    \end{itemize}
    \item \textbf{Scenario 2}. This scenario considers all the 6 baseline functions in Figure \ref{fig:signals}. We consider two different cases. The first one assumes increasing both sample size and the number of change points $(T,m)=\{(500,2),(1000,4),(5000,8),(10000,12)\}$, while the second one fixes the sample size $T=1000$ and gradually increases the number of change points $m=\{2,3,4,5,6\}$. For each combination $(T,m)$ and for each replicate, we sample new locations for the $m$ change points enforcing a minimum spacing between them ($100$ for $T=\{500,1000\}$, $200$ for $T=5000$, $300$ for $T=10000$). Then, we assign at each segment one of the pre-specified mean functions. Here the error is stationary and Gaussian $\varepsilon_t\sim N(0,3^2)$. The way the baseline functions are connected at the change points give raise to two cases:\\
    \begin{itemize}\setlength\itemsep{0.5em}
        \item \textbf{Scenario 2.a} Here, no continuity in $\mu_t$ at the change points is imposed.
        \item \textbf{Scenario 2.b} Here, we enforce continuity through a misspecification of the mean function in (1):  
        \begin{equation*}
            \mu_{t}= \sum_{l=1}^{L_{j}} \left(\beta_{jl}^{(1)} \sin (2\pi \omega_{jl}t+a_{jl})+\beta_{jl}^{(2)} \cos (2\pi \omega_{jl}t+a_{jl})\right),\qquad t_{j-1} < t \leq t_j,
        \end{equation*}
        where $a_{jl}$ introduces a shift, and it is optimized such that $\lim_{\epsilon\rightarrow 0}\mu_{t_j+\epsilon}=\mu_{t_j}$, therefore ensuring a continuous signal at the change-point locations. Although continuity is guaranteed, it might show a kink in $t_j$ more or less pronounced, since no constraints are imposed on the derivatives. An example is given in Figure \ref{fig:signals_cont}.
        \begin{figure}[!ht]
        \centering
            \includegraphics[width=0.7\textwidth]{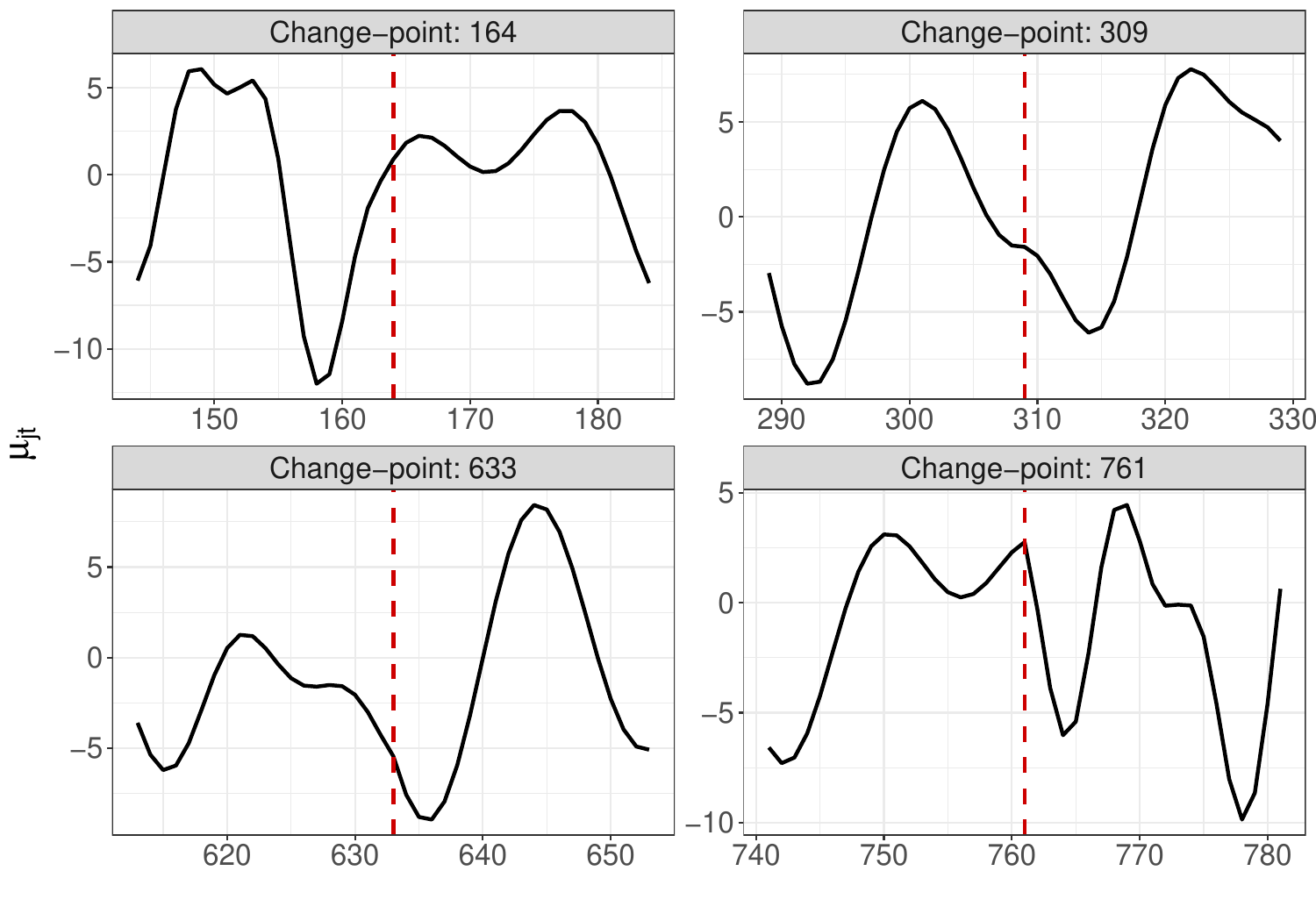}\\
            \caption{{Continuous signal at change points}. The figure depicts an example of a signal in Scenario 2.b with $n=1000$ and $m=4$ where the functions are connected to enforce continuity.}
            \label{fig:signals_cont}
        \end{figure}
    \end{itemize}
    \item \textbf{Scenario 3}. This scenario considers a multivariate, dense regime. We consider the multivariate model in (1) with $d=\{2,3,4,5,10,20\}$ series, fixed time series length $T=1000$ and two possible number of change points $m=4,6$. Similar to Scenario 2, for each replicate, we sample new locations for the $m$ change points. Moreover, no continuity in $\mu_{it}$ at the change points is imposed, for $i=1,\ldots,d$. Again, we have two sub-cases according to the nature of $\varepsilon_{it}$. Initially, the error is stationary and Gaussian with same standard deviation for all $i=1,\ldots,d$: $\varepsilon_{it}\sim N(0,3^2)$. Then, we introduce an increasing number of more noisy series. We set $d=10$, $m=4$, and consider two noise regimes: $\sigma_{i}=3$ for $i=1,\ldots,d_1$ and $\sigma_{i}=9$ for $i=d_1+1,\ldots,10$, where $d_1\in\{9,7,4\}$.
    \item \textbf{Scenario 4}. This scenario considers the piece-wise autoregressive process as in \cite{davis2006structural}. Specifically, the data-generating mechanism is:
    \begin{equation}
        y_t = \begin{cases}
            0.9 y_{t-1} + \varepsilon_t, \qquad\qquad\qquad\qquad 1 \leq t \leq 512 \\
            1.69 y_{t-1} - 0.81 y_{t-2}  + \varepsilon_t, \qquad 512 < t \leq 768 \\
            1.32 y_{t-1} - 0.81 y_{t-2}  + \varepsilon_t, \qquad 768 < t \leq 1024 \\
            \end{cases}, \qquad \varepsilon_t\sim N(0,1).
    \end{equation}
    \item \textbf{Scenario 5}. This scenario assumes a slowly varying autoregressive process as in \cite{davis2006structural}. The data-generating mechanism is a time-varying parameter autoregressive process: 
    \begin{align}
        y_t &= a_t y_{t-1} - 0.81 y_{t-2} + \varepsilon_t, \qquad \varepsilon_t \sim N(0,1), \qquad t=1,\ldots,1031, \label{eq:slow_auto1}\\
        a_t &= 0.8(1-0.5\cos(\pi t/1031)) \label{eq:slow_auto2},
    \end{align}
    where the parameter associated with the first lag changes gradually over time, whereas the one associated with the second lag remains constant.
    \item \textbf{Scenario 6}. This scenario considers a setting where no-signal and no-change point are present. The data-generating mechanism is a white noise process $y_t\sim N(0,3^2)$ for $t=1,\ldots,1000$.
\end{itemize}
Note that in Scenarios 1.a, 2.a, 3 \algonamesp is correctly specified, while in the remaining ones we consider model misspecification.

\subsection{Additional Tables and Figures}

% Table generated by Excel2LaTeX from sheet 'Foglio1'
\begin{table}[htbp]
  \centering
\resizebox{1\textwidth}{!}{
    \begin{tabular}{lrrrrrrrrrrrrrrr}
    \toprule
          &       & \multicolumn{4}{c}{$C(\mathcal{P},\hat{\mathcal{P}})$}       &       & \multicolumn{4}{c}{$H(\mathcal{P},\hat{\mathcal{P}})$}      &       & \multicolumn{4}{c}{$B(\mathcal{P},\hat{\mathcal{P}})$} \\
          &       & \multicolumn{1}{l}{$\sigma_j=1$} & \multicolumn{1}{l}{$\sigma_j=3$} & \multicolumn{1}{l}{t-Student} & \multicolumn{1}{l}{NS-noise} &       & \multicolumn{1}{l}{$\sigma_j=1$} & \multicolumn{1}{l}{$\sigma_j=3$} & \multicolumn{1}{l}{t-Student} & \multicolumn{1}{l}{NS-noise} &       & \multicolumn{1}{l}{$\sigma_j=1$} & \multicolumn{1}{l}{$\sigma_j=3$} & \multicolumn{1}{l}{t-Student} & \multicolumn{1}{l}{NS-noise} \\
          \cmidrule{1-1}\cmidrule{3-6}\cmidrule{8-11}\cmidrule{13-16}
    CPVS(1) &       & \cellcolor[rgb]{ .984,  .976,  .988}0,98 & \cellcolor[rgb]{ .984,  .973,  .984}0,98 & \cellcolor[rgb]{ .984,  .973,  .984}0,97 & \cellcolor[rgb]{ .984,  .984,  .996}0,98 &       & \cellcolor[rgb]{ .988,  .988,  1}0,01 & \cellcolor[rgb]{ .988,  .988,  1}0,01 & \cellcolor[rgb]{ .988,  .98,  .992}0,02 & \cellcolor[rgb]{ .988,  .988,  1}0,01 &       & \cellcolor[rgb]{ .988,  .988,  1}0,00 & \cellcolor[rgb]{ .988,  .988,  1}0,00 & \cellcolor[rgb]{ .988,  .973,  .984}0,06 & \cellcolor[rgb]{ .988,  .988,  1}0,00 \\
    CPVS(2) &       & \cellcolor[rgb]{ .984,  .976,  .988}0,98 & \cellcolor[rgb]{ .984,  .976,  .988}0,98 & \cellcolor[rgb]{ .984,  .984,  .996}0,98 & \cellcolor[rgb]{ .984,  .976,  .988}0,97 &       & \cellcolor[rgb]{ .988,  .988,  1}0,01 & \cellcolor[rgb]{ .988,  .988,  1}0,01 & \cellcolor[rgb]{ .988,  .988,  1}0,01 & \cellcolor[rgb]{ .988,  .984,  .996}0,01 &       & \cellcolor[rgb]{ .988,  .988,  1}0,00 & \cellcolor[rgb]{ .988,  .988,  1}0,00 & \cellcolor[rgb]{ .988,  .984,  .996}0,02 & \cellcolor[rgb]{ .988,  .98,  .992}0,04 \\
    CPVS  &       & \cellcolor[rgb]{ .984,  .98,  .992}0,98 & \cellcolor[rgb]{ .984,  .976,  .988}0,98 & \cellcolor[rgb]{ .984,  .984,  .996}0,98 & \cellcolor[rgb]{ .988,  .988,  1}0,98 &       & \cellcolor[rgb]{ .988,  .988,  1}0,01 & \cellcolor[rgb]{ .988,  .988,  1}0,01 & \cellcolor[rgb]{ .988,  .984,  .996}0,01 & \cellcolor[rgb]{ .988,  .988,  1}0,01 &       & \cellcolor[rgb]{ .988,  .988,  1}0,00 & \cellcolor[rgb]{ .988,  .988,  1}0,00 & \cellcolor[rgb]{ .988,  .973,  .984}0,06 & \cellcolor[rgb]{ .988,  .988,  1}0,00 \\
    ANOM(0.1) &       & \cellcolor[rgb]{ .984,  .965,  .976}0,97 & \cellcolor[rgb]{ .988,  .988,  1}0,99 & \cellcolor[rgb]{ .984,  .973,  .984}0,97 & \cellcolor[rgb]{ .984,  .965,  .976}0,97 &       & \cellcolor[rgb]{ .988,  .973,  .984}0,04 & \cellcolor[rgb]{ .988,  .988,  1}0,00 & \cellcolor[rgb]{ .988,  .925,  .937}0,06 & \cellcolor[rgb]{ .988,  .937,  .949}0,05 &       & \cellcolor[rgb]{ .988,  .925,  .937}0,22 & \cellcolor[rgb]{ .988,  .988,  1}0,00 & \cellcolor[rgb]{ .988,  .902,  .914}0,30 & \cellcolor[rgb]{ .988,  .91,  .922}0,28 \\
    ANOM(1) &       & \cellcolor[rgb]{ .984,  .949,  .961}0,95 & \cellcolor[rgb]{ .984,  .984,  .996}0,99 & \cellcolor[rgb]{ .984,  .961,  .973}0,96 & \cellcolor[rgb]{ .984,  .933,  .945}0,95 &       & \cellcolor[rgb]{ .988,  .961,  .973}0,05 & \cellcolor[rgb]{ .988,  .988,  1}0,00 & \cellcolor[rgb]{ .988,  .902,  .914}0,08 & \cellcolor[rgb]{ .988,  .922,  .933}0,06 &       & \cellcolor[rgb]{ .988,  .89,  .902}0,34 & \cellcolor[rgb]{ .988,  .988,  1}0,00 & \cellcolor[rgb]{ .988,  .875,  .886}0,40 & \cellcolor[rgb]{ .988,  .875,  .886}0,40 \\
    ANOM(10) &       & \cellcolor[rgb]{ .984,  .937,  .945}0,93 & \cellcolor[rgb]{ .984,  .984,  .996}0,99 & \cellcolor[rgb]{ .984,  .949,  .961}0,96 & \cellcolor[rgb]{ .984,  .91,  .922}0,93 &       & \cellcolor[rgb]{ .988,  .945,  .957}0,08 & \cellcolor[rgb]{ .988,  .988,  1}0,00 & \cellcolor[rgb]{ .988,  .863,  .871}0,11 & \cellcolor[rgb]{ .988,  .871,  .882}0,10 &       & \cellcolor[rgb]{ .984,  .843,  .855}0,50 & \cellcolor[rgb]{ .988,  .988,  1}0,00 & \cellcolor[rgb]{ .984,  .804,  .816}0,64 & \cellcolor[rgb]{ .984,  .78,  .792}0,72 \\
    AP    &       & \cellcolor[rgb]{ .984,  .906,  .918}0,90 & \cellcolor[rgb]{ .973,  .412,  .42}0,35 & \cellcolor[rgb]{ .976,  .608,  .62}0,73 & \cellcolor[rgb]{ .976,  .561,  .569}0,70 &       & \cellcolor[rgb]{ .988,  .929,  .941}0,10 & \cellcolor[rgb]{ .973,  .412,  .42}0,99 & \cellcolor[rgb]{ .98,  .565,  .573}0,34 & \cellcolor[rgb]{ .976,  .498,  .506}0,39 &       & \cellcolor[rgb]{ .988,  .925,  .937}0,22 & \cellcolor[rgb]{ .973,  .412,  .42}1,98 & \cellcolor[rgb]{ .984,  .769,  .78}0,76 & \cellcolor[rgb]{ .984,  .753,  .761}0,82 \\
    ADA(7) &       & \cellcolor[rgb]{ .984,  .961,  .973}0,96 & \cellcolor[rgb]{ .973,  .416,  .424}0,35 & \cellcolor[rgb]{ .976,  .671,  .682}0,77 & \cellcolor[rgb]{ .976,  .584,  .592}0,72 &       & \cellcolor[rgb]{ .988,  .98,  .992}0,02 & \cellcolor[rgb]{ .976,  .42,  .427}0,98 & \cellcolor[rgb]{ .984,  .714,  .725}0,22 & \cellcolor[rgb]{ .98,  .596,  .604}0,32 &       & \cellcolor[rgb]{ .988,  .973,  .984}0,06 & \cellcolor[rgb]{ .976,  .42,  .427}1,96 & \cellcolor[rgb]{ .984,  .78,  .792}0,72 & \cellcolor[rgb]{ .984,  .757,  .769}0,80 \\
    ADA(15) &       & \cellcolor[rgb]{ .984,  .961,  .973}0,96 & \cellcolor[rgb]{ .973,  .427,  .435}0,36 & \cellcolor[rgb]{ .984,  .875,  .882}0,91 & \cellcolor[rgb]{ .984,  .925,  .937}0,94 &       & \cellcolor[rgb]{ .988,  .976,  .988}0,03 & \cellcolor[rgb]{ .976,  .431,  .439}0,96 & \cellcolor[rgb]{ .988,  .867,  .878}0,10 & \cellcolor[rgb]{ .988,  .937,  .949}0,05 &       & \cellcolor[rgb]{ .988,  .933,  .945}0,20 & \cellcolor[rgb]{ .976,  .431,  .439}1,92 & \cellcolor[rgb]{ .988,  .878,  .89}0,38 & \cellcolor[rgb]{ .988,  .969,  .98}0,08 \\
    {LSW}   &       & \cellcolor[rgb]{ .976,  .69,  .698}0,66 & \cellcolor[rgb]{ .973,  .467,  .475}0,41 & \cellcolor[rgb]{ .973,  .412,  .42}0,60 & \cellcolor[rgb]{ .973,  .486,  .494}0,65 &       & \cellcolor[rgb]{ .984,  .757,  .769}0,40 & \cellcolor[rgb]{ .976,  .482,  .49}0,87 & \cellcolor[rgb]{ .973,  .412,  .42}0,46 & \cellcolor[rgb]{ .976,  .478,  .486}0,41 &       & \cellcolor[rgb]{ .98,  .698,  .71}1,00 & \cellcolor[rgb]{ .976,  .471,  .478}1,78 & \cellcolor[rgb]{ .98,  .686,  .698}1,04 & \cellcolor[rgb]{ .98,  .694,  .702}1,02 \\
    {OBMB}  &       & \cellcolor[rgb]{ .984,  .949,  .961}0,95 & \cellcolor[rgb]{ .98,  .843,  .851}0,83 & \cellcolor[rgb]{ .984,  .91,  .922}0,93 & \cellcolor[rgb]{ .984,  .847,  .859}0,89 &       & \cellcolor[rgb]{ .988,  .976,  .988}0,03 & \cellcolor[rgb]{ .988,  .894,  .902}0,17 & \cellcolor[rgb]{ .988,  .925,  .937}0,06 & \cellcolor[rgb]{ .988,  .875,  .886}0,10 &       & \cellcolor[rgb]{ .988,  .973,  .984}0,06 & \cellcolor[rgb]{ .988,  .863,  .875}0,44 & \cellcolor[rgb]{ .988,  .949,  .961}0,14 & \cellcolor[rgb]{ .988,  .922,  .933}0,24 \\
    \bottomrule
    \end{tabular}%
  }
    \vspace{0.5em}
  \caption{{Scenario 1 (a-c): segmentation performance}. CPVS (our proposal), ANOM \citep{hadj2020bayesian}, AP \citep{davis2006structural}, ADA \citep{rosen2012adaptspec}, {LSW \citep{korkas2017multiple}, and OBMB \citep{wu2024frequency} are considered}. Three metrics: segmentation coverage $C(\mathcal{P},\hat{\mathcal{P}})\in (0,1)$ (the higher, the better), Hausdorff distance $H(\mathcal{P},\hat{\mathcal{P}})\in (0,1)$ (the lower, the better), and bias $B(m,\hat{m})\in [0,m]$ (the lower, the better). The table shows the results for Gaussian noise with two choices of the standard deviation $\sigma_j=1,3$, t-Student errors with $3$ degrees of freedom, {and the non-stationary noise used in \cite{wu2024frequency} (model M1 in their paper)}. The red scale depicts worse performances.}  \label{tab:results_Beniamino}%
\end{table}%

% Table generated by Excel2LaTeX from sheet 'Foglio1'
\begin{table}[htbp]
  \centering
  \resizebox{0.95\textwidth}{!}{
    \begin{tabular}{clrrrrrrrrrrr}
    \toprule
          &      &       & \multicolumn{4}{c}{$(T,\,m)$}        &       & \multicolumn{5}{c}{$(T=1000,\,m)$} \\
          &   Method    &       & $(500,2)$ & $(1000,4)$ & $(5000,8)$ & $(10000,12)$ &       & $2$     & $3$     & $4$     & $5$     & $6$ \\
              \cmidrule{2-2}\cmidrule{4-7}\cmidrule{9-13}
    \multirow{6}[0]{*}{\begin{turn}{90}$RMSE(\mu_t,\hat{\mu}_t)$\end{turn}} & CPVS(1) &       & \cellcolor[rgb]{ .98,  .592,  .6}1,26 & \cellcolor[rgb]{ .98,  .6,  .608}1,24 & \cellcolor[rgb]{ .98,  .671,  .678}1,09 & \cellcolor[rgb]{ .98,  .624,  .631}1,19 &       & \cellcolor[rgb]{ .988,  .859,  .871}0,99 & \cellcolor[rgb]{ .988,  .851,  .863}1,02 & \cellcolor[rgb]{ .984,  .788,  .796}1,24 & \cellcolor[rgb]{ .984,  .761,  .769}1,34 & \cellcolor[rgb]{ .984,  .725,  .737}1,45 \\
          & CPVS(2) &       & \cellcolor[rgb]{ .98,  .651,  .659}1,13 & \cellcolor[rgb]{ .98,  .635,  .643}1,17 & \cellcolor[rgb]{ .98,  .675,  .686}1,08 & \cellcolor[rgb]{ .98,  .616,  .627}1,20 &       & \cellcolor[rgb]{ .988,  .894,  .906}0,87 & \cellcolor[rgb]{ .988,  .859,  .867}1,00 & \cellcolor[rgb]{ .984,  .808,  .82}1,17 & \cellcolor[rgb]{ .984,  .773,  .784}1,28 & \cellcolor[rgb]{ .984,  .769,  .78}1,30 \\
          & CPVS  &       & \cellcolor[rgb]{ .98,  .686,  .698}1,05 & \cellcolor[rgb]{ .98,  .663,  .671}1,11 & \cellcolor[rgb]{ .98,  .694,  .702}1,04 & \cellcolor[rgb]{ .98,  .635,  .643}1,16 &       & \cellcolor[rgb]{ .988,  .894,  .906}0,86 & \cellcolor[rgb]{ .988,  .871,  .882}0,95 & \cellcolor[rgb]{ .984,  .824,  .835}1,11 & \cellcolor[rgb]{ .984,  .78,  .792}1,26 & \cellcolor[rgb]{ .984,  .765,  .773}1,32 \\
          & ANOM(0.1) &       & \cellcolor[rgb]{ .98,  .573,  .58}1,30 & \cellcolor[rgb]{ .973,  .412,  .42}1,63 & \cellcolor[rgb]{ .976,  .533,  .541}1,38 & \cellcolor[rgb]{ .976,  .427,  .435}1,60 &       & \cellcolor[rgb]{ .988,  .957,  .969}0,66 & \cellcolor[rgb]{ .988,  .859,  .871}0,99 & \cellcolor[rgb]{ .98,  .675,  .682}1,63 & \cellcolor[rgb]{ .976,  .537,  .549}2,10 & \cellcolor[rgb]{ .973,  .412,  .42}2,53 \\
          & ANOM(1) &       & \cellcolor[rgb]{ .98,  .686,  .694}1,06 & \cellcolor[rgb]{ .98,  .663,  .675}1,10 & \cellcolor[rgb]{ .988,  .922,  .933}0,56 & \cellcolor[rgb]{ .988,  .878,  .89}0,65 &       & \cellcolor[rgb]{ .988,  .988,  1}0,54 & \cellcolor[rgb]{ .988,  .949,  .961}0,69 & \cellcolor[rgb]{ .984,  .827,  .839}1,10 & \cellcolor[rgb]{ .984,  .733,  .741}1,43 & \cellcolor[rgb]{ .976,  .549,  .557}2,06 \\
          & ANOM(10) &       & \cellcolor[rgb]{ .984,  .749,  .761}0,92 & \cellcolor[rgb]{ .984,  .765,  .773}0,89 & \cellcolor[rgb]{ .988,  .965,  .976}0,46 & \cellcolor[rgb]{ .988,  .988,  1}0,41 &       & \cellcolor[rgb]{ .988,  .988,  1}0,54 & \cellcolor[rgb]{ .988,  .957,  .969}0,65 & \cellcolor[rgb]{ .988,  .886,  .898}0,89 & \cellcolor[rgb]{ .984,  .792,  .8}1,23 & \cellcolor[rgb]{ .984,  .718,  .725}1,48 \\
          &       &       &       &       &       &       &       &       &       &       &       &  \\
    \multirow{6}[0]{*}{\begin{turn}{90}$RMSE(y_t,\hat{y}_t)$\end{turn}} & CPVS(1) &       & \cellcolor[rgb]{ .988,  .902,  .914}3,02 & \cellcolor[rgb]{ .984,  .808,  .82}3,09 & \cellcolor[rgb]{ .984,  .753,  .761}3,13 & \cellcolor[rgb]{ .98,  .663,  .671}3,19 &       & \cellcolor[rgb]{ .988,  .961,  .973}3,00 & \cellcolor[rgb]{ .988,  .957,  .969}3,01 & \cellcolor[rgb]{ .988,  .906,  .918}3,09 & \cellcolor[rgb]{ .988,  .894,  .906}3,11 & \cellcolor[rgb]{ .988,  .863,  .875}3,15 \\
          & CPVS(2) &       & \cellcolor[rgb]{ .988,  .953,  .965}2,99 & \cellcolor[rgb]{ .988,  .863,  .871}3,05 & \cellcolor[rgb]{ .984,  .757,  .769}3,12 & \cellcolor[rgb]{ .98,  .663,  .671}3,19 &       & \cellcolor[rgb]{ .988,  .976,  .988}2,98 & \cellcolor[rgb]{ .988,  .961,  .973}3,01 & \cellcolor[rgb]{ .988,  .929,  .941}3,05 & \cellcolor[rgb]{ .988,  .91,  .922}3,08 & \cellcolor[rgb]{ .988,  .91,  .922}3,08 \\
          & CPVS  &       & \cellcolor[rgb]{ .988,  .988,  1}2,96 & \cellcolor[rgb]{ .988,  .89,  .902}3,03 & \cellcolor[rgb]{ .984,  .776,  .788}3,11 & \cellcolor[rgb]{ .98,  .678,  .686}3,18 &       & \cellcolor[rgb]{ .988,  .976,  .988}2,99 & \cellcolor[rgb]{ .988,  .973,  .984}2,99 & \cellcolor[rgb]{ .988,  .945,  .957}3,03 & \cellcolor[rgb]{ .988,  .914,  .925}3,08 & \cellcolor[rgb]{ .988,  .902,  .914}3,09 \\
          & ANOM(0.1) &       & \cellcolor[rgb]{ .984,  .749,  .761}3,13 & \cellcolor[rgb]{ .976,  .506,  .514}3,30 & \cellcolor[rgb]{ .98,  .592,  .6}3,24 & \cellcolor[rgb]{ .973,  .412,  .42}3,36 &       & \cellcolor[rgb]{ .988,  .976,  .988}2,98 & \cellcolor[rgb]{ .988,  .929,  .941}3,05 & \cellcolor[rgb]{ .984,  .765,  .776}3,30 & \cellcolor[rgb]{ .98,  .612,  .62}3,53 & \cellcolor[rgb]{ .973,  .412,  .42}3,82 \\
          & ANOM(1) &       & \cellcolor[rgb]{ .988,  .882,  .894}3,04 & \cellcolor[rgb]{ .984,  .847,  .855}3,06 & \cellcolor[rgb]{ .988,  .945,  .957}2,99 & \cellcolor[rgb]{ .988,  .89,  .902}3,03 &       & \cellcolor[rgb]{ .988,  .988,  1}2,96 & \cellcolor[rgb]{ .988,  .984,  .996}2,97 & \cellcolor[rgb]{ .988,  .922,  .933}3,06 & \cellcolor[rgb]{ .988,  .847,  .859}3,17 & \cellcolor[rgb]{ .98,  .624,  .635}3,51 \\
          & ANOM(10) &       & \cellcolor[rgb]{ .988,  .953,  .965}2,99 & \cellcolor[rgb]{ .988,  .941,  .953}3,00 & \cellcolor[rgb]{ .988,  .969,  .98}2,98 & \cellcolor[rgb]{ .988,  .945,  .957}2,99 &       & \cellcolor[rgb]{ .988,  .988,  1}2,96 & \cellcolor[rgb]{ .988,  .988,  1}2,96 & \cellcolor[rgb]{ .988,  .969,  .98}3,00 & \cellcolor[rgb]{ .988,  .906,  .918}3,09 & \cellcolor[rgb]{ .984,  .839,  .847}3,19 \\
    \bottomrule
    \end{tabular}%
  }
  \vspace{0.5em}
\caption{{{Scenario 2.a: signal estimation and data fitting}. CPVS is our proposal, ANOM by \cite{hadj2020bayesian}. Two metrics: root mean squared error between true signal and estimated one $RMSE(\mu_t,\hat{\mu}_t)$ (the lower, the better), root mean squared error between observed data and fitted values $RMSE(y_t,\hat{y}_t)$ (the lower, the better). The table shows the results in both varying $T$ and fixed $T=1000$ scenarios. The red scale depicts worse performances.}}\label{tab:resultsMSE1}%
\end{table}%

% Table generated by Excel2LaTeX from sheet 'Foglio1'
\begin{table}[htbp]
  \centering
    \resizebox{1\textwidth}{!}{
    \begin{tabular}{clrrrrrrrrrrr}
    \toprule
          &      &       & \multicolumn{4}{c}{$(T,\,m)$}        &       & \multicolumn{5}{c}{$(T=1000,\,m)$} \\
          &   Method    &       & $(500,2)$ & $(1000,4)$ & $(5000,8)$ & $(10000,12)$ &       & $2$     & $3$     & $4$     & $5$     & $6$ \\
              \cmidrule{2-2}\cmidrule{4-7}\cmidrule{9-13}
    \multirow{6}[0]{*}{\begin{turn}{90}$C(\mathcal{P},\hat{\mathcal{P}})$\end{turn}} & CPVS(1) &       & \cellcolor[rgb]{ .984,  .886,  .898}0,88 & \cellcolor[rgb]{ .98,  .824,  .835}0,82 & \cellcolor[rgb]{ .98,  .831,  .843}0,82 & \cellcolor[rgb]{ .98,  .788,  .8}0,78 &       & \cellcolor[rgb]{ .98,  .835,  .843}0,83 & \cellcolor[rgb]{ .98,  .835,  .847}0,83 & \cellcolor[rgb]{ .98,  .824,  .835}0,82 & \cellcolor[rgb]{ .98,  .765,  .773}0,75 & \cellcolor[rgb]{ .98,  .745,  .753}0,73 \\
          & CPVS(2) &       & \cellcolor[rgb]{ .984,  .886,  .898}0,88 & \cellcolor[rgb]{ .98,  .824,  .831}0,81 & \cellcolor[rgb]{ .98,  .792,  .8}0,78 & \cellcolor[rgb]{ .98,  .761,  .769}0,75 &       & \cellcolor[rgb]{ .984,  .855,  .867}0,85 & \cellcolor[rgb]{ .984,  .843,  .855}0,84 & \cellcolor[rgb]{ .98,  .824,  .831}0,81 & \cellcolor[rgb]{ .98,  .765,  .776}0,76 & \cellcolor[rgb]{ .976,  .686,  .698}0,67 \\
          & CPVS  &       & \cellcolor[rgb]{ .984,  .914,  .925}0,91 & \cellcolor[rgb]{ .984,  .847,  .859}0,84 & \cellcolor[rgb]{ .98,  .839,  .851}0,83 & \cellcolor[rgb]{ .98,  .796,  .808}0,79 &       & \cellcolor[rgb]{ .984,  .878,  .89}0,87 & \cellcolor[rgb]{ .984,  .878,  .886}0,87 & \cellcolor[rgb]{ .984,  .847,  .859}0,84 & \cellcolor[rgb]{ .98,  .788,  .8}0,78 & \cellcolor[rgb]{ .98,  .741,  .753}0,73 \\
          & ANOM(0.1) &       & \cellcolor[rgb]{ .984,  .851,  .863}0,84 & \cellcolor[rgb]{ .98,  .722,  .733}0,71 & \cellcolor[rgb]{ .98,  .733,  .745}0,72 & \cellcolor[rgb]{ .976,  .631,  .643}0,62 &       & \cellcolor[rgb]{ .984,  .969,  .98}0,97 & \cellcolor[rgb]{ .984,  .859,  .867}0,85 & \cellcolor[rgb]{ .98,  .722,  .733}0,71 & \cellcolor[rgb]{ .973,  .537,  .545}0,51 & \cellcolor[rgb]{ .973,  .412,  .42}0,38 \\
          & ANOM(1) &       & \cellcolor[rgb]{ .984,  .91,  .922}0,91 & \cellcolor[rgb]{ .984,  .898,  .91}0,90 & \cellcolor[rgb]{ .984,  .933,  .945}0,93 & \cellcolor[rgb]{ .984,  .922,  .929}0,92 &       & \cellcolor[rgb]{ .984,  .984,  .996}0,98 & \cellcolor[rgb]{ .984,  .957,  .969}0,96 & \cellcolor[rgb]{ .984,  .898,  .91}0,90 & \cellcolor[rgb]{ .98,  .725,  .733}0,71 & \cellcolor[rgb]{ .976,  .592,  .6}0,57 \\
          & ANOM(10) &       & \cellcolor[rgb]{ .984,  .925,  .937}0,92 & \cellcolor[rgb]{ .984,  .914,  .925}0,91 & \cellcolor[rgb]{ .984,  .976,  .988}0,97 & \cellcolor[rgb]{ .984,  .973,  .984}0,97 &       & \cellcolor[rgb]{ .988,  .988,  1}0,99 & \cellcolor[rgb]{ .984,  .976,  .988}0,98 & \cellcolor[rgb]{ .984,  .914,  .925}0,91 & \cellcolor[rgb]{ .98,  .808,  .82}0,80 & \cellcolor[rgb]{ .98,  .733,  .741}0,72 \\
          &       &       &       &       &       &       &       &       &       &       &       &  \\
    \multirow{6}[0]{*}{\begin{turn}{90}$H(\mathcal{P},\hat{\mathcal{P}})$\end{turn}} & CPVS(1) &       & \cellcolor[rgb]{ .988,  .918,  .929}0,07 & \cellcolor[rgb]{ .988,  .898,  .906}0,09 & \cellcolor[rgb]{ .988,  .91,  .922}0,08 & \cellcolor[rgb]{ .988,  .918,  .929}0,07 &       & \cellcolor[rgb]{ .984,  .843,  .855}0,14 & \cellcolor[rgb]{ .988,  .875,  .886}0,11 & \cellcolor[rgb]{ .988,  .898,  .906}0,09 & \cellcolor[rgb]{ .988,  .871,  .882}0,11 & \cellcolor[rgb]{ .988,  .882,  .89}0,11 \\
          & CPVS(2) &       & \cellcolor[rgb]{ .988,  .925,  .937}0,07 & \cellcolor[rgb]{ .988,  .906,  .918}0,08 & \cellcolor[rgb]{ .988,  .902,  .91}0,09 & \cellcolor[rgb]{ .988,  .914,  .925}0,08 &       & \cellcolor[rgb]{ .988,  .867,  .878}0,12 & \cellcolor[rgb]{ .988,  .875,  .882}0,11 & \cellcolor[rgb]{ .988,  .906,  .918}0,08 & \cellcolor[rgb]{ .988,  .886,  .894}0,10 & \cellcolor[rgb]{ .988,  .851,  .863}0,13 \\
          & CPVS  &       & \cellcolor[rgb]{ .988,  .929,  .941}0,06 & \cellcolor[rgb]{ .988,  .906,  .918}0,08 & \cellcolor[rgb]{ .988,  .91,  .922}0,08 & \cellcolor[rgb]{ .988,  .925,  .933}0,07 &       & \cellcolor[rgb]{ .988,  .89,  .898}0,10 & \cellcolor[rgb]{ .988,  .91,  .922}0,08 & \cellcolor[rgb]{ .988,  .906,  .918}0,08 & \cellcolor[rgb]{ .988,  .886,  .898}0,10 & \cellcolor[rgb]{ .988,  .871,  .882}0,11 \\
          & ANOM(0.1) &       & \cellcolor[rgb]{ .984,  .808,  .82}0,17 & \cellcolor[rgb]{ .98,  .671,  .682}0,29 & \cellcolor[rgb]{ .984,  .733,  .745}0,24 & \cellcolor[rgb]{ .984,  .753,  .765}0,22 &       & \cellcolor[rgb]{ .988,  .969,  .98}0,03 & \cellcolor[rgb]{ .984,  .824,  .835}0,15 & \cellcolor[rgb]{ .98,  .671,  .682}0,29 & \cellcolor[rgb]{ .976,  .502,  .51}0,44 & \cellcolor[rgb]{ .973,  .412,  .42}0,52 \\
          & ANOM(1) &       & \cellcolor[rgb]{ .988,  .89,  .898}0,10 & \cellcolor[rgb]{ .988,  .894,  .906}0,09 & \cellcolor[rgb]{ .988,  .929,  .941}0,06 & \cellcolor[rgb]{ .988,  .929,  .941}0,06 &       & \cellcolor[rgb]{ .988,  .988,  1}0,01 & \cellcolor[rgb]{ .988,  .965,  .976}0,03 & \cellcolor[rgb]{ .988,  .894,  .906}0,09 & \cellcolor[rgb]{ .984,  .725,  .733}0,24 & \cellcolor[rgb]{ .98,  .596,  .604}0,36 \\
          & ANOM(10) &       & \cellcolor[rgb]{ .988,  .914,  .925}0,08 & \cellcolor[rgb]{ .988,  .91,  .922}0,08 & \cellcolor[rgb]{ .988,  .976,  .988}0,02 & \cellcolor[rgb]{ .988,  .976,  .988}0,02 &       & \cellcolor[rgb]{ .988,  .988,  1}0,01 & \cellcolor[rgb]{ .988,  .984,  .996}0,01 & \cellcolor[rgb]{ .988,  .91,  .922}0,08 & \cellcolor[rgb]{ .984,  .804,  .812}0,17 & \cellcolor[rgb]{ .984,  .741,  .753}0,23 \\
          &       &       &       &       &       &       &       &       &       &       &       &  \\
    \multirow{6}[0]{*}{\begin{turn}{90}$B(m,\widehat{m})$\end{turn}} & CPVS(1) &       & \cellcolor[rgb]{ .988,  .98,  .992}0,20 & \cellcolor[rgb]{ .988,  .933,  .945}0,70 & \cellcolor[rgb]{ .988,  .871,  .882}1,34 & \cellcolor[rgb]{ .984,  .839,  .851}1,66 &       & \cellcolor[rgb]{ .988,  .918,  .929}0,56 & \cellcolor[rgb]{ .988,  .925,  .937}0,50 & \cellcolor[rgb]{ .988,  .902,  .91}0,70 & \cellcolor[rgb]{ .988,  .875,  .886}0,90 & \cellcolor[rgb]{ .988,  .855,  .867}1,04 \\
          & CPVS(2) &       & \cellcolor[rgb]{ .988,  .984,  .996}0,18 & \cellcolor[rgb]{ .988,  .949,  .961}0,54 & \cellcolor[rgb]{ .988,  .882,  .894}1,20 & \cellcolor[rgb]{ .984,  .804,  .816}2,02 &       & \cellcolor[rgb]{ .988,  .933,  .945}0,46 & \cellcolor[rgb]{ .988,  .933,  .945}0,44 & \cellcolor[rgb]{ .988,  .922,  .933}0,54 & \cellcolor[rgb]{ .988,  .882,  .894}0,84 & \cellcolor[rgb]{ .984,  .82,  .831}1,30 \\
          & CPVS  &       & \cellcolor[rgb]{ .988,  .988,  1}0,12 & \cellcolor[rgb]{ .988,  .949,  .961}0,54 & \cellcolor[rgb]{ .988,  .882,  .894}1,20 & \cellcolor[rgb]{ .984,  .827,  .839}1,76 &       & \cellcolor[rgb]{ .988,  .949,  .961}0,32 & \cellcolor[rgb]{ .988,  .961,  .973}0,24 & \cellcolor[rgb]{ .988,  .922,  .933}0,54 & \cellcolor[rgb]{ .988,  .882,  .894}0,84 & \cellcolor[rgb]{ .984,  .835,  .847}1,18 \\
          & ANOM(0.1) &       & \cellcolor[rgb]{ .988,  .957,  .969}0,46 & \cellcolor[rgb]{ .984,  .843,  .855}1,62 & \cellcolor[rgb]{ .98,  .694,  .702}3,14 & \cellcolor[rgb]{ .973,  .412,  .42}6,00 &       & \cellcolor[rgb]{ .988,  .98,  .992}0,10 & \cellcolor[rgb]{ .988,  .906,  .918}0,66 & \cellcolor[rgb]{ .984,  .78,  .788}1,62 & \cellcolor[rgb]{ .98,  .584,  .592}3,10 & \cellcolor[rgb]{ .973,  .412,  .42}4,38 \\
          & ANOM(1) &       & \cellcolor[rgb]{ .988,  .976,  .988}0,26 & \cellcolor[rgb]{ .988,  .957,  .969}0,46 & \cellcolor[rgb]{ .988,  .933,  .945}0,68 & \cellcolor[rgb]{ .988,  .867,  .878}1,38 &       & \cellcolor[rgb]{ .988,  .988,  1}0,04 & \cellcolor[rgb]{ .988,  .976,  .988}0,12 & \cellcolor[rgb]{ .988,  .933,  .945}0,46 & \cellcolor[rgb]{ .984,  .757,  .769}1,78 & \cellcolor[rgb]{ .98,  .592,  .604}3,02 \\
          & ANOM(10) &       & \cellcolor[rgb]{ .988,  .984,  .996}0,18 & \cellcolor[rgb]{ .988,  .961,  .973}0,40 & \cellcolor[rgb]{ .988,  .984,  .996}0,16 & \cellcolor[rgb]{ .988,  .969,  .976}0,36 &       & \cellcolor[rgb]{ .988,  .988,  1}0,02 & \cellcolor[rgb]{ .988,  .984,  .996}0,06 & \cellcolor[rgb]{ .988,  .941,  .953}0,40 & \cellcolor[rgb]{ .984,  .835,  .843}1,20 & \cellcolor[rgb]{ .984,  .745,  .753}1,88 \\
    \bottomrule
    \end{tabular}%
  }
  \vspace{0.5em}
\caption{{{Scenario 2.b: segmentation performance}. CPVS is our proposal, ANOM by \cite{hadj2020bayesian}. Three metrics: segmentation coverage $C(\mathcal{P},\hat{\mathcal{P}})\in (0,1)$ (the higher, the better), Hausdorff distance $H(\mathcal{P},\hat{\mathcal{P}})\in (0,1)$ (the lower, the better), and bias $B(m,\hat{m})\in [0,m]$ (the lower, the better). The table shows the results in both varying $T$ and fixed $T=1000$ scenarios where the true signal is forced to be continuous at the change points. The red scale depicts worse performances.}}\label{tab:results1_cont}%
\end{table}%

% Table generated by Excel2LaTeX from sheet 'Foglio1'
\begin{table}[htbp]
  \centering
  \resizebox{0.95\textwidth}{!}{
    \begin{tabular}{clrrrrrrrrrrr}
    \toprule
          &      &       & \multicolumn{4}{c}{$(T,\,m)$}        &       & \multicolumn{5}{c}{$(T=1000,\,m)$} \\
          &   Method    &       & $(500,2)$ & $(1000,4)$ & $(5000,8)$ & $(10000,12)$ &       & $2$     & $3$     & $4$     & $5$     & $6$ \\
              \cmidrule{2-2}\cmidrule{4-7}\cmidrule{9-13}
    \multirow{6}[0]{*}{\begin{turn}{90}$RMSE(\mu_t,\hat{\mu}_t)$\end{turn}} & CPVS(1) &       & \cellcolor[rgb]{ .98,  .686,  .698}1,29 & \cellcolor[rgb]{ .98,  .573,  .58}1,61 & \cellcolor[rgb]{ .976,  .498,  .506}1,82 & \cellcolor[rgb]{ .976,  .431,  .439}2,01 &       & \cellcolor[rgb]{ .984,  .812,  .824}1,15 & \cellcolor[rgb]{ .984,  .722,  .733}1,46 & \cellcolor[rgb]{ .98,  .682,  .69}1,61 & \cellcolor[rgb]{ .98,  .627,  .635}1,80 & \cellcolor[rgb]{ .98,  .592,  .6}1,93 \\
          & CPVS(2) &       & \cellcolor[rgb]{ .984,  .702,  .714}1,24 & \cellcolor[rgb]{ .98,  .596,  .604}1,55 & \cellcolor[rgb]{ .976,  .49,  .498}1,84 & \cellcolor[rgb]{ .973,  .412,  .42}2,06 &       & \cellcolor[rgb]{ .984,  .827,  .835}1,10 & \cellcolor[rgb]{ .984,  .749,  .757}1,37 & \cellcolor[rgb]{ .98,  .698,  .71}1,55 & \cellcolor[rgb]{ .98,  .639,  .651}1,75 & \cellcolor[rgb]{ .98,  .576,  .584}1,98 \\
          & CPVS  &       & \cellcolor[rgb]{ .984,  .725,  .737}1,18 & \cellcolor[rgb]{ .98,  .616,  .627}1,49 & \cellcolor[rgb]{ .976,  .518,  .529}1,76 & \cellcolor[rgb]{ .976,  .439,  .447}1,98 &       & \cellcolor[rgb]{ .984,  .831,  .843}1,08 & \cellcolor[rgb]{ .984,  .776,  .788}1,27 & \cellcolor[rgb]{ .984,  .714,  .725}1,49 & \cellcolor[rgb]{ .98,  .647,  .659}1,73 & \cellcolor[rgb]{ .98,  .608,  .616}1,87 \\
          & ANOM(0.1) &       & \cellcolor[rgb]{ .98,  .675,  .686}1,32 & \cellcolor[rgb]{ .98,  .573,  .58}1,61 & \cellcolor[rgb]{ .98,  .647,  .659}1,40 & \cellcolor[rgb]{ .98,  .588,  .6}1,56 &       & \cellcolor[rgb]{ .988,  .949,  .961}0,67 & \cellcolor[rgb]{ .984,  .816,  .824}1,14 & \cellcolor[rgb]{ .98,  .678,  .69}1,61 & \cellcolor[rgb]{ .976,  .522,  .529}2,17 & \cellcolor[rgb]{ .973,  .412,  .42}2,55 \\
          & ANOM(1) &       & \cellcolor[rgb]{ .984,  .8,  .812}0,97 & \cellcolor[rgb]{ .984,  .784,  .796}1,01 & \cellcolor[rgb]{ .988,  .91,  .922}0,66 & \cellcolor[rgb]{ .988,  .91,  .922}0,66 &       & \cellcolor[rgb]{ .988,  .984,  .996}0,54 & \cellcolor[rgb]{ .988,  .945,  .953}0,68 & \cellcolor[rgb]{ .988,  .851,  .859}1,01 & \cellcolor[rgb]{ .98,  .686,  .694}1,59 & \cellcolor[rgb]{ .98,  .565,  .573}2,02 \\
          & ANOM(10) &       & \cellcolor[rgb]{ .984,  .8,  .812}0,97 & \cellcolor[rgb]{ .984,  .808,  .82}0,94 & \cellcolor[rgb]{ .988,  .984,  .996}0,46 & \cellcolor[rgb]{ .988,  .988,  1}0,44 &       & \cellcolor[rgb]{ .988,  .988,  1}0,52 & \cellcolor[rgb]{ .988,  .957,  .969}0,64 & \cellcolor[rgb]{ .988,  .871,  .882}0,94 & \cellcolor[rgb]{ .984,  .753,  .761}1,36 & \cellcolor[rgb]{ .98,  .671,  .678}1,65 \\
          &       &       &       &       &       &       &       &       &       &       &       &  \\
    \multirow{6}[0]{*}{\begin{turn}{90}$RMSE(y_t,\hat{y}_t)$\end{turn}} & CPVS(1) &       & \cellcolor[rgb]{ .988,  .89,  .902}3,09 & \cellcolor[rgb]{ .98,  .686,  .694}3,31 & \cellcolor[rgb]{ .976,  .537,  .545}3,47 & \cellcolor[rgb]{ .976,  .435,  .443}3,58 &       & \cellcolor[rgb]{ .988,  .882,  .894}3,12 & \cellcolor[rgb]{ .984,  .8,  .808}3,25 & \cellcolor[rgb]{ .984,  .757,  .765}3,31 & \cellcolor[rgb]{ .984,  .714,  .725}3,37 & \cellcolor[rgb]{ .98,  .675,  .682}3,43 \\
          & CPVS(2) &       & \cellcolor[rgb]{ .988,  .906,  .918}3,07 & \cellcolor[rgb]{ .984,  .71,  .718}3,29 & \cellcolor[rgb]{ .976,  .529,  .541}3,48 & \cellcolor[rgb]{ .973,  .412,  .42}3,61 &       & \cellcolor[rgb]{ .988,  .894,  .906}3,10 & \cellcolor[rgb]{ .984,  .824,  .831}3,21 & \cellcolor[rgb]{ .984,  .773,  .78}3,29 & \cellcolor[rgb]{ .984,  .733,  .741}3,34 & \cellcolor[rgb]{ .98,  .667,  .675}3,44 \\
          & CPVS  &       & \cellcolor[rgb]{ .988,  .933,  .945}3,04 & \cellcolor[rgb]{ .984,  .741,  .749}3,25 & \cellcolor[rgb]{ .98,  .569,  .576}3,44 & \cellcolor[rgb]{ .976,  .451,  .459}3,57 &       & \cellcolor[rgb]{ .988,  .894,  .906}3,10 & \cellcolor[rgb]{ .988,  .851,  .863}3,17 & \cellcolor[rgb]{ .984,  .796,  .804}3,25 & \cellcolor[rgb]{ .984,  .741,  .749}3,33 & \cellcolor[rgb]{ .984,  .702,  .714}3,39 \\
          & ANOM(0.1) &       & \cellcolor[rgb]{ .988,  .867,  .878}3,11 & \cellcolor[rgb]{ .98,  .69,  .702}3,30 & \cellcolor[rgb]{ .984,  .733,  .745}3,26 & \cellcolor[rgb]{ .98,  .663,  .671}3,34 &       & \cellcolor[rgb]{ .988,  .973,  .98}2,99 & \cellcolor[rgb]{ .988,  .89,  .902}3,11 & \cellcolor[rgb]{ .984,  .761,  .769}3,30 & \cellcolor[rgb]{ .98,  .569,  .576}3,59 & \cellcolor[rgb]{ .973,  .412,  .42}3,82 \\
          & ANOM(1) &       & \cellcolor[rgb]{ .988,  .988,  1}2,98 & \cellcolor[rgb]{ .988,  .929,  .941}3,04 & \cellcolor[rgb]{ .988,  .945,  .957}3,02 & \cellcolor[rgb]{ .988,  .937,  .949}3,03 &       & \cellcolor[rgb]{ .988,  .988,  1}2,96 & \cellcolor[rgb]{ .988,  .98,  .992}2,97 & \cellcolor[rgb]{ .988,  .937,  .949}3,04 & \cellcolor[rgb]{ .984,  .792,  .804}3,25 & \cellcolor[rgb]{ .98,  .651,  .663}3,46 \\
          & ANOM(10) &       & \cellcolor[rgb]{ .988,  .988,  1}2,98 & \cellcolor[rgb]{ .988,  .957,  .969}3,01 & \cellcolor[rgb]{ .988,  .98,  .992}2,99 & \cellcolor[rgb]{ .988,  .973,  .984}3,00 &       & \cellcolor[rgb]{ .988,  .988,  1}2,96 & \cellcolor[rgb]{ .988,  .988,  1}2,96 & \cellcolor[rgb]{ .988,  .953,  .965}3,01 & \cellcolor[rgb]{ .988,  .867,  .878}3,14 & \cellcolor[rgb]{ .984,  .792,  .804}3,25 \\
    \bottomrule
    \end{tabular}%
  }
  \vspace{0.5em}
\caption{{{Scenario 2.b: signal estimation and data fitting}. CPVS is our proposal, ANOM by \cite{hadj2020bayesian}. Two metrics: root mean squared error between true signal and estimated one $RMSE(\mu_t,\hat{\mu}_t)$ (the lower, the better), root mean squared error between observed data and fitted values $RMSE(y_t,\hat{y}_t)$ (the lower, the better). The table shows the results in both varying $T$ and fixed $T=1000$ scenarios where the true signal is forced to be continuous at the change points. The red scale depicts worse performances.}}\label{tab:resultsMSE1_cont}%
\end{table}%

\textbf{Scenario 3 (multivariate, dense regime)} We compare the performance of the joint change point detection with the one achieved with a sequence of $d$ univariate estimations. In the latter case, the joint metrics are the average over the $d$ individual ones. An important comment is that the estimation of the model for a given change point candidate can be done in parallel, therefore not particularly affecting the computational efficiency. {We also investigate the performance of \algonamesp in presence of an increasing number of noisy series in the panel.}\\

% Table generated by Excel2LaTeX from sheet 'Foglio1'
\begin{table}[htbp]
  \centering
  \resizebox{1\textwidth}{!}{
    \begin{tabular}{clrrrrrrrrrrrrrr}
    \toprule
          &       &       & \multicolumn{6}{c}{$(T=1000,\,m=4,\,d)$}                       &       & \multicolumn{6}{c}{$(T=1000,\,m=6,\,d)$} \\
          &       &       & 2     & 3     & 4     & 5     & 10    & 20    &       & 2     & 3     & 4     & 5     & 10    & 20 \\
              \cmidrule{2-2}\cmidrule{4-9}\cmidrule{11-16}
    \multirow{8}[0]{*}{\begin{turn}{90}$C(\mathcal{P},\hat{\mathcal{P}})$\end{turn}} & CPVS(1) &       & \cellcolor[rgb]{ .988,  .988,  1}0,83 & \cellcolor[rgb]{ .98,  .988,  .992}0,83 & \cellcolor[rgb]{ .961,  .98,  .976}0,84 & \cellcolor[rgb]{ .961,  .98,  .976}0,84 & \cellcolor[rgb]{ .957,  .976,  .973}0,84 & \cellcolor[rgb]{ .957,  .976,  .973}0,84 &       & \cellcolor[rgb]{ .918,  .961,  .937}0,78 & \cellcolor[rgb]{ .988,  .988,  1}0,75 & \cellcolor[rgb]{ .929,  .965,  .949}0,78 & \cellcolor[rgb]{ .969,  .98,  .98}0,76 & \cellcolor[rgb]{ .965,  .98,  .976}0,76 & \cellcolor[rgb]{ .973,  .984,  .988}0,76 \\
          & CPVS(2) &       & \cellcolor[rgb]{ .784,  .906,  .824}0,88 & \cellcolor[rgb]{ .859,  .937,  .89}0,87 & \cellcolor[rgb]{ .82,  .922,  .855}0,88 & \cellcolor[rgb]{ .843,  .929,  .875}0,87 & \cellcolor[rgb]{ .82,  .922,  .855}0,88 & \cellcolor[rgb]{ .835,  .929,  .871}0,87 &       & \cellcolor[rgb]{ .824,  .922,  .859}0,81 & \cellcolor[rgb]{ .847,  .933,  .878}0,81 & \cellcolor[rgb]{ .824,  .922,  .855}0,81 & \cellcolor[rgb]{ .875,  .945,  .902}0,79 & \cellcolor[rgb]{ .847,  .933,  .878}0,80 & \cellcolor[rgb]{ .855,  .937,  .886}0,80 \\
          & CPVS  &       & \cellcolor[rgb]{ .824,  .922,  .859}0,87 & \cellcolor[rgb]{ .882,  .945,  .91}0,86 & \cellcolor[rgb]{ .831,  .925,  .867}0,87 & \cellcolor[rgb]{ .839,  .929,  .875}0,87 & \cellcolor[rgb]{ .808,  .918,  .843}0,88 & \cellcolor[rgb]{ .839,  .929,  .871}0,87 &       & \cellcolor[rgb]{ .871,  .941,  .898}0,80 & \cellcolor[rgb]{ .898,  .953,  .922}0,79 & \cellcolor[rgb]{ .863,  .937,  .89}0,80 & \cellcolor[rgb]{ .91,  .957,  .933}0,78 & \cellcolor[rgb]{ .882,  .945,  .91}0,79 & \cellcolor[rgb]{ .886,  .949,  .914}0,79 \\
          & MCPVS(1) &       & \cellcolor[rgb]{ .706,  .875,  .757}0,91 & \cellcolor[rgb]{ .549,  .812,  .62}0,95 & \cellcolor[rgb]{ .478,  .784,  .561}0,96 & \cellcolor[rgb]{ .475,  .78,  .557}0,97 & \cellcolor[rgb]{ .431,  .765,  .518}0,98 & \cellcolor[rgb]{ .408,  .753,  .498}0,98 &       & \cellcolor[rgb]{ .651,  .851,  .71}0,88 & \cellcolor[rgb]{ .588,  .827,  .655}0,90 & \cellcolor[rgb]{ .522,  .8,  .6}0,92 & \cellcolor[rgb]{ .529,  .804,  .608}0,92 & \cellcolor[rgb]{ .447,  .769,  .533}0,95 & \cellcolor[rgb]{ .412,  .757,  .502}0,96 \\
          & MCPVS(2) &       & \cellcolor[rgb]{ .643,  .851,  .702}0,92 & \cellcolor[rgb]{ .576,  .824,  .643}0,94 & \cellcolor[rgb]{ .525,  .804,  .604}0,95 & \cellcolor[rgb]{ .494,  .788,  .573}0,96 & \cellcolor[rgb]{ .475,  .78,  .557}0,97 & \cellcolor[rgb]{ .427,  .761,  .514}0,98 &       & \cellcolor[rgb]{ .655,  .855,  .714}0,87 & \cellcolor[rgb]{ .565,  .82,  .635}0,91 & \cellcolor[rgb]{ .522,  .8,  .6}0,92 & \cellcolor[rgb]{ .506,  .792,  .584}0,93 & \cellcolor[rgb]{ .435,  .765,  .525}0,95 & \cellcolor[rgb]{ .424,  .761,  .514}0,96 \\
          & MCPVS &       & \cellcolor[rgb]{ .608,  .835,  .675}0,93 & \cellcolor[rgb]{ .514,  .796,  .592}0,96 & \cellcolor[rgb]{ .447,  .769,  .533}0,97 & \cellcolor[rgb]{ .431,  .765,  .522}0,98 & \cellcolor[rgb]{ .42,  .761,  .51}0,98 & \cellcolor[rgb]{ .388,  .745,  .482}0,99 &       & \cellcolor[rgb]{ .631,  .843,  .694}0,88 & \cellcolor[rgb]{ .569,  .82,  .639}0,90 & \cellcolor[rgb]{ .494,  .788,  .576}0,93 & \cellcolor[rgb]{ .506,  .796,  .584}0,93 & \cellcolor[rgb]{ .412,  .757,  .502}0,96 & \cellcolor[rgb]{ .388,  .745,  .482}0,97 \\
          & ADA-X(7) &       & 0,30  & 0,27  & 0,30  & 0,30  & 0,31  & 0,40  &       & 0,17  & 0,17  & 0,18  & 0,18  & 0,20  & 0,27 \\
          & ADA-X(15) &       & 0,35  & 0,31  & 0,37  & 0,33  & 0,36  & 0,44  &       & 0,18  & 0,18  & 0,19  & 0,19  & 0,22  & 0,25 \\
          &       &       &       &       &       &       &       &       &       &       &       &       &       &       &  \\
    \multirow{8}[0]{*}{\begin{turn}{90}$H(\mathcal{P},\hat{\mathcal{P}})$\end{turn}} & CPVS(1) &       & \cellcolor[rgb]{ .929,  .961,  .949}0,10 & \cellcolor[rgb]{ .976,  .98,  .988}0,10 & \cellcolor[rgb]{ .984,  .984,  .996}0,11 & \cellcolor[rgb]{ .937,  .965,  .957}0,10 & \cellcolor[rgb]{ .969,  .98,  .984}0,10 & \cellcolor[rgb]{ .988,  .988,  1}0,11 &       & \cellcolor[rgb]{ .851,  .933,  .882}0,11 & \cellcolor[rgb]{ .988,  .988,  1}0,14 & \cellcolor[rgb]{ .878,  .941,  .906}0,12 & \cellcolor[rgb]{ .906,  .953,  .929}0,12 & \cellcolor[rgb]{ .875,  .941,  .902}0,12 & \cellcolor[rgb]{ .902,  .953,  .925}0,12 \\
          & CPVS(2) &       & \cellcolor[rgb]{ .678,  .863,  .733}0,05 & \cellcolor[rgb]{ .78,  .902,  .82}0,07 & \cellcolor[rgb]{ .733,  .882,  .78}0,06 & \cellcolor[rgb]{ .725,  .882,  .776}0,06 & \cellcolor[rgb]{ .71,  .875,  .761}0,06 & \cellcolor[rgb]{ .741,  .886,  .788}0,06 &       & \cellcolor[rgb]{ .702,  .871,  .753}0,08 & \cellcolor[rgb]{ .757,  .894,  .8}0,09 & \cellcolor[rgb]{ .69,  .867,  .741}0,08 & \cellcolor[rgb]{ .737,  .886,  .784}0,09 & \cellcolor[rgb]{ .698,  .871,  .749}0,08 & \cellcolor[rgb]{ .706,  .871,  .757}0,08 \\
          & CPVS  &       & \cellcolor[rgb]{ .757,  .894,  .8}0,07 & \cellcolor[rgb]{ .863,  .937,  .89}0,09 & \cellcolor[rgb]{ .839,  .925,  .871}0,08 & \cellcolor[rgb]{ .8,  .91,  .835}0,07 & \cellcolor[rgb]{ .773,  .902,  .816}0,07 & \cellcolor[rgb]{ .839,  .925,  .871}0,08 &       & \cellcolor[rgb]{ .82,  .918,  .855}0,11 & \cellcolor[rgb]{ .863,  .937,  .89}0,11 & \cellcolor[rgb]{ .804,  .914,  .843}0,10 & \cellcolor[rgb]{ .839,  .925,  .871}0,11 & \cellcolor[rgb]{ .8,  .91,  .835}0,10 & \cellcolor[rgb]{ .808,  .914,  .843}0,10 \\
          & MCPVS(1) &       & \cellcolor[rgb]{ .608,  .831,  .671}0,04 & \cellcolor[rgb]{ .518,  .796,  .592}0,03 & \cellcolor[rgb]{ .455,  .773,  .537}0,02 & \cellcolor[rgb]{ .451,  .769,  .537}0,01 & \cellcolor[rgb]{ .412,  .753,  .502}0,01 & \cellcolor[rgb]{ .4,  .749,  .49}0,01 &       & \cellcolor[rgb]{ .592,  .827,  .659}0,05 & \cellcolor[rgb]{ .533,  .804,  .604}0,04 & \cellcolor[rgb]{ .494,  .788,  .573}0,03 & \cellcolor[rgb]{ .494,  .788,  .573}0,03 & \cellcolor[rgb]{ .412,  .753,  .502}0,01 & \cellcolor[rgb]{ .4,  .749,  .494}0,01 \\
          & MCPVS(2) &       & \cellcolor[rgb]{ .541,  .808,  .616}0,03 & \cellcolor[rgb]{ .51,  .792,  .584}0,02 & \cellcolor[rgb]{ .471,  .776,  .553}0,02 & \cellcolor[rgb]{ .439,  .765,  .525}0,01 & \cellcolor[rgb]{ .439,  .765,  .525}0,01 & \cellcolor[rgb]{ .416,  .757,  .506}0,01 &       & \cellcolor[rgb]{ .569,  .816,  .639}0,05 & \cellcolor[rgb]{ .49,  .784,  .573}0,03 & \cellcolor[rgb]{ .463,  .773,  .545}0,03 & \cellcolor[rgb]{ .459,  .773,  .541}0,03 & \cellcolor[rgb]{ .408,  .753,  .498}0,01 & \cellcolor[rgb]{ .404,  .749,  .498}0,01 \\
          & MCPVS &       & \cellcolor[rgb]{ .537,  .804,  .612}0,03 & \cellcolor[rgb]{ .486,  .784,  .569}0,02 & \cellcolor[rgb]{ .427,  .761,  .518}0,01 & \cellcolor[rgb]{ .408,  .753,  .498}0,01 & \cellcolor[rgb]{ .408,  .753,  .498}0,01 & \cellcolor[rgb]{ .388,  .745,  .482}0,00 &       & \cellcolor[rgb]{ .604,  .831,  .671}0,06 & \cellcolor[rgb]{ .506,  .792,  .584}0,04 & \cellcolor[rgb]{ .467,  .776,  .553}0,03 & \cellcolor[rgb]{ .478,  .78,  .561}0,03 & \cellcolor[rgb]{ .396,  .749,  .49}0,01 & \cellcolor[rgb]{ .388,  .745,  .482}0,01 \\
          & ADA-X(7) &       & 0,86  & 0,91  & 0,86  & 0,86  & 0,79  & 0,63  &       & 0,96  & 0,96  & 0,91  & 0,92  & 0,87  & 0,69 \\
          & ADA-X(15) &       & 0,78  & 0,83  & 0,74  & 0,80  & 0,70  & 0,57  &       & 0,94  & 0,93  & 0,91  & 0,88  & 0,83  & 0,75 \\
          &       &       &       &       &       &       &       &       &       &       &       &       &       &       &  \\
    \multirow{8}[0]{*}{\begin{turn}{90}$B(m,\widehat{m})$\end{turn}} & CPVS(1) &       & \cellcolor[rgb]{ .988,  .988,  1}0,53 & \cellcolor[rgb]{ .945,  .969,  .961}0,49 & \cellcolor[rgb]{ .965,  .976,  .98}0,51 & \cellcolor[rgb]{ .839,  .925,  .871}0,40 & \cellcolor[rgb]{ .875,  .941,  .902}0,43 & \cellcolor[rgb]{ .894,  .949,  .918}0,45 &       & \cellcolor[rgb]{ .898,  .949,  .922}0,94 & \cellcolor[rgb]{ .988,  .988,  1}1,10 & \cellcolor[rgb]{ .914,  .957,  .933}0,97 & \cellcolor[rgb]{ .957,  .976,  .973}1,05 & \cellcolor[rgb]{ .902,  .953,  .925}0,94 & \cellcolor[rgb]{ .937,  .965,  .957}1,01 \\
          & CPVS(2) &       & \cellcolor[rgb]{ .702,  .871,  .753}0,28 & \cellcolor[rgb]{ .8,  .91,  .839}0,37 & \cellcolor[rgb]{ .663,  .855,  .722}0,25 & \cellcolor[rgb]{ .69,  .867,  .741}0,27 & \cellcolor[rgb]{ .6,  .831,  .667}0,19 & \cellcolor[rgb]{ .573,  .82,  .639}0,16 &       & \cellcolor[rgb]{ .686,  .867,  .741}0,55 & \cellcolor[rgb]{ .737,  .886,  .784}0,65 & \cellcolor[rgb]{ .627,  .839,  .686}0,44 & \cellcolor[rgb]{ .753,  .89,  .796}0,67 & \cellcolor[rgb]{ .651,  .851,  .71}0,49 & \cellcolor[rgb]{ .667,  .855,  .722}0,51 \\
          & CPVS  &       & \cellcolor[rgb]{ .827,  .922,  .863}0,39 & \cellcolor[rgb]{ .929,  .965,  .949}0,48 & \cellcolor[rgb]{ .843,  .929,  .875}0,41 & \cellcolor[rgb]{ .788,  .906,  .827}0,36 & \cellcolor[rgb]{ .741,  .886,  .784}0,31 & \cellcolor[rgb]{ .788,  .906,  .827}0,36 &       & \cellcolor[rgb]{ .91,  .957,  .933}0,96 & \cellcolor[rgb]{ .929,  .961,  .949}0,99 & \cellcolor[rgb]{ .867,  .937,  .894}0,88 & \cellcolor[rgb]{ .937,  .965,  .953}1,01 & \cellcolor[rgb]{ .875,  .941,  .902}0,89 & \cellcolor[rgb]{ .882,  .945,  .91}0,91 \\
          & MCPVS(1) &       & \cellcolor[rgb]{ .635,  .843,  .694}0,22 & \cellcolor[rgb]{ .498,  .788,  .576}0,10 & \cellcolor[rgb]{ .455,  .773,  .537}0,06 & \cellcolor[rgb]{ .431,  .761,  .518}0,04 & \cellcolor[rgb]{ .388,  .745,  .482}0,00 & \cellcolor[rgb]{ .388,  .745,  .482}0,00 &       & \cellcolor[rgb]{ .604,  .831,  .671}0,40 & \cellcolor[rgb]{ .537,  .804,  .612}0,28 & \cellcolor[rgb]{ .486,  .784,  .565}0,18 & \cellcolor[rgb]{ .529,  .8,  .604}0,26 & \cellcolor[rgb]{ .408,  .753,  .498}0,04 & \cellcolor[rgb]{ .396,  .749,  .49}0,02 \\
          & MCPVS(2) &       & \cellcolor[rgb]{ .545,  .808,  .616}0,14 & \cellcolor[rgb]{ .569,  .816,  .635}0,16 & \cellcolor[rgb]{ .455,  .773,  .537}0,06 & \cellcolor[rgb]{ .388,  .745,  .482}0,00 & \cellcolor[rgb]{ .431,  .761,  .518}0,04 & \cellcolor[rgb]{ .431,  .761,  .518}0,04 &       & \cellcolor[rgb]{ .561,  .816,  .631}0,32 & \cellcolor[rgb]{ .506,  .792,  .584}0,22 & \cellcolor[rgb]{ .451,  .769,  .537}0,12 & \cellcolor[rgb]{ .475,  .78,  .557}0,16 & \cellcolor[rgb]{ .396,  .749,  .49}0,02 & \cellcolor[rgb]{ .388,  .745,  .482}0,00 \\
          & MCPVS &       & \cellcolor[rgb]{ .498,  .788,  .576}0,10 & \cellcolor[rgb]{ .478,  .78,  .557}0,08 & \cellcolor[rgb]{ .408,  .753,  .498}0,02 & \cellcolor[rgb]{ .388,  .745,  .482}0,00 & \cellcolor[rgb]{ .388,  .745,  .482}0,00 & \cellcolor[rgb]{ .388,  .745,  .482}0,00 &       & \cellcolor[rgb]{ .616,  .835,  .678}0,42 & \cellcolor[rgb]{ .529,  .8,  .604}0,26 & \cellcolor[rgb]{ .463,  .773,  .545}0,14 & \cellcolor[rgb]{ .494,  .788,  .576}0,20 & \cellcolor[rgb]{ .396,  .749,  .49}0,02 & \cellcolor[rgb]{ .396,  .749,  .49}0,02 \\
          & ADA-X(7) &       & 3,68  & 3,78  & 3,70  & 3,66  & 3,58  & 3,04  &       & 5,88  & 5,88  & 5,76  & 5,82  & 5,62  & 5,02 \\
          & ADA-X(15) &       & 3,44  & 3,62  & 3,44  & 3,52  & 3,34  & 2,96  &       & 5,84  & 5,78  & 5,70  & 5,70  & 5,52  & 5,26 \\
          \bottomrule
    \end{tabular}%
  }
    \vspace{0.5em}
  \caption{{Scenario 3: segmentation performance}. MCPVS is our multivariate proposal, ADA-X by \cite{bertolacci2022adaptspec}. Three metrics: segmentation coverage $C(\mathcal{P},\hat{\mathcal{P}})\in (0,1)$ (the higher, the better), Hausdorff distance $H(\mathcal{P},\hat{\mathcal{P}})\in (0,1)$ (the lower, the better), and bias $B(m,\hat{m})\in [0,m]$ (the lower, the better). The table shows the results when increasing the number of series $d$ modeled jointly. The green scale identifies better performances.}  \label{tab:results2}%
\end{table}%

\begin{table}[htbp]
   \centering
    \resizebox{0.9\textwidth}{!}{
    \begin{tabular}{lrrrrrrrrrrr}
    \toprule
          & \multicolumn{3}{c}{$C(\mathcal{P},\hat{\mathcal{P}})$} &       & \multicolumn{3}{c}{$H(\mathcal{P},\hat{\mathcal{P}})$} &       & \multicolumn{3}{c}{$B(m,\widehat{m})$} \\
          & \multicolumn{1}{l}{$d_1=9$} & \multicolumn{1}{l}{$d_1=7$} & \multicolumn{1}{l}{$d_1=4$} &       & \multicolumn{1}{l}{$d_1=9$} & \multicolumn{1}{l}{$d_1=7$} & \multicolumn{1}{l}{$d_1=4$} &       & \multicolumn{1}{l}{$d_1=9$} & \multicolumn{1}{l}{$d_1=7$} & \multicolumn{1}{l}{$d_1=4$} \\
    \cmidrule{2-4}\cmidrule{6-8}\cmidrule{10-12}
    MCPVS(1) & \cellcolor[rgb]{ .388,  .745,  .482}0,93 & \cellcolor[rgb]{ .404,  .753,  .498}0,92 & \cellcolor[rgb]{ .478,  .78,  .561}0,83 &       & \cellcolor[rgb]{ .388,  .745,  .482}0,02 & \cellcolor[rgb]{ .392,  .745,  .486}0,03 & \cellcolor[rgb]{ .443,  .765,  .529}0,11 &       & \cellcolor[rgb]{ .388,  .745,  .482}0,10 & \cellcolor[rgb]{ .4,  .749,  .49}0,18 & \cellcolor[rgb]{ .498,  .788,  .576}0,82 \\
    MCPVS(2) & \cellcolor[rgb]{ .443,  .769,  .529}0,87 & \cellcolor[rgb]{ .451,  .773,  .537}0,86 & \cellcolor[rgb]{ .49,  .788,  .573}0,82 &       & \cellcolor[rgb]{ .416,  .757,  .506}0,07 & \cellcolor[rgb]{ .416,  .757,  .506}0,07 & \cellcolor[rgb]{ .443,  .765,  .529}0,11 &       & \cellcolor[rgb]{ .443,  .765,  .529}0,46 & \cellcolor[rgb]{ .447,  .769,  .533}0,50 & \cellcolor[rgb]{ .498,  .788,  .576}0,82 \\
    MCPVS & \cellcolor[rgb]{ .4,  .753,  .494}0,92 & \cellcolor[rgb]{ .424,  .761,  .514}0,89 & \cellcolor[rgb]{ .482,  .784,  .565}0,83 &       & \cellcolor[rgb]{ .396,  .745,  .49}0,04 & \cellcolor[rgb]{ .404,  .749,  .498}0,05 & \cellcolor[rgb]{ .451,  .769,  .537}0,13 &       & \cellcolor[rgb]{ .4,  .749,  .49}0,18 & \cellcolor[rgb]{ .424,  .757,  .514}0,34 & \cellcolor[rgb]{ .518,  .796,  .596}0,96 \\
    ADA-X(7) & \cellcolor[rgb]{ .933,  .969,  .953}0,30 & \cellcolor[rgb]{ .984,  .988,  .996}0,24 & \cellcolor[rgb]{ .988,  .988,  1}0,24 &       & \cellcolor[rgb]{ .882,  .945,  .91}0,82 & \cellcolor[rgb]{ .969,  .98,  .98}0,96 & \cellcolor[rgb]{ .988,  .988,  1}0,99 &       & \cellcolor[rgb]{ .922,  .961,  .941}3,56 & \cellcolor[rgb]{ .973,  .98,  .988}3,90 & \cellcolor[rgb]{ .988,  .988,  1}3,98 \\
    ADA-X(15) & \cellcolor[rgb]{ .863,  .937,  .89}0,39 & \cellcolor[rgb]{ .949,  .973,  .969}0,28 & \cellcolor[rgb]{ .988,  .988,  1}0,24 &       & \cellcolor[rgb]{ .792,  .91,  .831}0,68 & \cellcolor[rgb]{ .918,  .957,  .937}0,88 & \cellcolor[rgb]{ .984,  .984,  .996}0,99 &       & \cellcolor[rgb]{ .851,  .929,  .882}3,10 & \cellcolor[rgb]{ .941,  .969,  .961}3,70 & \cellcolor[rgb]{ .988,  .988,  1}3,98 \\
     \bottomrule
    \end{tabular}%
  }
  \vspace{0.5em}
\caption{{Scenario 3: segmentation performance with noisy series}. MCPVS is our multivariate proposal, ADA-X by \cite{bertolacci2022adaptspec}. Three metrics: segmentation coverage $C(\mathcal{P},\hat{\mathcal{P}})\in (0,1)$ (the higher, the better), Hausdorff distance $H(\mathcal{P},\hat{\mathcal{P}})\in (0,1)$ (the lower, the better), and bias $B(m,\hat{m})\in [0,m]$ (the lower, the better). The table shows the results when increasing the number of noisy series in the panel of series modeled jointly. The green scale identifies better performances.}\label{tab:resultsMultiNoise}%
\end{table}%

\textbf{Scenario 4 (piece-wise autoregressive process)}.\cite{davis2006structural}.
In this setting both \algonamesp and ANOM fail to detect change-points when working with original series. However, we observed that the performance is improved when taking the series of first differences $y_t-y_{t-1}$. The results in Table \ref{tab:results_AR} consider the latter case.\\

\textbf{Scenario 5 (slowly varying autoregressive process)}.\cite{davis2006structural}.Figure \ref{fig:spectra_slow_auto} panel(a) depicts the true continuous spectrum implied by the model specification $f$, panel (b) depicts the estimated piecewise constant functions with different methodologies $\hat{f}$. Table \ref{tab:results_AR} shows the results.\\

% Table generated by Excel2LaTeX from sheet 'COV'
\begin{table}[htbp]
  \centering
  \resizebox{0.8\textwidth}{!}{
    \begin{tabular}{lrrrrrrrrrr}
     \toprule
          &       & \multicolumn{5}{c}{Piece-wise AR}             &       & \multicolumn{3}{c}{Slowly varying AR} \\
          &       & \multicolumn{1}{l}{$C(\mathcal{P},\hat{\mathcal{P}})$} &       & \multicolumn{1}{l}{$H(\mathcal{P},\hat{\mathcal{P}})$} &       & \multicolumn{1}{l}{$B(m,\widehat{m})$} &       &       & \multicolumn{1}{l}{$\Vert f-\hat{f}\Vert_2^2$} &  \\
          \cmidrule{1-1}\cmidrule{3-3}\cmidrule{5-5}\cmidrule{7-7}\cmidrule{10-10}
    CPVS(1) &       & \cellcolor[rgb]{ .98,  .776,  .784}0,89 &       & \cellcolor[rgb]{ .988,  .941,  .953}0,05 &       & \cellcolor[rgb]{ .988,  .929,  .941}0,24 &       &   \hspace{1cm}    & \cellcolor[rgb]{ .988,  .973,  .984}0,17 & \hspace{1cm}  \\
    CPVS(2) &       & \cellcolor[rgb]{ .976,  .631,  .643}0,83 &       & \cellcolor[rgb]{ .988,  .902,  .914}0,09 &       & \cellcolor[rgb]{ .988,  .871,  .882}0,32 &       &       & \cellcolor[rgb]{ .988,  .988,  1}0,15 &  \\
    CPVS  &       & \cellcolor[rgb]{ .98,  .776,  .788}0,89 &       & \cellcolor[rgb]{ .988,  .937,  .949}0,06 &       & \cellcolor[rgb]{ .988,  .988,  1}0,16 &       &       & \cellcolor[rgb]{ .988,  .976,  .988}0,16 &  \\
    ANOM(0.1) &       & \cellcolor[rgb]{ .973,  .412,  .42}0,74 &       & \cellcolor[rgb]{ .973,  .412,  .42}0,49 &       & \cellcolor[rgb]{ .973,  .412,  .42}0,94 &       &       & \cellcolor[rgb]{ .976,  .435,  .443}0,67 &  \\
    ANOM(1) &       & \cellcolor[rgb]{ .973,  .42,  .427}0,74 &       & \cellcolor[rgb]{ .976,  .427,  .435}0,48 &       & \cellcolor[rgb]{ .976,  .443,  .451}0,90 &       &       & \cellcolor[rgb]{ .973,  .412,  .42}0,69 &  \\
    ANOM(10) &       & \cellcolor[rgb]{ .973,  .478,  .486}0,77 &       & \cellcolor[rgb]{ .976,  .463,  .471}0,45 &       & \cellcolor[rgb]{ .976,  .502,  .51}0,82 &       &       & \cellcolor[rgb]{ .976,  .533,  .541}0,58 &  \\
    AP    &       & \cellcolor[rgb]{ .984,  .945,  .957}0,96 &       & \cellcolor[rgb]{ .988,  .984,  .996}0,02 &       & \cellcolor[rgb]{ .988,  .961,  .973}0,20 &       &       & \cellcolor[rgb]{ .988,  .973,  .984}0,16 &  \\
    ADA(7) &       & \cellcolor[rgb]{ .984,  .984,  .996}0,97 &       & \cellcolor[rgb]{ .988,  .988,  1}0,01 &       & \cellcolor[rgb]{ .988,  .988,  1}0,16 &       &       & \cellcolor[rgb]{ .984,  .753,  .765}0,37 &  \\
    ADA(15) &       & \cellcolor[rgb]{ .988,  .988,  1}0,98 &       & \cellcolor[rgb]{ .988,  .988,  1}0,02 &       & \cellcolor[rgb]{ .988,  .961,  .973}0,20 &       &       & \cellcolor[rgb]{ .984,  .71,  .722}0,41 &  \\
    {LSW}   &       & \cellcolor[rgb]{ .976,  .631,  .643}0,83 &       & \cellcolor[rgb]{ .984,  .835,  .847}0,14 &       & \cellcolor[rgb]{ .988,  .871,  .882}0,32 &       &     & $-$   \\
    \bottomrule
    \end{tabular}%
  }
    \vspace{0.5em}
  \caption{{Scenarios 4 and 5: segmentation performance}. CPVS (our proposal), ANOM \citep{hadj2020bayesian}, AP \citep{davis2006structural}, ADA \citep{rosen2012adaptspec}, {and LSW \citep{korkas2017multiple}} are considered. Three metrics for Scenario 5 (piece-wise AR): segmentation coverage $C(\mathcal{P},\hat{\mathcal{P}})\in (0,1)$ (the higher, the better), Hausdorff distance $H(\mathcal{P},\hat{\mathcal{P}})\in (0,1)$ (the lower, the better), and bias $B(m,\hat{m})\in [0,m]$ (the lower, the better). For Scenario 6 (slowly varying AR) the mean squared error between the true log-spectrum and the estimated one $\Vert f-\hat{f}\Vert_2^2$ (the lower, the better) is computed. The red scale identifies worse performances.}\label{tab:results_AR}%
\end{table}%

\begin{figure}[!ht]
\centering
\subfigure[]{\includegraphics[width=.47\textwidth]{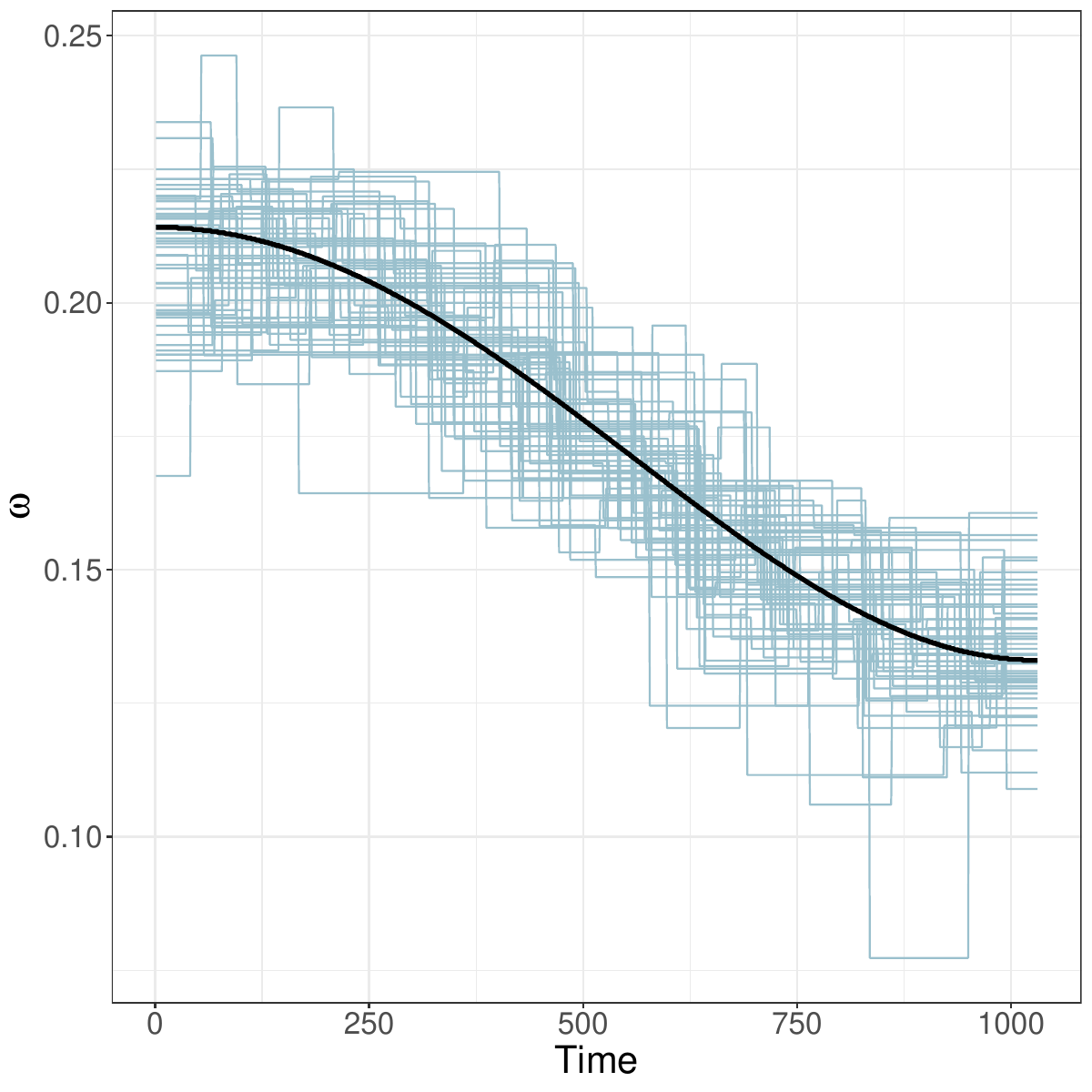}}
\subfigure[]{\includegraphics[width=.47\textwidth]{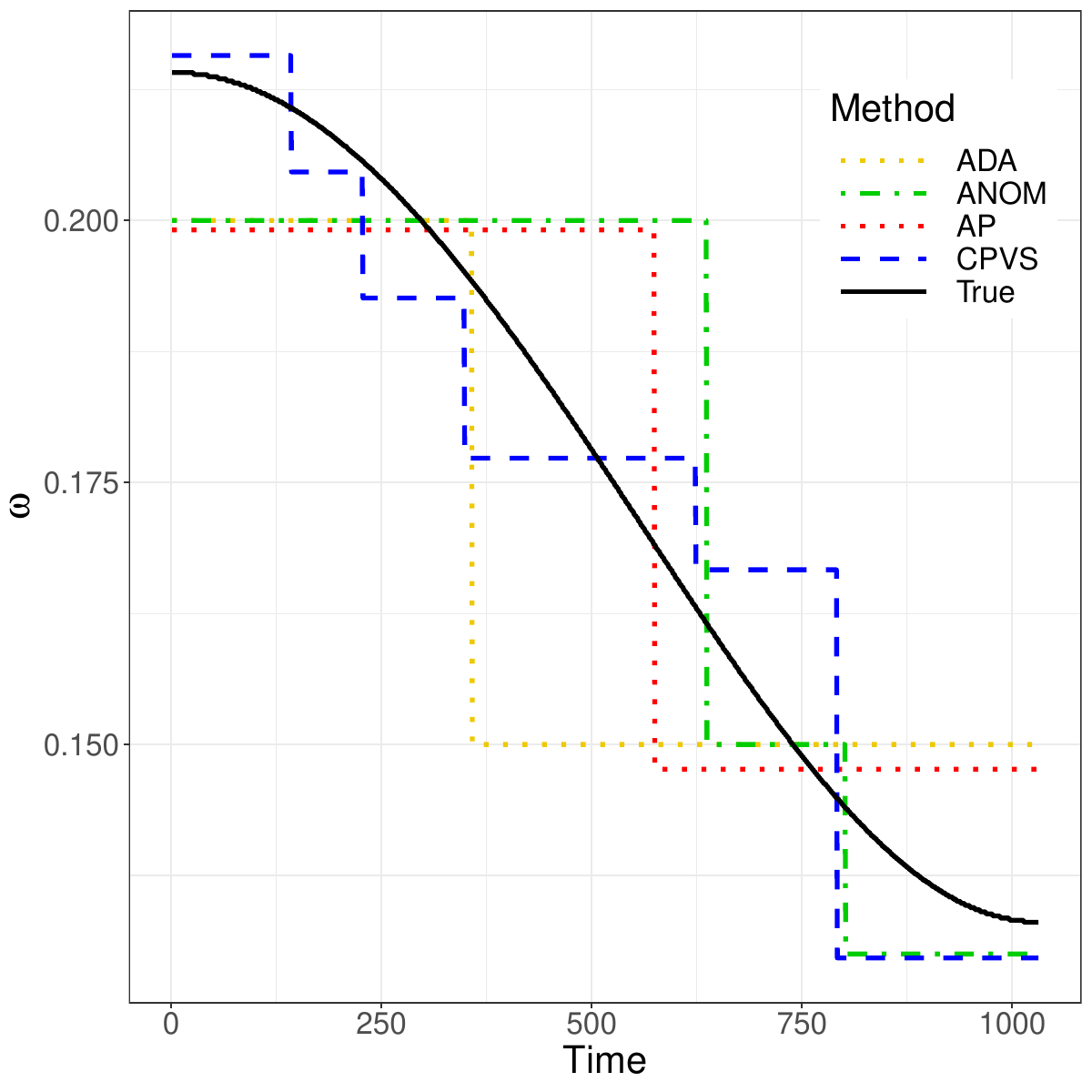}}
\caption{{Scenario 5: time-varying frequency}. (a) Time-varying frequency with the highest power spectrum over replicates (lightblue lines) against the true one (solid black line). (b) Estimated time-varying frequency with the highest power specturm across methods and against the true one for a single replicate.}\label{fig:spectra_slow_auto}
\end{figure}

\clearpage
\section{Applications: additional results}
\label{app:eeg_data}
We complement Section 4 with additional figures that are not included in the main manuscript.

\begin{figure}[!ht]
\centering
\subfigure[Original data $y_t$]{\includegraphics[width=.48\textwidth]{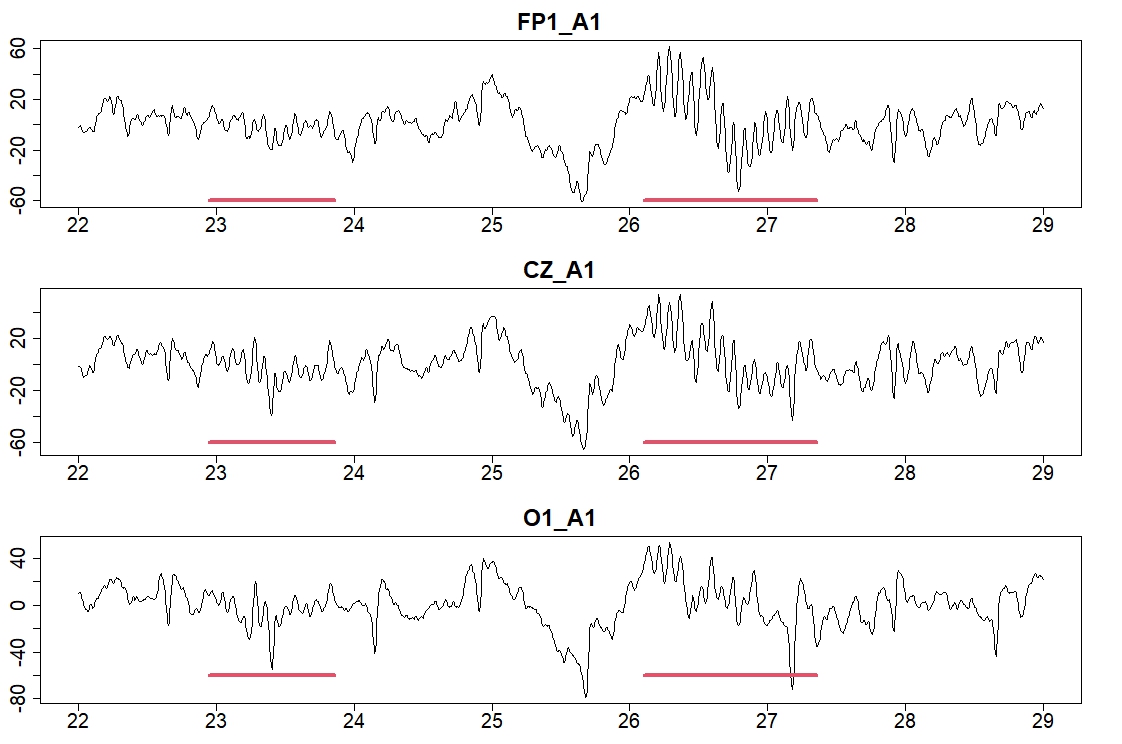}}
\subfigure[First difference $y_{t}-y_{t-1}$]{\includegraphics[width=.48\textwidth]{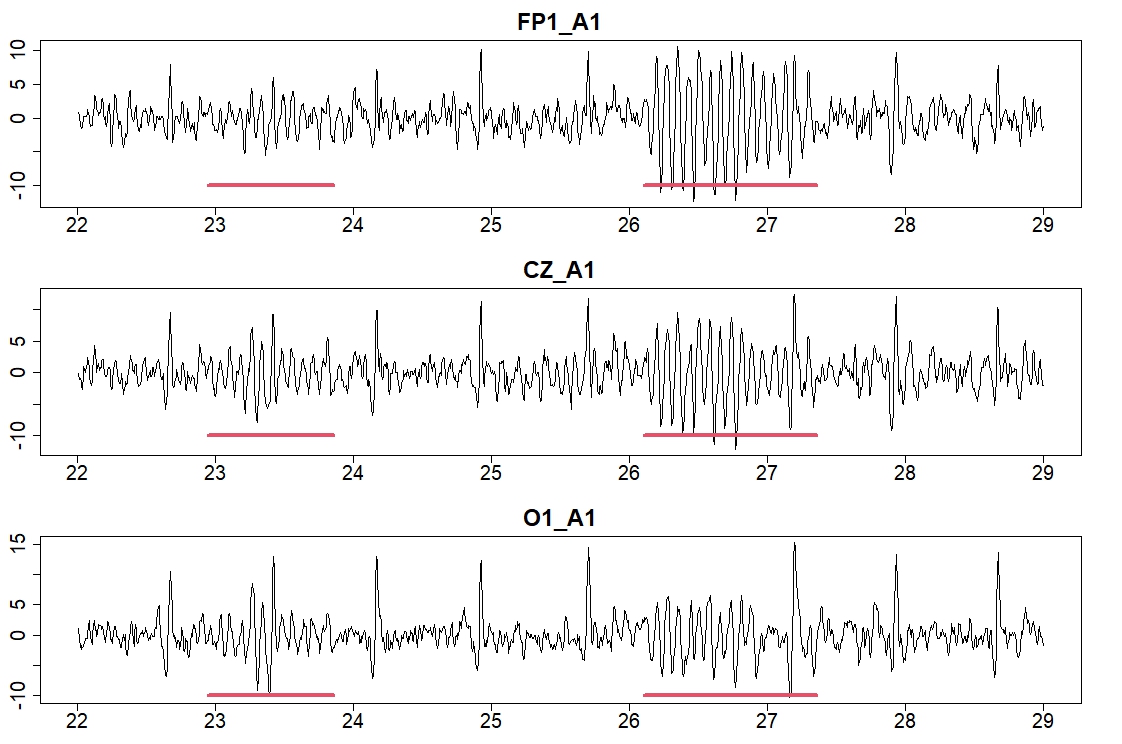}}
\caption{{Multichannel EEG sleep data}. This plot depicts the excerpts used in the analysis. Sleep spindles annotated by the experts are highlighted in red.}\label{fig:sleep_series}
\end{figure}

\begin{figure}[!ht]
\centering
\includegraphics[width=.9\textwidth]{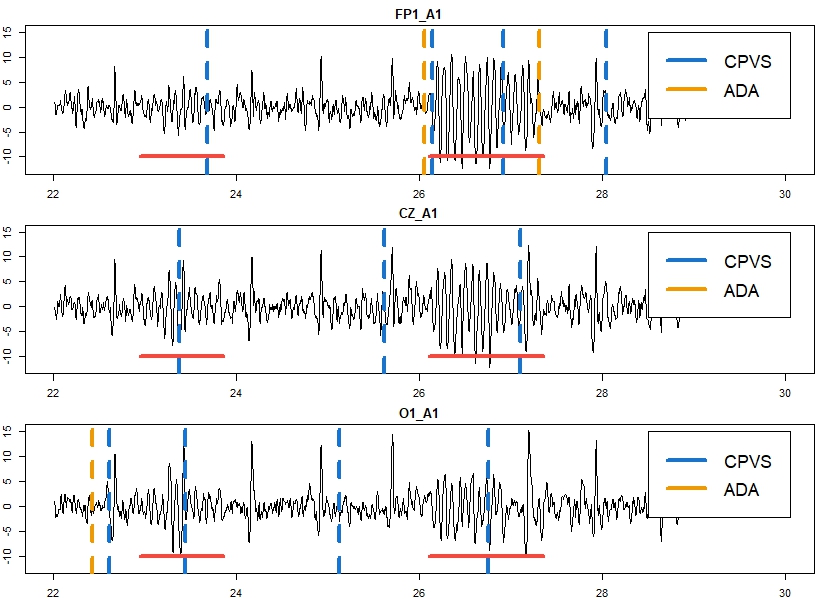}
\caption{{EEG data: change points detected by univariate methods}. The plot depicts the segmentation provided by CPVS(2), and the ADA(7) \citep{rosen2012adaptspec}. Both methods are run separately on each channel series. Sleep spindles annotated by the experts are highlighted in red.}\label{fig:sleep1_cp}
\end{figure}

\begin{figure}[!ht]
\centering
\subfigure[Estimated frequencies and intensities by segment for first channel]{\includegraphics[width=.78\textwidth]{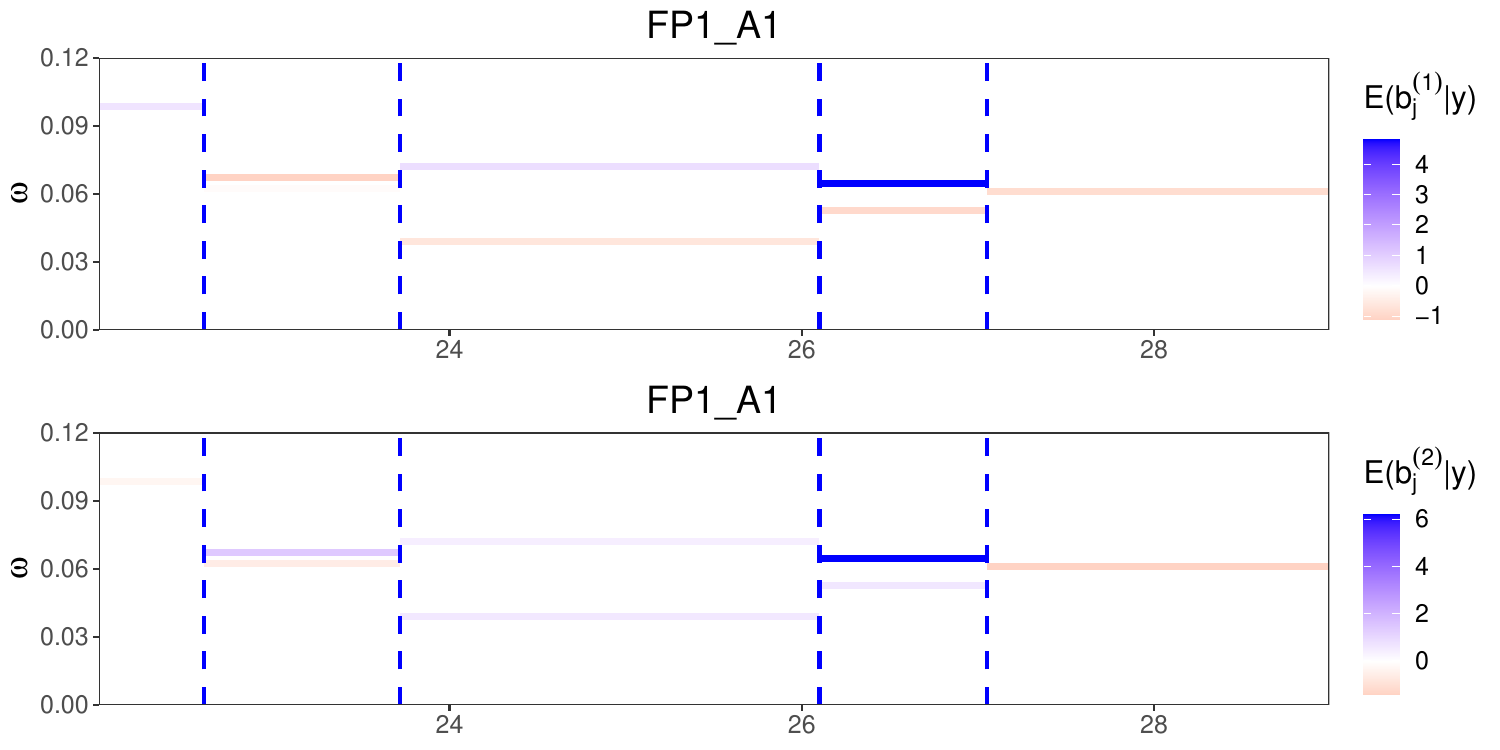}}\\
\subfigure[Segment 1]{\includegraphics[width=.32\textwidth]{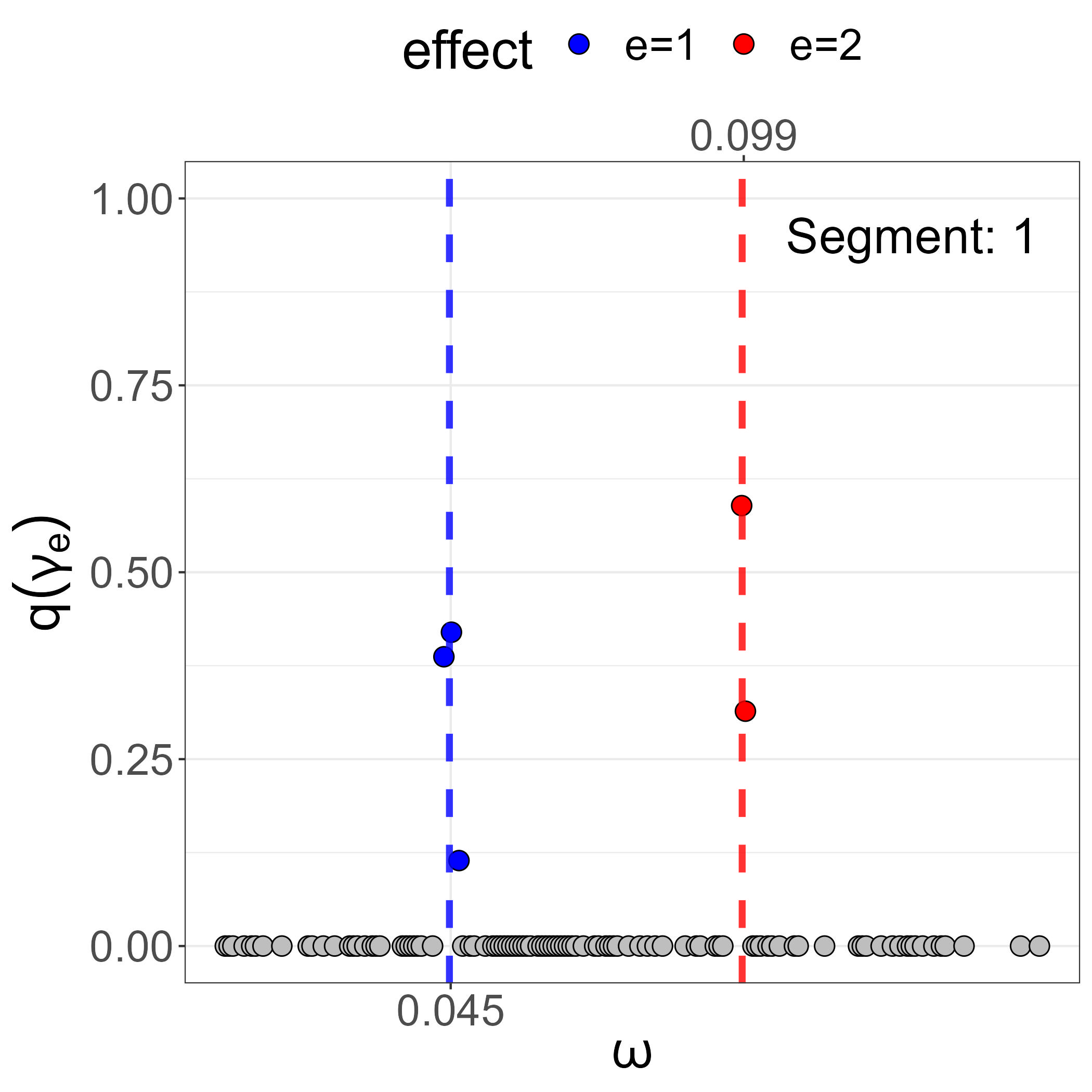}}
\subfigure[Segment 2]{\includegraphics[width=.32\textwidth]{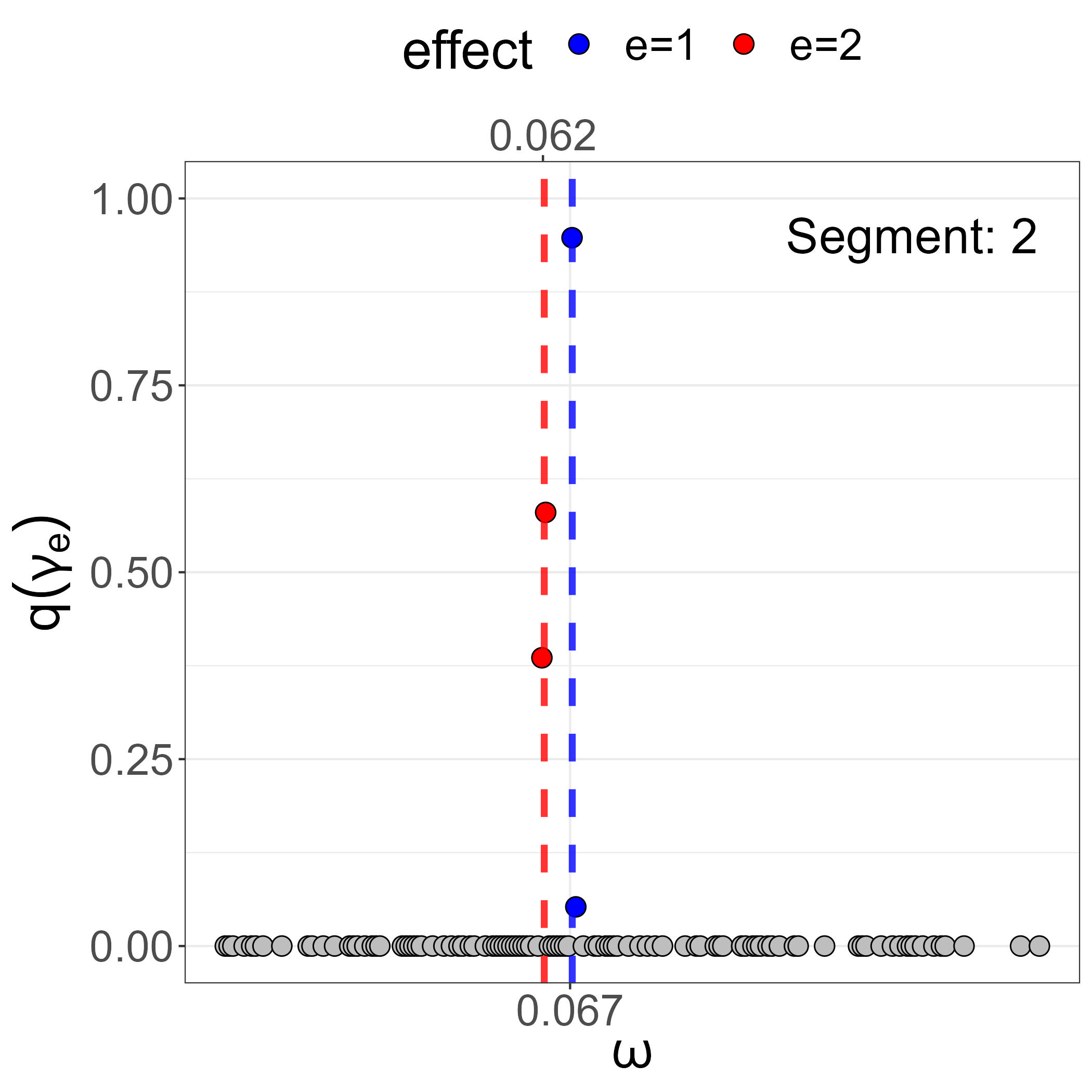}}
\subfigure[Segment 3]{\includegraphics[width=.32\textwidth]{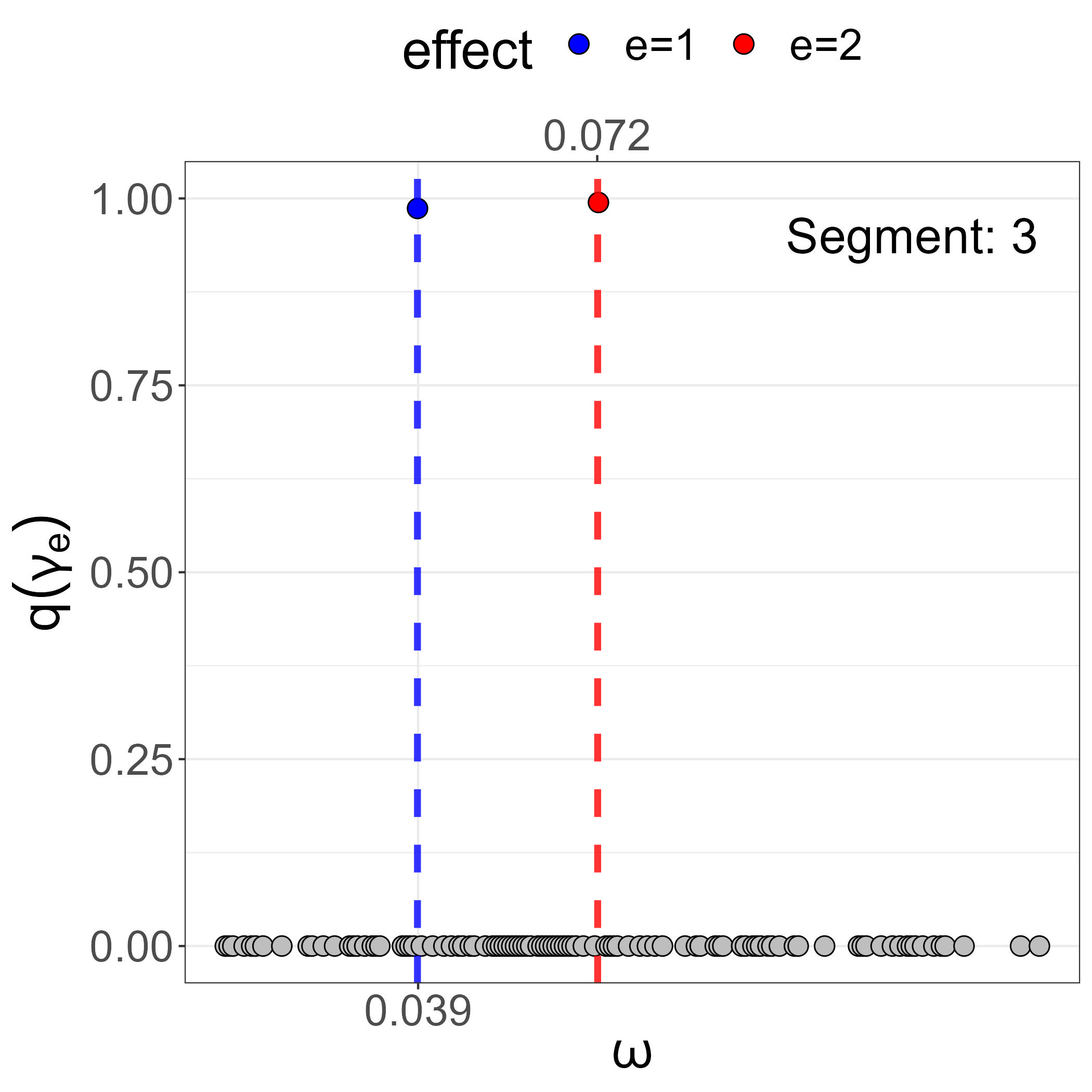}}\\
\subfigure[Segment 4]{\includegraphics[width=.32\textwidth]{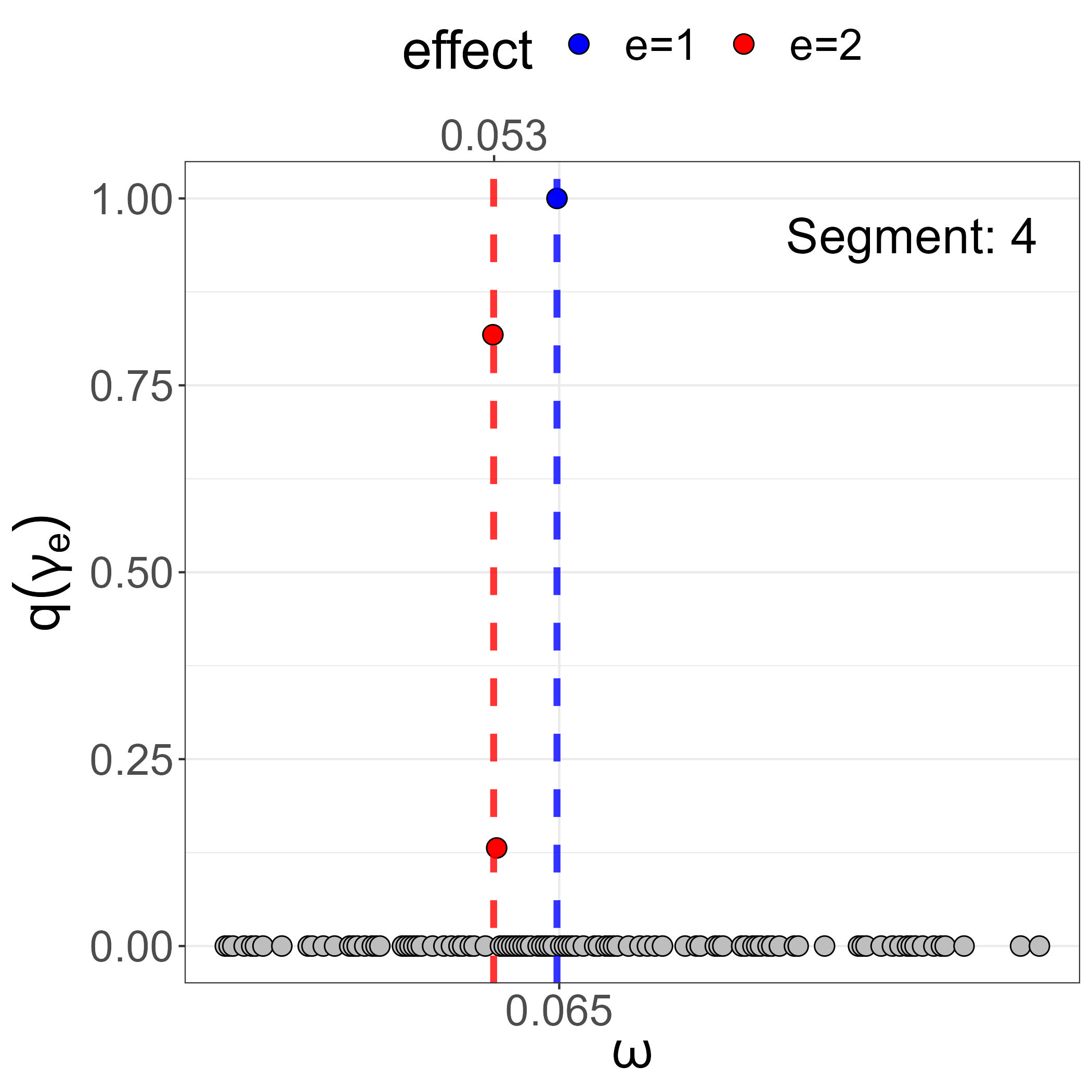}}
\subfigure[Segment 5]{\includegraphics[width=.32\textwidth]{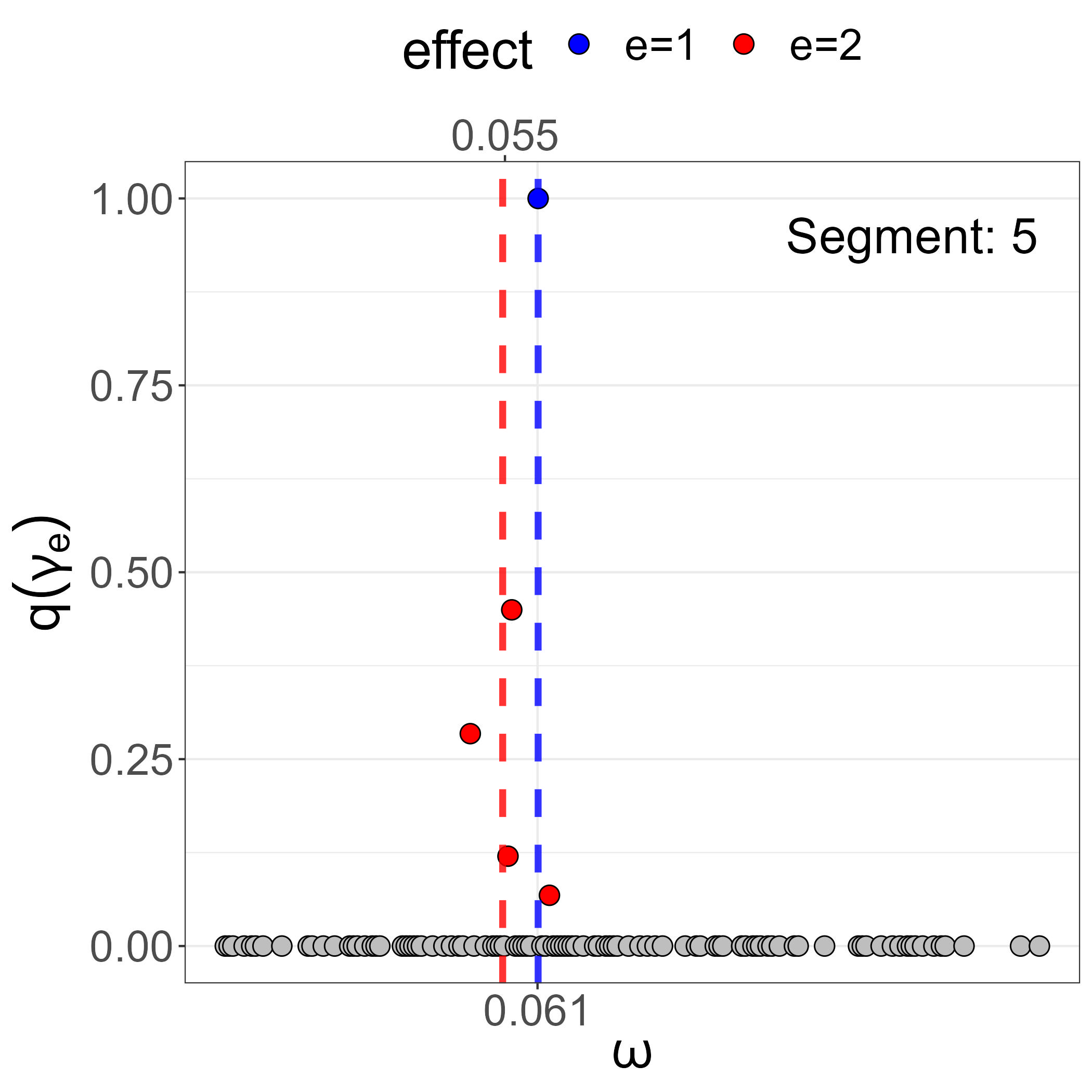}}
\caption{{EEG data: frequencies' uncertainty quantification ($N_E=2$)}. Panel (a) depicts the estimated frequencies and intensities in each segment for the segmentation of the first channel. It serves as a reference for panels (b)-(f). Panels (b)-(f) show the posterior inclusion probabilities for the frequencies and the dashed vertical lines denote the mean of the frequencies weighted by the inclusion probability. Different colours refer to different effects in the SuSiE algorithm.}\label{fig:sleep_omega_uncertainty}
\end{figure}

\begin{figure}[!ht]
\centering
\subfigure[Original data $y_t$]{\includegraphics[width=.48\textwidth]{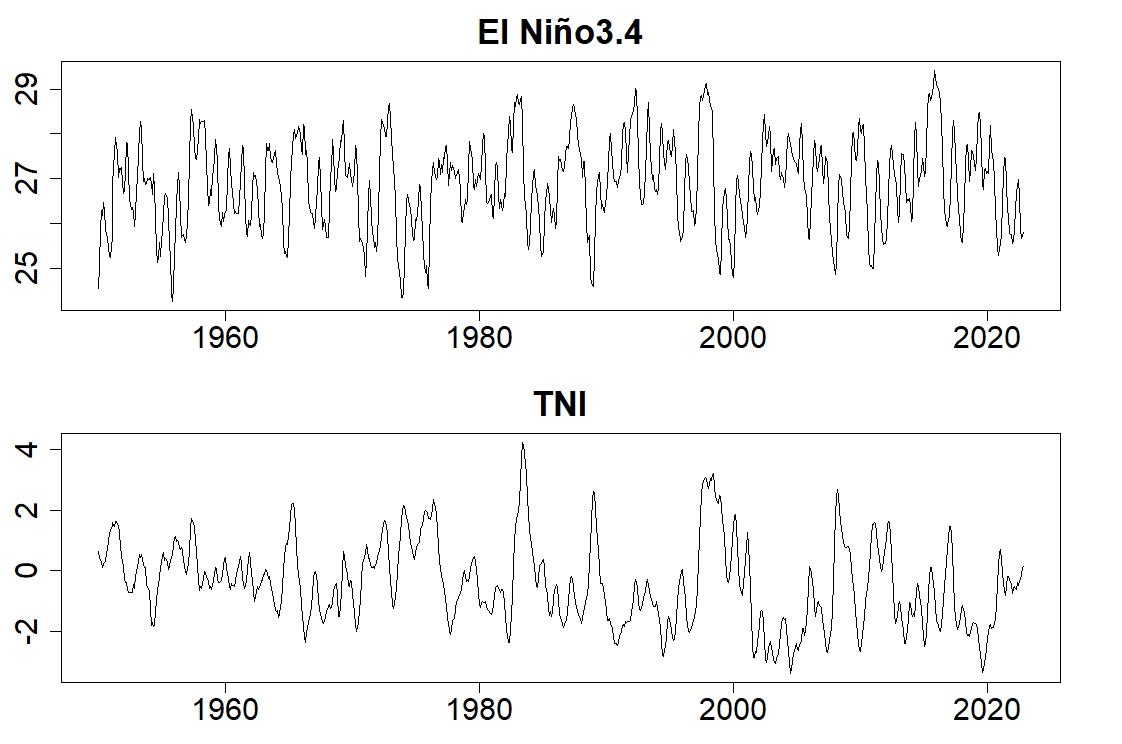}}
\subfigure[First difference $y_{t}-y_{t-1}$]{\includegraphics[width=.48\textwidth]{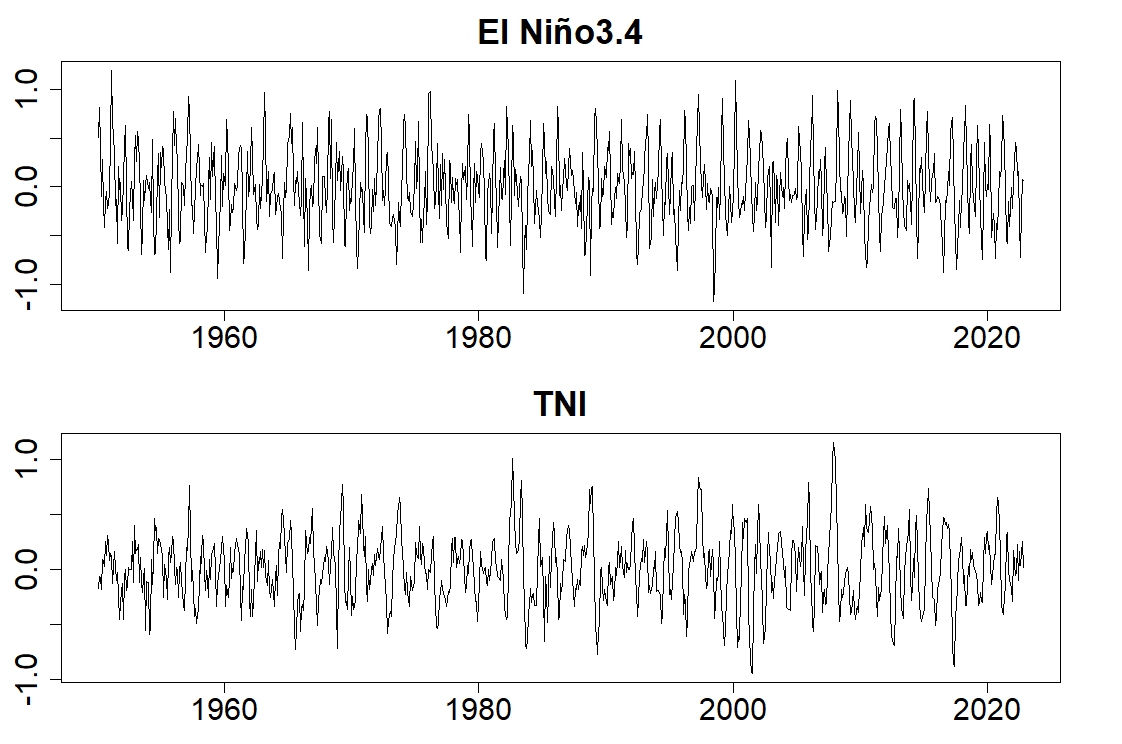}}
\caption{{ENSO data: El Niño3.4 and Trans-Niño indexes}. Plots show the data used in the analysis spanning a period from 1950 to 2022.}\label{fig:elnino_data}
\end{figure}

\begin{figure}[!ht]
\centering
\subfigure[]{\includegraphics[width=.48\textwidth]{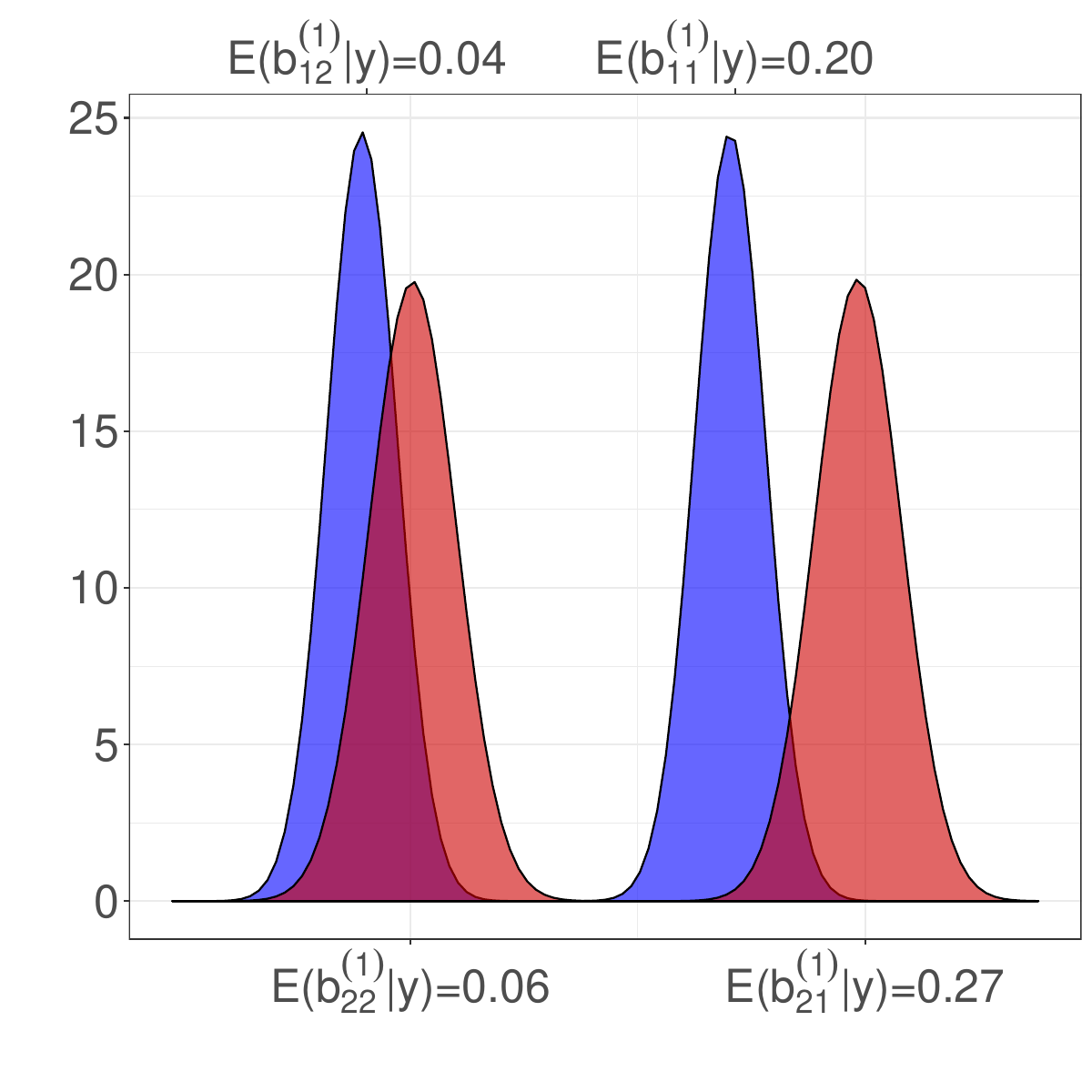}}
\subfigure[]{\includegraphics[width=.48\textwidth]{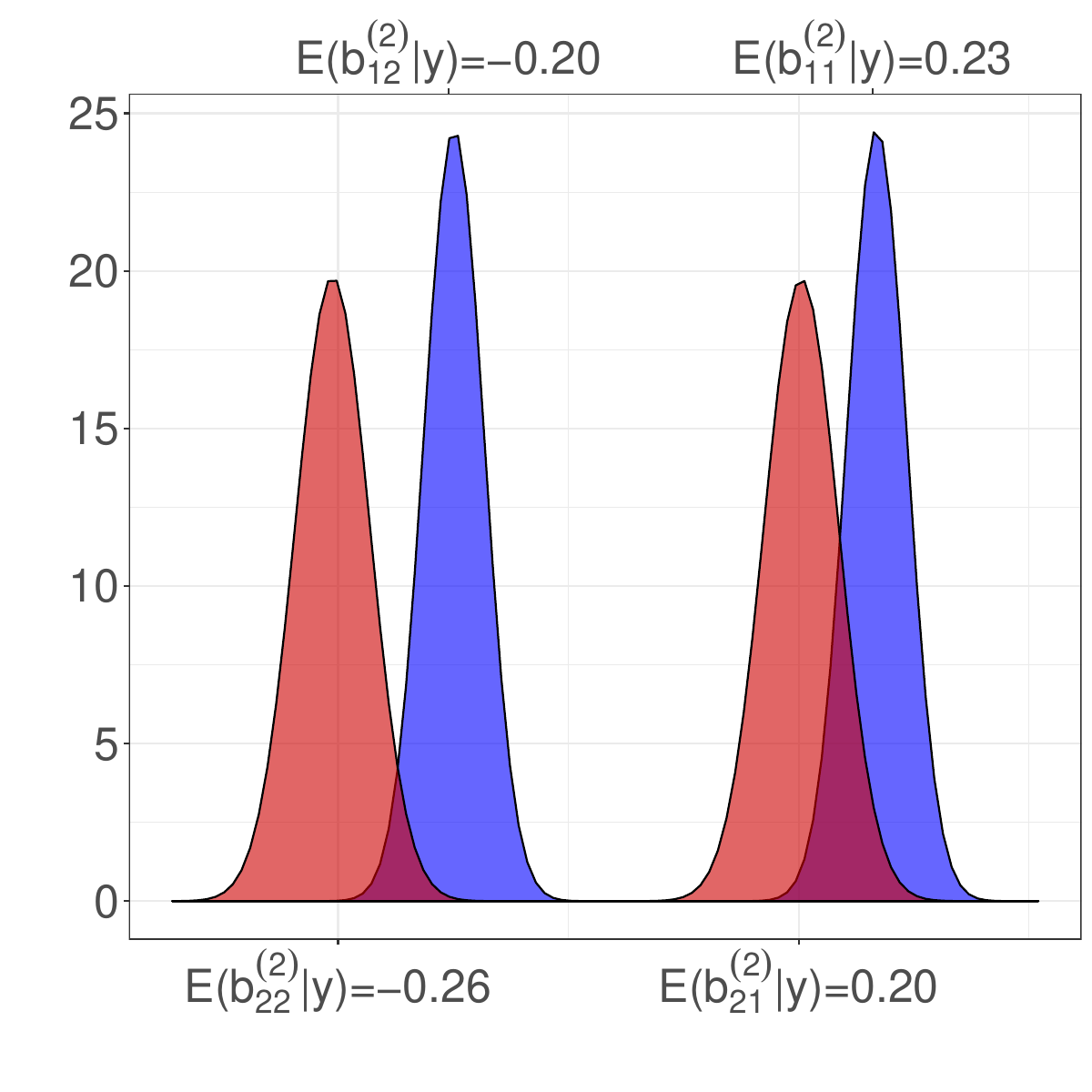}}
\caption{{ENSO data: Posterior distributions of the intensities for $N_E=2$}.  The plot depicts $q(b_{je}^{(1)})$ (a) and $q(b_{je}^{(2)})$ (b) in the two segments $j=1,2$ (blue and red, respectively) and for the two effects in the SuSiE $e=1,2$. The posterior mean is indicated, as a point estimate for the intensity.}\label{fig:elnino_betas2}
\end{figure}

\begin{figure}[!ht]
\centering
\subfigure[Estimated frequencies and intensities by segment for El Niño3.4 index]{\includegraphics[width=.8\textwidth]{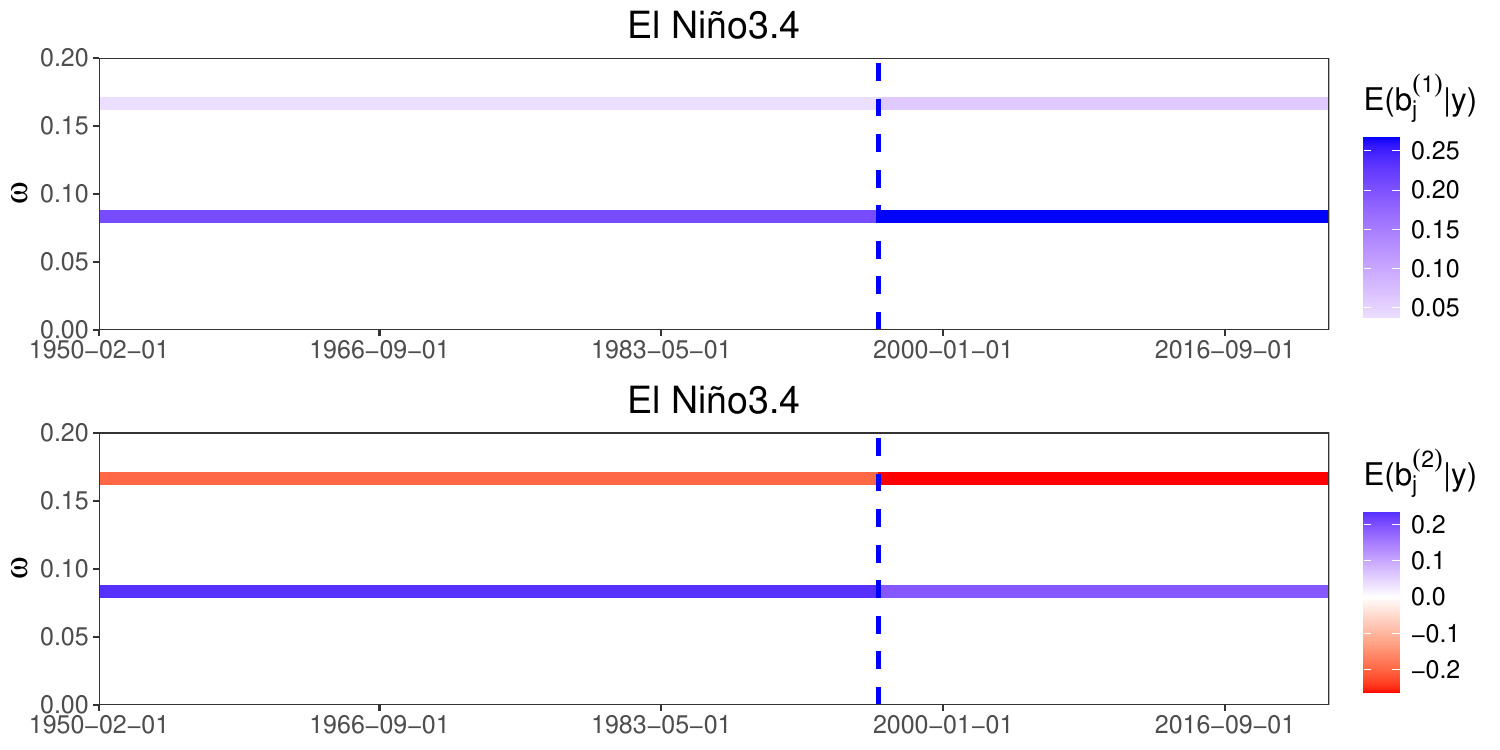}}\\
\subfigure[Segment 1: $N_E=2$]{\includegraphics[width=.4\textwidth]{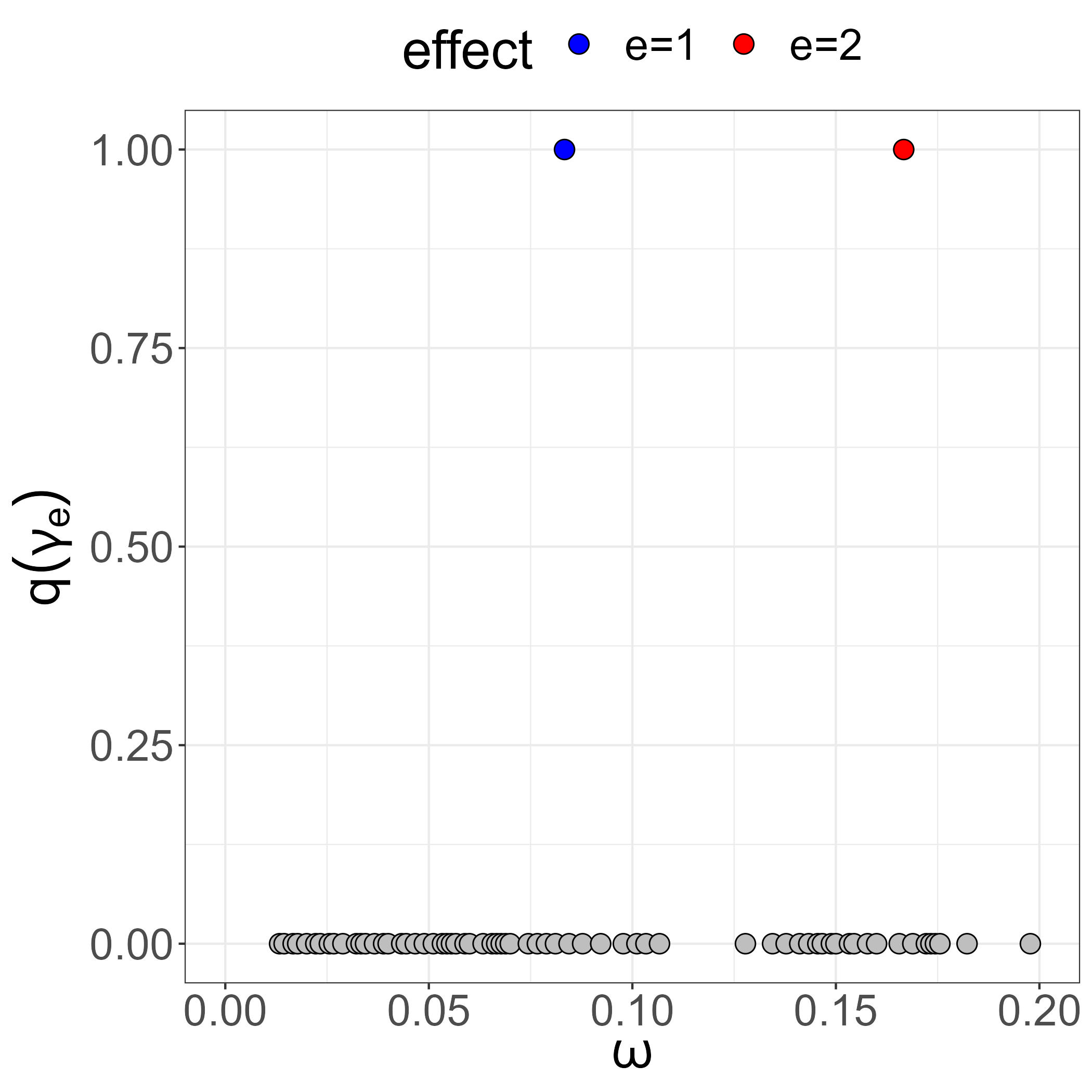}}
\subfigure[Segment 2: $N_E=2$]{\includegraphics[width=.4\textwidth]{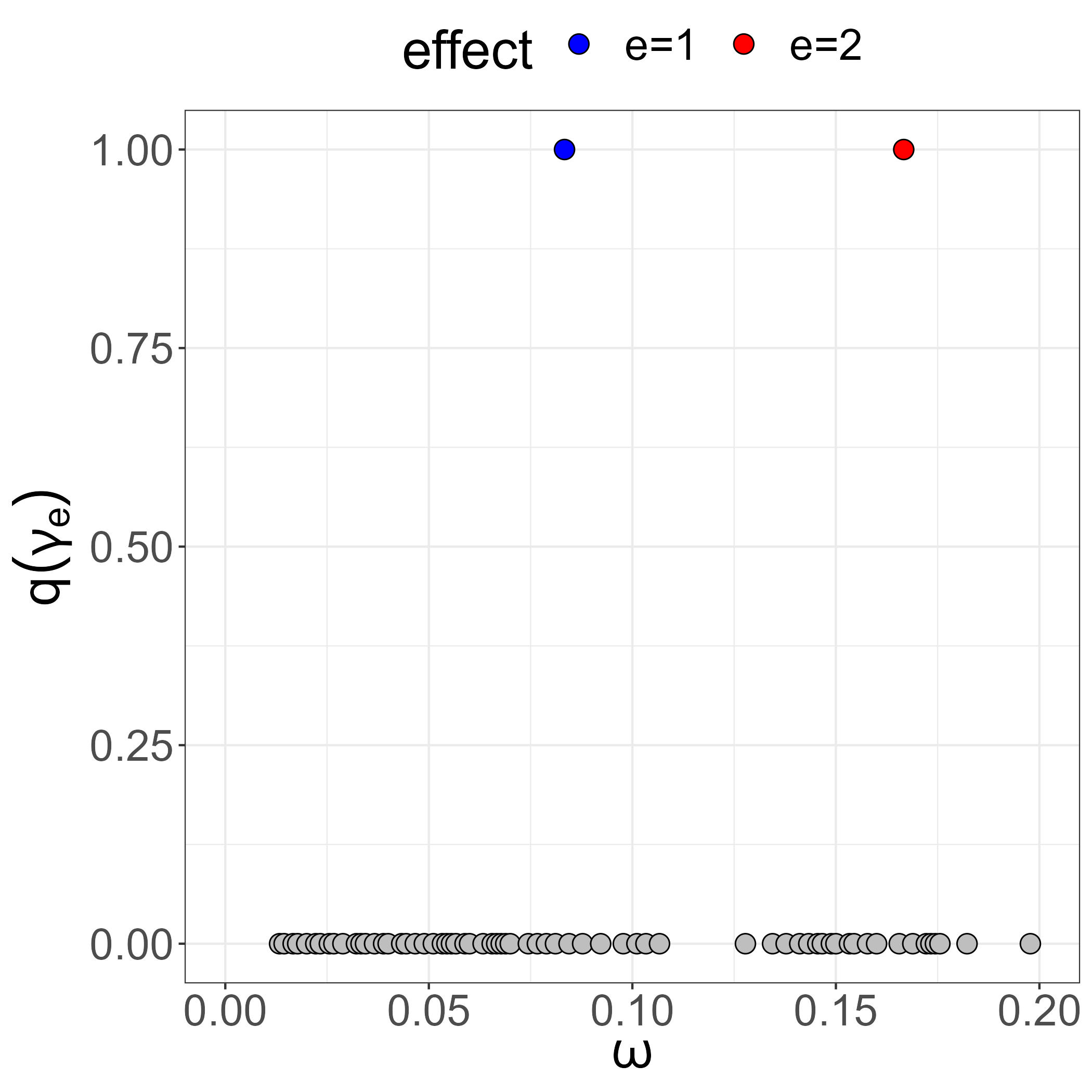}}\\
\caption{{ENSO data: frequencies' uncertainty quantification}. Panel (a) depicts the estimated frequencies and intensities in each segment for El Niño3.4 index. It serves as a reference for panels (b)-(c). Panels (b)-(c) show the posterior inclusion probabilities $q(\gamma_{je})$ for the frequencies in the two segments. Different colours refer to different effects in the SuSiE algorithm with $N_E=2$.}\label{fig:nino_i1_omega_uncertainty}
\end{figure}

\begin{figure}[!ht]
\centering
\subfigure[Estimated frequencies and intensities by segment for TNI]{\includegraphics[width=.8\textwidth]{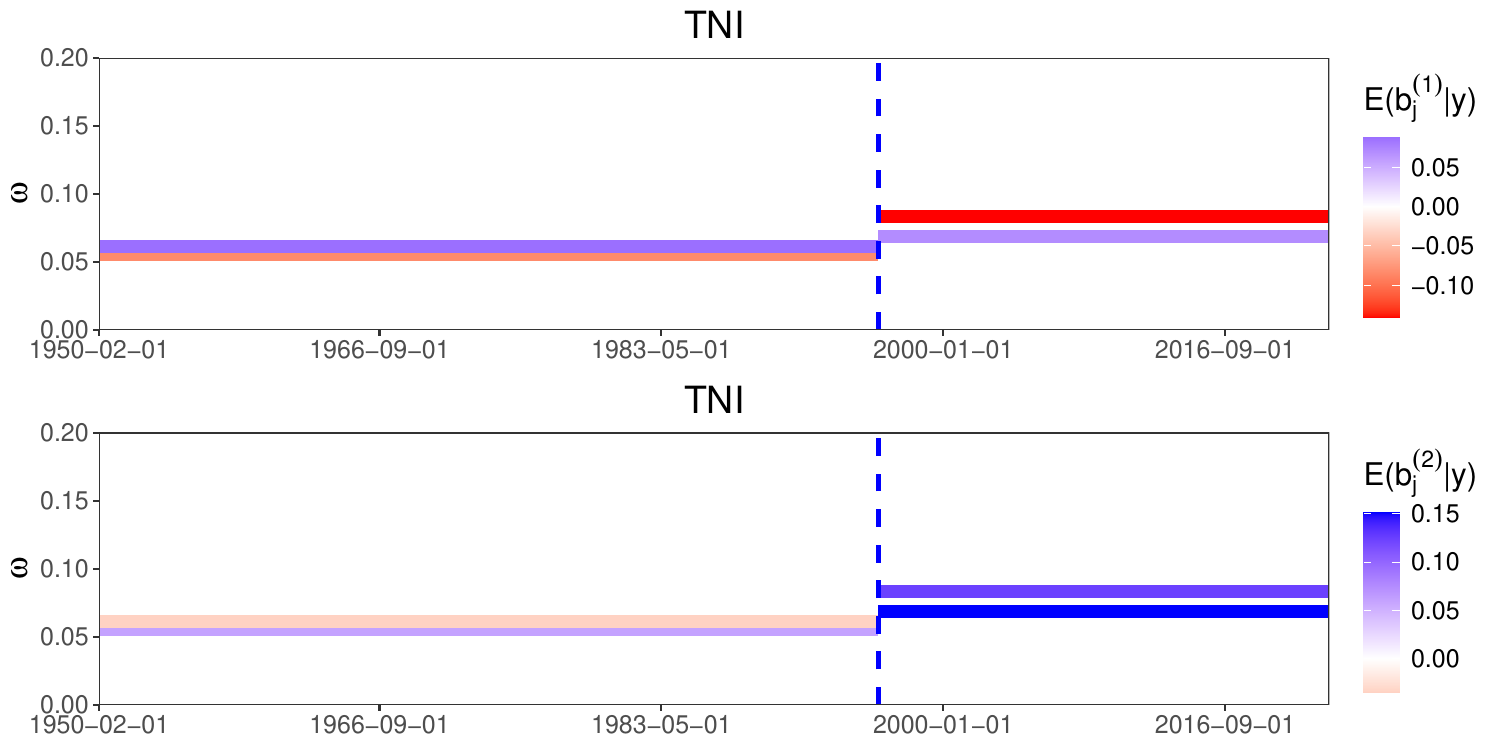}}\\
\subfigure[Segment 1: $N_E=2$]{\includegraphics[width=.4\textwidth]{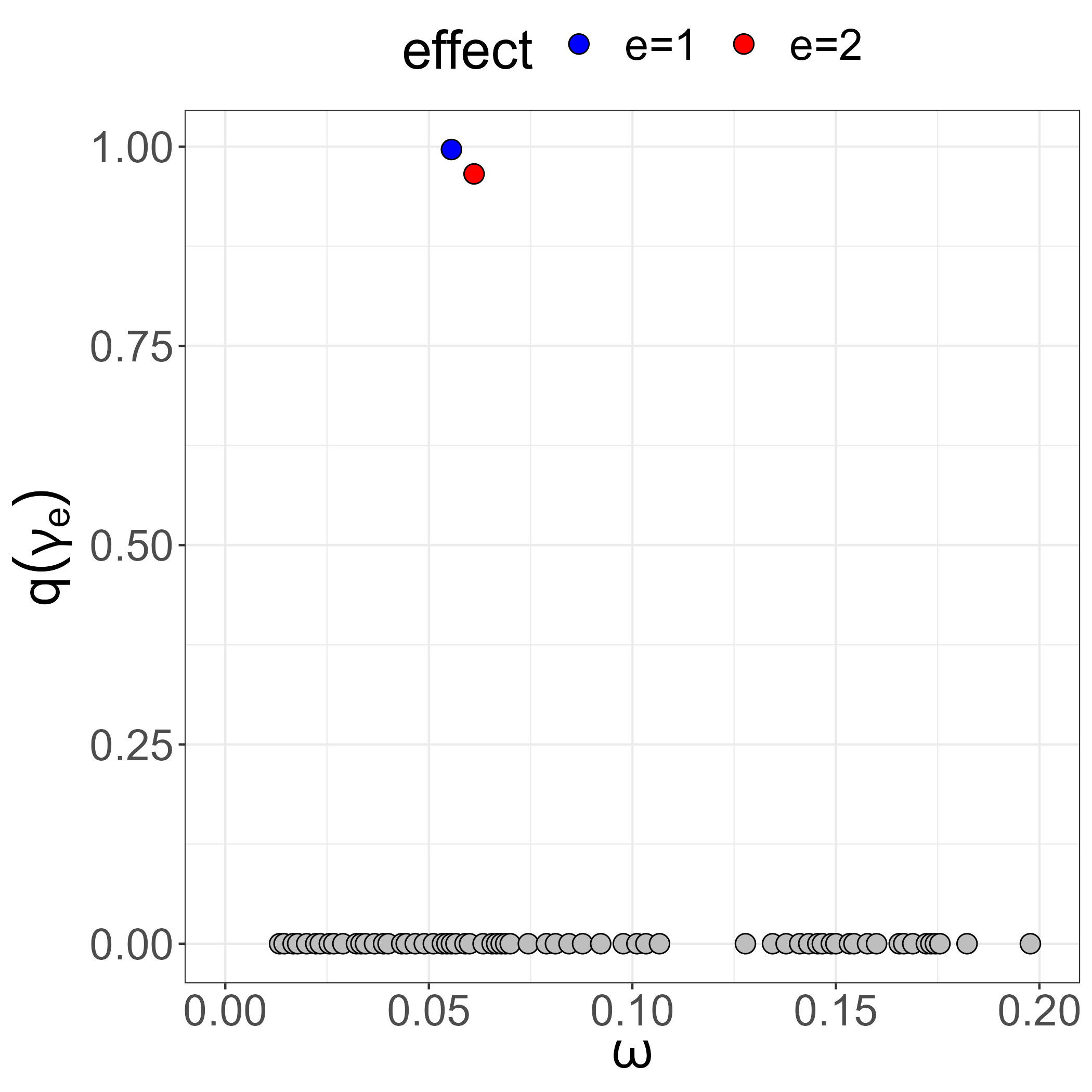}}
\subfigure[Segment 2: $N_E=2$]{\includegraphics[width=.4\textwidth]{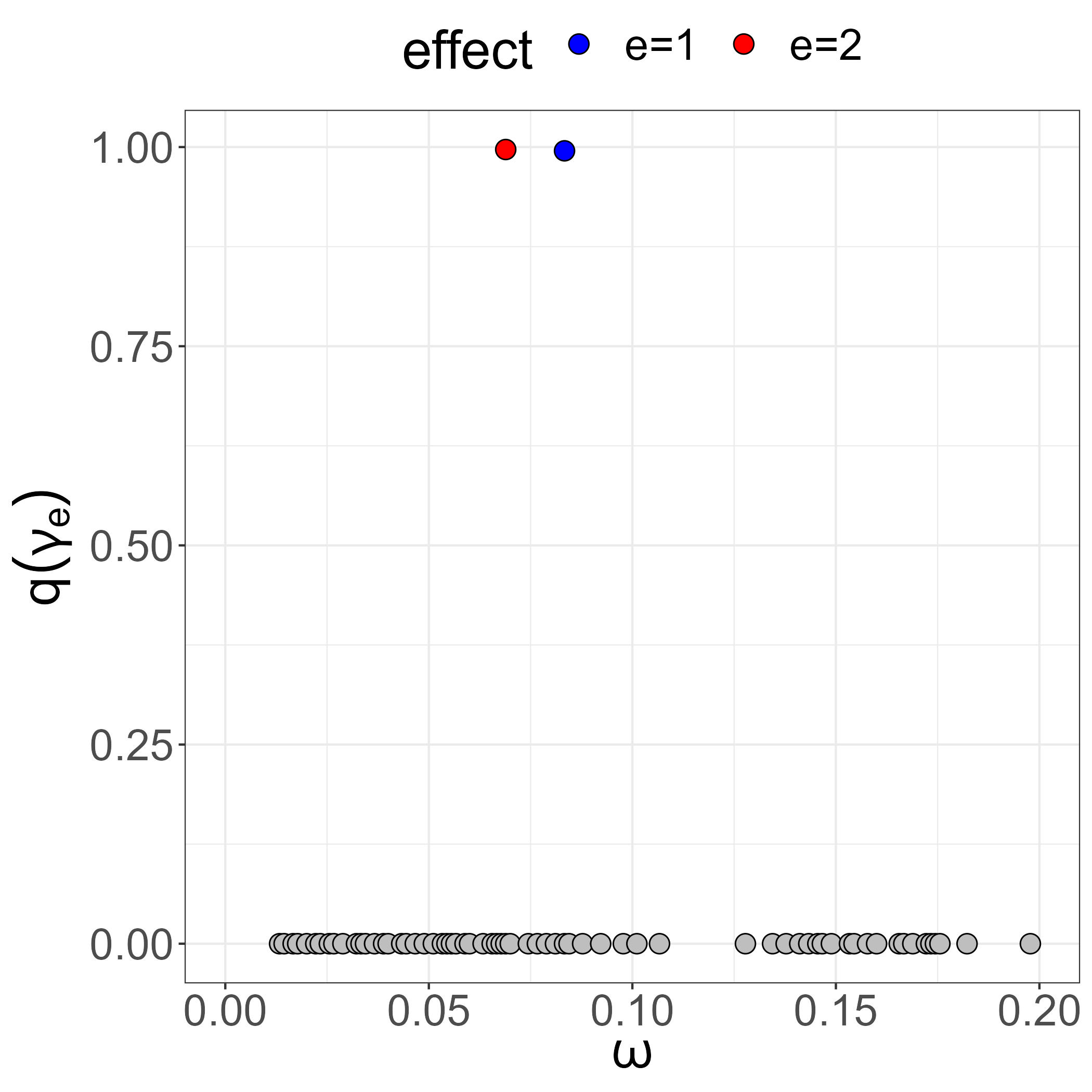}}
\caption{{ENSO data: frequencies' uncertainty quantification}. Panel (a) depicts the estimated frequencies and intensities in each segment for TNI. It serves as a reference for panels (b)-(c). Panels (b)-(c) show the posterior inclusion probabilities for the frequencies in the two segments. Different colours refer to different effects in the SuSiE algorithm with $N_E=2$.}\label{fig:nino_i2_omega_uncertainty}
\end{figure}

\end{document}